\documentclass[12pt]{article}

\usepackage[colorlinks, linkcolor=blue, citecolor=blue]{hyperref}           
\usepackage{color,mathrsfs}
\usepackage{graphicx,subfigure,amsmath,amssymb,amsfonts,bm,epsfig,epsf,url,dsfont}
\usepackage{amsthm}
\usepackage{tikz}
\usepackage{bbm}      
\usepackage{booktabs}
\usepackage{cases}
\usepackage{fullpage}
\usepackage[small,bf]{caption}
\usepackage[top=1in,bottom=1in,left=1in,right=1in]{geometry}
\usepackage{fancybox}
\usepackage{algorithm}
\usepackage{verbatim}
\usepackage{algorithmic}
\usepackage{cancel}
\usepackage{multirow}

\usepackage{mathtools}

\usepackage{titlesec}
\titlespacing*{\paragraph}
{0pt}      % left margin
{1ex}    % space before
{1em}      % space after (before text)

\usepackage{scalerel,stackengine}
\stackMath
\newcommand\reallywidehat[1]{%
\savestack{\tmpbox}{\stretchto{%
  \scaleto{%
    \scalerel*[\widthof{\ensuremath{#1}}]{\kern-.6pt\bigwedge\kern-.6pt}%
    {\rule[-\textheight/2]{1ex}{\textheight}}%WIDTH-LIMITED BIG WEDGE
  }{\textheight}% 
}{0.5ex}}%
\stackon[1pt]{#1}{\tmpbox}%
}

\newcommand{\bE}{\mathbb{E}}

\newcommand{\G}{\mathbb{G}}

\newcommand{\TV}{\mathsf{TV}}

\newtheorem{definition}{Definition}
\newtheorem{observation}{Observation}

\newtheorem{remark}{remark}
\newtheorem{theorem}{Theorem}

\newtheorem{conjecture}{Conjecture}
\newtheorem{lemma}{Lemma}
\newtheorem{prop}{Proposition}

\newtheorem{rmk}{Remark}

\newenvironment{fminipage}%
  {\begin{Sbox}\begin{minipage}}%
  {\end{minipage}\end{Sbox}\fbox{\TheSbox}}

\makeatletter
\newcommand*{\rom}[1]{\expandafter\@slowromancap\romannumeral #1@}
\makeatother

\newcommand{\Ind}{\mathds{1}}

\newcommand{\abs}[1]{\left|#1\right|}
\newcommand{\R}{\mathbb{R}}

\newcommand{\Q}{\mathbb{Q}}

\newcommand{\E}{\mathbb{E}}

\def\P{{\mathbb P}}
\def\S{{\mathbb S}}

\newcommand {\bs} {\mbox{\boldmath $s$}}

\newcommand {\pr} {\mathbb{P}}

\newcommand{\calA}{{\cal A}}

\newcommand{\calU}{{\cal U}}
\newcommand{\calD}{{\cal D}}

\newcommand{\calE}{{\cal E}}
\newcommand{\calF}{{\cal F}}
\newcommand{\calG}{{\cal G}}
\newcommand{\calH}{{\cal H}}

\newcommand{\calK}{{\cal K}}
\newcommand{\calL}{{\cal L}}

\newcommand{\calN}{{\cal N}}

\newcommand{\calP}{{\cal P}}
\newcommand{\calQ}{{\cal Q}}

\newcommand{\calS}{{\cal S}}
\newcommand{\calT}{{\cal T}}
\newcommand{\calV}{{\cal V}}

\newcommand{\calX}{{\cal X}}

\newcommand{\Tr}{{\mathsf{Tr}}}

\newcommand{\be}{\begin{equation}}
\newcommand{\ee}{\end{equation}}
\newcommand{\beqna}{\begin{eqnarray}}
\newcommand{\eeqna}{\end{eqnarray}}

\newcommand\independent{\protect\mathpalette{\protect\independenT}{\perp}}
\def\independenT#1#2{\mathrel{\rlap{$#1#2$}\mkern2mu{#1#2}}}
\renewcommand{\bs}[1]{\mathbf{#1}}

\makeatletter
\newcommand{\vast}{\bBigg@{4}}
\newcommand{\Vast}{\bBigg@{5}}
\makeatother

\newcommand{\p}[1]{\left(#1\right)}
\newcommand{\pp}[1]{\left[#1\right]}
\newcommand{\ppp}[1]{\left\{#1\right\}}
\newcommand{\norm}[1]{\left\|#1\right\|}
\newcommand{\innerP}[1]{\left\langle#1\right\rangle}

\usepackage{xspace}

\newcommand{\s}[1]{\mathsf{#1}}

\renewcommand*{\bs}[1]{\mathbf{#1}}

\usepackage{tikz}
\usetikzlibrary{shapes,arrows}
\usepackage{verbatim}

\tikzstyle{block} = [draw, fill=blue!20, rectangle, 
    minimum height=3em, minimum width=6em]
\tikzstyle{sum} = [draw, fill=blue!20, circle, node distance=1cm]
\tikzstyle{input} = [coordinate]
\tikzstyle{output} = [coordinate]
\tikzstyle{pinstyle} = [pin edge={to-,thin,black}]

\makeatletter
\def\thanks#1{\protected@xdef\@thanks{\@thanks
        \protect\footnotetext{#1}}}
\makeatother

\allowdisplaybreaks
\begin{document}

%\title{Phase Transitions in the Detection of a Hidden Geometry in Random Graphs}
\title{Testing for a Hidden Geometry in Random Graphs}
%\title{Testing for a Hidden High-Dimensional Geometry in Random Graphs}

\author{Amit~Silber~~~~~~~~~Mor~Oren-Loberman~~~~~~~~~Wasim Huleihel\thanks{A. Silber, M. Oren-Loberman , and W. Huleihel are with the School of Electrical \& Computer Engineering at Tel Aviv University, {T}el {A}viv 6997801, Israel (e-mails:  \texttt{amitsilber@mail.tau.ac.il, orenmor@mail.tau.ac.il, wasimh@tauex.tau.ac.il}). This work is supported by the ISRAEL SCIENCE FOUNDATION (grant No. 1734/21).}}

\date{\today}

\maketitle

\begin{abstract}

%In this paper, we study the problem of detecting a small hidden geometric structure in a random graph. This is formulated as a hypothesis testing problem, where under the null hypothesis, the observed graph is a realization of an Erd\H{o}s--R\'{e}nyi random graph $\calG(n,q)$, with edge density $q\in(0,1)$; under the alternative a high-dimensional random geometric graph $\calG(k,p,d)$ on $k\leq n$ vertices is planted in $\calG(n,q)$, where each of the $k$ vertices corresponds to an independent random point distributed uniformly on the sphere $\mathbb{S}^{d-1}$, and two vertices are connected if the corresponding latent vectors are close enough, such that the probability of an edge is $p$. For this problem, we investigate the fundamental limits from both the statistical and computational perspectives. We establish a lower bound for the detection of the latent geometry, expressed in terms of the above parameters. Furthermore, we propose and analyze several algorithms that achieve this lower bound. We also explore the problem of testing in polynomial-time. Similarly to many structured high-dimensional problems, our model exhibits an ``easy-hard-impossible" phase transition, where computational limitations can significantly impact statistical performance. As evidence of this phenomenon, we demonstrate that the class of low-degree polynomial algorithms fail in our conjecturally hard regime.

In this work, we investigate the fundamental problem of detecting a faint geometric signal hidden within an otherwise random graph. We formulate this task as a hypothesis testing problem: under the null hypothesis, the observed graph is an Erd\H{o}s--R\'enyi random graph $\calG(n,q)$ with edge density $q\in(0,1)$; under the alternative, a high-dimensional geometric structure is clandestinely embedded. Specifically, a random geometric graph $\calG(k,q,d)$ on $k\le n$ vertices is planted inside $\calG(n,q)$, where each of the $k$ vertices corresponds to an independent random point drawn uniformly from the unit sphere $\mathbb{S}^{d-1}$, and edges are formed according to latent proximity, resulting in the same edge probability $q$.

Our objective is to characterize the limits of detectability of this hidden geometry, from both statistical and computational perspectives. We derive sharp information-theoretic lower bounds that characterize the regimes in which detection is fundamentally impossible, expressed explicitly in terms of the problem parameters. Complementing these impossibility results, we propose and analyze several algorithms that provably attain these limits whenever detection is feasible.

We also explore the algorithmic landscape of the problem and investigate which regimes admit efficient, polynomial-time testing procedures. As in many other structured high-dimensional inference problems, our model exhibits a pronounced \emph{easy--hard--impossible} phase transition: there exist regimes in which detection is statistically possible yet computationally prohibitive, as well as regimes in which detection is impossible even with unbounded computational resources. As concrete evidence of this computational barrier, we show that the entire class of low-degree polynomial algorithms fails in the conjecturally hard regime, highlighting a sharp separation between statistical possibility and algorithmic feasibility.

\end{abstract}

\newpage

\tableofcontents

\section{Introduction}\label{sec:intro}

Networks with latent structure arise throughout modern data science, from social and biological systems to communication and information networks. A large body of work models such data using random graphs endowed with an underlying geometric or feature-based structure, where vertices correspond to latent points in a metric space and edges form preferentially between nearby points. Canonical examples include random geometric graphs, random dot product graphs, and more general latent space models; see, for example, the monograph \cite{Mathew07} and the references therein. These models provide a principled way to capture dependencies and correlations that are absent from edge-independent baselines such as the Erd\H{o}s--R\'enyi random graph.

In many applications, however, the latent variables themselves are unobserved, and only the graph structure is available. This leads to a fundamental statistical question: given a single observed graph, can one determine whether its edges encode an underlying geometric structure, or whether they are effectively indistinguishable from those of a purely random graph? From a probabilistic perspective, this question has been studied in the high-dimensional setting, where geometry may be lost as the ambient dimension grows. Early work showed that classical random geometric graphs become asymptotically indistinguishable from Erd\H{o}s--R\'enyi graphs as the dimension tends to infinity~\cite{devroye2011high}. Subsequent results identified a sharp phase transition in the dense regime, revealing that high-dimensional geometric structure disappears precisely when the dimension exceeds a cubic threshold in the number of vertices~\cite{bubeck2016testing}. More recent extensions have explored how noise and softened geometric dependence affect this transition; see, for instance,~\cite{liu2023phase,liu2023probabilistic}.

While these results address the detectability of \emph{global} geometric structure, many problems of interest involve \emph{localized} signals that are present only on a small subset of vertices. Detecting such faint, structured signals hidden inside high-dimensional noise is a recurring theme across modern statistics, probability, and theoretical computer science. In network data, this theme manifests through \emph{planted subgraph} models, where one observes a random graph drawn under one of two hypotheses: either a featureless null model, or an alternative in which a small subset of vertices carries additional structure. These models have played a central role in clarifying the statistical limits of detection and in revealing striking computational--statistical gaps, where information-theoretic possibility does not coincide with what is achievable by efficient algorithms; see, for example, the dense-subgraph and community-detection literature~\cite{arias2014community,verzelen2015community}, as well as recent complexity frameworks, see, e.g., \cite{brennan18a,Hopkins18,brennan20a}.

\subsection{A planted geometric structure} 
In this paper, we study the problem of detecting a small hidden \emph{geometric} structure embedded in a random graph. We formulate this task as a hypothesis testing problem. Under the null hypothesis, the observed graph is an Erd\H{o}s--R\'enyi random graph $\calG(n,q)$ with edge density $q\in(0,1)$. Under the alternative, a high-dimensional random geometric graph $\calG(k,q,d)$ on $k\le n$ vertices is planted inside $\calG(n,q)$. Each of the $k$ planted vertices corresponds to an independent random point drawn uniformly from the unit sphere $\mathbb{S}^{d-1}$, and edges within the planted subgraph arise from latent proximity, resulting in \emph{the same}\footnote{Allowing a distinct edge probability inside the planted set introduces an explicit density contrast, in which case detection is governed primarily by classical planted dense subgraph considerations and the ambient dimension $d$ plays no essential role in the dense setting. We briefly discuss this regime later for comparison.} marginal edge probability $q$. All remaining edges behave as in the null model.

This formulation captures a setting in which geometry is both \emph{localized}—only a small subset of vertices participates in the geometric structure—and \emph{high-dimensional}. The latter aspect is particularly important: as the ambient dimension grows, geometric information becomes increasingly diffuse, and it is known that high-dimensional geometric graphs may become indistinguishable from Erd\H{o}s--R\'enyi graphs. Understanding where this loss of geometry occurs, and how it interacts with the size of the planted
subgraph, is central to our investigation.

From a modeling perspective, the proposed alternative can be viewed as a geometry-enriched analogue of classical planted dense-subgraph or community models. Rather than postulating an \emph{ad hoc} increase in edge probability on a hidden vertex set, the planted structure here is generated by a latent geometric mechanism that induces correlations among edges. In particular, the planted subgraph is not merely ``denser'' than the background—it carries a coherent geometric footprint that couples many local statistics. Specifically, real-world networks often exhibit hidden organization, such as communities or coordinated groups, that must be recovered from noisy data. Classical models typically assume that such structure is revealed through simple signals like increased edge density or rigid patterns. While mathematically convenient, these assumptions often do not reflect how interactions arise in modern settings.

In many applications, connectivity is driven instead by latent features: each node is associated with a high-dimensional representation, and edges form according to relationships in this hidden space. As a result, the defining property of a subset is not higher connectivity, but consistent interaction patterns induced by an underlying geometry. This perspective is especially relevant when nodes are indistinguishable at the level of basic statistics, such as coordinated accounts blending in with ordinary users or firms engaged in collusion. In these cases, the subgraph of interest matches the background in marginal edge probabilities, and can only be identified through higher-order dependencies shaped by the latent structure. We focus on this regime by removing any density-based signal, so that detection must rely entirely on geometric correlations. \emph{To our knowledge, our work is the first to provide a sharp detection theory for a planted structure whose defining feature is latent geometry.}

Beyond statistical detectability, our study is also motivated by computational considerations. Many planted-structure problems in random graphs exhibit pronounced gaps between what is information-theoretically possible and what can be achieved by efficient algorithms. A key objective of this work is to determine whether such computational--statistical gaps arise in the presence of latent geometry, and to characterize the resulting phase transitions between easy, hard, and impossible regimes.

\subsection{Main contributions} We investigate the fundamental limits of detecting latent geometric structure from both statistical and computational perspectives. Our first contribution is an information-theoretic characterization of detectability in terms of the parameters $(n,k,d)$. A central technical challenge in this setting is that naive second-moment methods for the likelihood ratio can be dominated by rare tail events and therefore fail to capture the true statistical threshold. To address this issue, we develop a truncated second-moment analysis, inspired by classical ideas in minimax hypothesis testing and truncation techniques in high-dimensional statistics~\cite{Ingster1997HypothesisTesting,verzelen2015community}. This approach yields a finite second-moment regime and allows us to identify sharp conditions under which detection is information-theoretically impossible.

We complement these lower bounds by proposing and analyzing three testing procedures that achieve the statistical threshold (up to constants and logarithmic factors). The first is a \emph{vanilla signed-triangle} test, based on counting signed triangles in the observed graph, extending ideas originally developed for detecting global high-dimensional geometry~\cite{bubeck2016testing}. The second is a \emph{scan signed-triangle} test, which counts signed triangles over every subset of $\binom{[n]}{k}$ vertices, corresponding to all candidate planted sets. The third is a \emph{geometry-agnostic scan test} that ignores the latent geometric structure and instead attempts to localize the planted subset purely through density-type evidence, connecting our problem to classical planted dense-subgraph and community-detection models~\cite{arias2014community}. Together, these tests demonstrate that latent geometry can be detected optimally using both geometry-aware and geometry-agnostic procedures, depending on the regime.

Our main results are summarized in the phase diagram shown in Figure~\ref{fig:spcaphasediagram2}. Specifically, we consider an asymptotic regime in which both $k$ and $d$ scale polynomially with $n$, namely
$d=\Theta(n^{\alpha})$ and $k=\Theta(n^{\beta})$ for some $\alpha>0$ and $\beta\in(0,1)$. Here, the exponent $\alpha$ captures the growth of the ambient dimension, while $\beta$ governs the size of the planted geometric structure. It is evident that the detection problem becomes statistically more challenging as $\alpha$ increases or $\beta$ decreases. In this regime, the $(\alpha,\beta)$ parameter space is partitioned into three distinct regions:
\begin{enumerate}
    \item \emph{Statistically impossible regime (gray):} detection is information-theoretically impossible when $\alpha > 2\beta \wedge 3$.    
    \item \emph{Computationally easy regime (blue):} there exists a polynomial-time algorithm for detection when $\alpha < 6\beta - 3$.    
    \item \emph{Computationally hard regime (red):} detection is information-theoretically possible when $\alpha < 2\beta$ and $\alpha > 0 \vee (6\beta-3)$, but computationally intractable in the sense
    that no polynomial-time algorithm is known; moreover, the class of low-degree polynomial tests fails throughout this region.
\end{enumerate}

\begin{figure}[t!]
\centering
%\resizebox{6cm}{7cm}{
\begin{tikzpicture}[scale=1]
\tikzstyle{every node}=[font=\footnotesize]
\def\xmin{0}
\def\xmax{4.5}
\def\ymin{0}
\def\ymax{7.5}

\draw[->] (\xmin,\ymin) -- (\xmax,\ymin) node[right] {$\beta$};
\draw[->] (\xmin,\ymin) -- (\xmin,\ymax) node[above] {$\alpha$};

\node at (4, 0) [below] {$1$};
\node at (2, 0) [below] {$\frac{1}{2}$};
\node at (3, 0) [below] {$\frac{3}{4}$};
\node at (0, 2) [left] {$1$};
\node at (0, 3) [left] {$\frac{3}{2}$};
\node at (0, 4) [left] {$2$};
\node at (0, 6) [left] {$3$};
\node at (0, 0) [left] {$0$};

\filldraw[fill=gray!25, draw=black] (0, 0) -- (3, 3) -- (4,6) -- (4, 7) -- (0,7);
\filldraw[fill=blue!25, draw=black] (2, 0) -- (4, 0) -- (4, 6) -- (2, 0);
\filldraw[fill=red!25, draw=black] (0, 0) -- (2, 0) -- (3, 3) -- (0, 0);
%\filldraw[fill=green!25, draw=black] (3, 0) -- (0, 4) -- (3, 4) -- (3, 0);
%\filldraw[fill=blue!25, draw=black] (3, 2) -- (3, 4) -- (0, 4) -- (1.5, 2) -- (3, 2);
\node[rotate=55] at (1.5, 3.8) {``$\s{\bs{Statistically}}\;\s{\bs{Impossible}}$"};
\node[rotate=80] at (3.4, 2.2) {``$\bs{Computationally}\;\s{\bs{Easy}}$"};
%\node[rotate=-60] at (1.57, 2.5) {``$\s{\bs{Computationally}}\;\s{\bs{Hard}}$"};
\node[rotate=40] at (1.4, 0.8) {``$\s{\bs{Hard}}$"};
\draw [dashed] (0,6) -- (4,6);
\draw [dashed] (0,3) -- (3,3);
\draw [dashed] (3,0) -- (3,3);
%\draw [dashed] (0,2) -- (1.5, 2);
\end{tikzpicture}%}
\caption{Phase diagram for detecting the presence of a planted random geometry subgraph, as a function of the geometry size $k = \Theta(n^{\beta})$ and the geometry dimension $d=\Theta(n^{\alpha})$.}
\label{fig:spcaphasediagram2}
\end{figure}
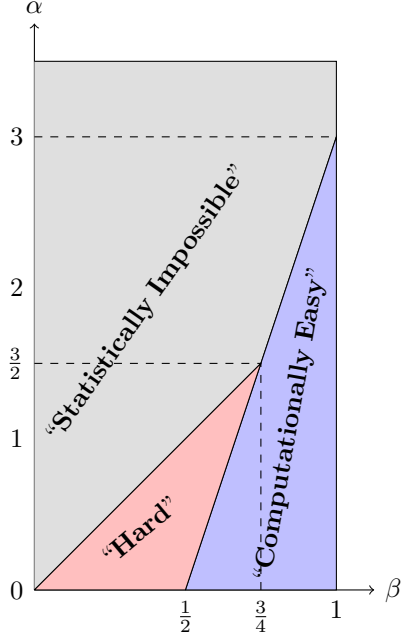

A key conceptual and technical challenge arises in the analysis of the scan signed-triangle test. Since this procedure ranges over an exponential family of candidate vertex sets, its analysis requires exponential tail bounds for the signed-triangle statistic, rather than the first- or second-moment estimates
that suffice for fixed tests. Indeed, moment-based methods constitute the main analytical tool in essentially all prior works on testing high-dimensional geometry in random graphs, including the signed-triangle analysis of Bubeck, Ding, Eldan, and R\'acz~\cite{bubeck2016testing}, subsequent extensions to noisy and softened geometric models~\cite{liu2023phase,liu2023probabilistic}, as well as related random-matrix approaches to geometric phase transitions such as anisotropic random geometric graphs and Wishart-type ensembles \cite{eldan_mikulincer_2020,brennan2021de}.  

While sharp upper-tail results for \emph{unsigned} triangle counts in Erd\H{o}s--R\'enyi graphs are by now classical—originating with the martingale-based approach of Kim and Vu~\cite{kim2004divide} and refined through subsequent combinatorial and variational methods \cite{janson_oleszkiewicz_rucinski_2004,BoucheronLugosiMassart2013,ganguly_hiesmayr_nam_2024}—these techniques do not extend to the signed setting, where substantial cancellations fundamentally alter the tail behavior. To overcome this obstacle, we develop new tools to characterize the relevant exponential rate, relying on strong decoupling inequalities for U-statistics \cite{de_la_pena_montgomery_smith_1995,delapena1999decoupling}.

Finally, we examine the computational landscape of the problem and provide evidence for an \emph{easy--hard--impossible} phase transition. In particular, we show that low-degree polynomial algorithms fail in a parameter regime that is conjecturally hard, leveraging the low-degree framework for average-case hardness in high-dimensional inference~\cite{Hopkins18,Dmitriy19}. This places the detection of planted geometric structure squarely within a broader class of inference problems exhibiting sharp separations between statistical possibility and computational feasibility.

\subsection{Related work} The study of random geometric graphs in high dimensions traces back to early work showing, via a multivariate central limit theorem, that geometric graphs become asymptotically indistinguishable from Erd\H{o}s--R\'enyi graphs as the dimension grows~\cite{devroye2011high}. This phenomenon was later sharpened by work that identified a precise phase transition in the dense regime, demonstrating that geometry is lost when the dimension exceeds a cubic threshold in the number of vertices~\cite{bubeck2016testing}. These results also uncovered deep connections between geometric graph models and classical random matrix ensembles, showing that a matching Wishart-to-GOE transition occurs at the same scale (see also~\cite{eldan_mikulincer_2020,brennan2024threshold}).

Extensions of these results have explored broader distributional assumptions and anisotropic settings, including information-theoretic limits for detecting geometry when the latent distribution is anisotropic~\cite{eldan_mikulincer_2020}. In the sparse regime, where the edge probability vanishes with the graph size, it was conjectured that geometry should be lost at significantly lower dimensions~\cite{bubeck2016testing}. Progress toward this conjecture, breaking the cubic barrier, was subsequently obtained~\cite{brennan2020phase}, with further refinements appearing in recent work on partial and masked Wishart ensembles~\cite{brennan2021de}. This remains an active and rapidly developing area. More recently, a unified framework for noisy high-dimensional geometric graphs that interpolate between purely geometric and Erd\H{o}s--R\'enyi models was introduced in~\cite{liu2023phase,liu2023probabilistic}. These results quantify how the strength of geometric dependence interacts with dimensionality to determine detectability. Our work complements this line by focusing on a \emph{localized} setting, where geometric structure is present only on a planted subset of vertices.

Beyond hard-threshold geometric graphs, a substantial literature studies \emph{soft} random geometric graphs, in which the probability of an edge depends smoothly on latent distances or inner products. Such models arise naturally in wireless communication, social networks, and biological systems. A foundational probabilistic treatment of geometric and soft geometric graphs, including connectivity properties in fixed dimensions, was developed in~\cite{Mathew07}. Subsequent works have analyzed connectivity and phase transitions from both probabilistic and statistical physics perspectives; see, for example, \cite{dettmann2016random} and related studies of one-dimensional and perturbed geometric networks~\cite{parthasarathy2017quest,wilsher2020connectivity}.

Detecting hidden structure in random graphs has also been extensively studied through planted subgraph models, including planted dense subgraphs and community detection. Sharp detection thresholds and computational barriers in dense regimes were characterized in~\cite{arias2014community,verzelen2015community}. More general formulations of planted subgraph detection, including arbitrary planted patterns, were studied in~\cite{huleihel2021inferring,elimelech2025detecting}, among others. These works typically posit an explicit increase in edge probability on a hidden vertex set, with edges remaining conditionally independent, a modeling choice that stands in contrast to the geometric mechanism considered here. In contrast, the alternative considered in the present paper is generated by a latent geometric mechanism, which induces structured dependencies among edges. This places our work at the intersection of planted subgraph detection and geometric random graph theory and, to our knowledge, provides the first sharp detection theory for a planted structure whose defining feature is high-dimensional geometry.

Concurrent and independent work by \cite{BokLiYu2026} studies the same geometric detection problem as in our paper and derives similar information-theoretic and computational limits. They also analyze the sparse regime where the edge density vanishes polynomially with the size of the graph, establishing corresponding statistical lower and upper bounds.

Finally, we mention \cite{Bet2020Detecting}. This paper studies the problem of detecting a botnet hidden within a larger network. The authors model the network as a graph and assume that a botnet appears as a group of nodes with slightly different connectivity patterns compared to normal users. More specifically, the hypothesis testing problem considered in their paper, is formulated with a geometric random graph as the null model, while under the alternative a ``small" Erdős–Rényi subgraph is planted within the geometric random graph. In contrast, our setting reverses this perspective: the null hypothesis is an Erdős–Rényi graph, and under the alternative we plant a geometric random graph. As a result, the underlying models, as well as the corresponding results and techniques, differ substantially.

\subsection{Notation} Given a probability distribution $\mathbb{P}$, we write $\mathbb{P}^{\otimes n}$ for the law of the $n$-dimensional random vector $(X_1,X_2,\dots,X_n)$, where the coordinates are independent and identically distributed according to $\mathbb{P}$. Likewise, $\mathbb{P}^{\otimes m \times n}$ denotes the distribution on $\mathbb{R}^{m\times n}$ whose entries are i.i.d. with common distribution $\mathbb{P}$. For a finite or measurable set $\mathcal{X}$, $\s{Unif}[\mathcal{X}]$ denotes the uniform distribution on $\mathcal{X}$. The notation $X\independent Y$ indicates that the random variables $X$ and $Y$ are independent.

For two $n\times n$ matrices $\bs{A}$ and $\s{B}$, we denote their Hadamard product by $\bs{A}\odot\s{B}$, and their matrix inner product by $\innerP{\bs{A},\s{B}}=\s{trace}(\bs{A}^T\s{B})$. For any real number $x\in\mathbb{R}$, we define $[x]_+=\max(x,0)$. The spectral (operator) norm of a symmetric matrix $\bs{A}$ is written as $\norm{\bs{A}}_{\s{op}}$, and $\bs{I}_n$ denotes the $n\times n$ identity matrix. For an $n\times n$ matrix $\bs{A}$ and a set $S\subseteq\{1,2,\ldots,n\}$, we use $\bs{A}|_S$ and $\bs{A}[S]$ to denote the submatrix obtained by restricting $\bs{A}$ to the rows and columns indexed by $S$.

We write $\calN(\mu,\sigma^2)$ for a univariate normal random variable with mean $\mu\in\mathbb{R}$ and variance $\sigma^2\in\mathbb{R}_{\ge0}$, and $\calN(\mu,\Sigma)$ for a multivariate normal random vector with mean $\mu\in\mathbb{R}^d$ and covariance matrix $\Sigma$, where $\Sigma$ is a $d\times d$ positive semidefinite matrix. Let $\Phi$ denote the cumulative distribution function of a standard normal random variable, given by $\Phi(x)=\int_{-\infty}^x e^{-t^2/2}\mathrm{d}t$. For probability measures $\mathbb{P}$ and $\mathbb{Q}$, we define the total variation distance, $\chi^2$-divergence, and Kullback--Leibler (KL) divergence, respectively, by $d_{\s{TV}}(\mathbb{P},\mathbb{Q})=\frac{1}{2}\int |\mathrm{d}\mathbb{P}-\mathrm{d}\mathbb{Q}|$, $\chi^2(\mathbb{P}||\mathbb{Q})=\int\frac{(\mathrm{d}\mathbb{P}-\mathrm{d}\mathbb{Q})^2}{\mathrm{d}\mathbb{Q}}$, $d_{\s{KL}}(\mathbb{P}||\mathbb{Q})=\bE_{\mathbb{P}}\pp{\log\frac{\mathrm{d}\mathbb{P}}{\mathrm{d}\mathbb{Q}}}$. We denote by $\s{Bern}(p)$ and $\s{Binomial}(n,p)$ the Bernoulli and Binomial distributions with parameters $p$ and $(n,p)$, respectively, and by $\s{Hypergeometric}(n,k,m)$ the Hypergeometric distribution with parameters $(n,k,m)$.

Throughout, we use standard asymptotic notation. For two positive sequences $\{a_n\}$ and $\{b_n\}$, we write $a_n=O(b_n)$ if there exists a constant $C$ such that $a_n\le C b_n$ for all $n$, and $a_n=\Omega(b_n)$ if $b_n=O(a_n)$. We write $a_n=\Theta(b_n)$ if both $a_n=O(b_n)$ and $a_n=\Omega(b_n)$ hold. Moreover, $a_n=o(b_n)$ (equivalently, $b_n=\omega(a_n)$) if $a_n/b_n\to0$ as $n\to\infty$. We write $a_n\ll b_n$ to indicate that $a_n$ is polynomially smaller than $b_n$ in $n$, that is, $\liminf_{n\to\infty}\log_n a_n<\liminf_{n\to\infty}\log_n b_n$. For $n\in\mathbb{N}$, we let $[n]=\{1,2,\dots,n\}$. For real numbers $a$ and $b$, we define $a\vee b=\max\{a,b\}$ and $a\wedge b=\min\{a,b\}$. Unless otherwise specified, $C$ denotes a generic constant independent of the problem parameters and may vary from line to line. For integers $n$ and $m$ we denote by $(n)_m$ the falling factorial, i.e., $(n)_m=n(n-1)(n-2)\cdots(n-m+1)$. Finally, for a set $S\subseteq\mathbb{R}$, $\mathds{1}\{S\}$ denotes its indicator function; we also write $\mathds{1}_S(x)$ with the same meaning, namely, $\mathds{1}_S(x)=1$ if $x\in S$ and $0$ otherwise.

\section{Setup and Problem Statement}\label{sec:setup}

We study the problem of detecting the presence of a small latent high-dimensional geometric subgraph in a random graph. Let $\calG(n,q)$ denote the Erd\H{o}s--R\'enyi random graph on $n$ vertices, in which each pair of vertices is connected independently with probability $q$. Under the null hypothesis $\calH_0$, the observed graph $\s{G}_n$ is drawn from $\calG(n,q)$.

To define the alternative hypothesis, we first recall the classical model of a random geometric graph. In this model, each vertex is associated with a point in a metric space, and an edge is present between two vertices if the distance between their corresponding labels is below a prescribed threshold. We focus on 
the case where the underlying metric space is the Euclidean sphere $\S^{d-1}$, and the latent labels are i.i.d. random vectors drawn uniformly from $\S^{d-1}$. We denote by $\calG_d(n,q)$ the ensemble of such random graphs, where $q$ is the marginal probability of an edge between any pair of vertices, and thus determines the threshold distance for connection.

Formally, $\calG_d(n,q)$ is defined as follows. Let $\bs{x}_1,\ldots,\bs{x}_n$ be independent random vectors, each uniformly distributed on $\S^{d-1}$. In $\calG_d(n,q)$, distinct vertices $i$ and $j$ are connected by an edge if and only if $\left\langle \bs{x}_i,\bs{x}_j \right\rangle\geq t_{q,d}$, where $|t_{q,d}|\leq1$ is chosen so that $\P\p{\left\langle \bs{x}_i,\bs{x}_j \right\rangle\geq t_{q,d}}=q$. Equivalently, this condition can be written as $\norm{\bs{x}_i-\bs{x}_j}_2^2\leq \tau_{q,d}$, where $\tau_{q,d}\triangleq2-2t_{q,d}\in[0,2]$. For notational convenience, we define $\sigma_{ij}\triangleq\mathds{1}\{\left\langle \bs{x}_i,\bs{x}_j \right\rangle\geq t_{q,d}\}$, for all $i,j\in[n]$. Under the alternative hypothesis $\calH_1$, the observed graph $\s{G}_n$ on $n$
vertices is generated as follows:
\begin{enumerate}
    \item A subset $\calK$ of $k$ vertices is selected uniformly at random from     the $n$ vertices. Each vertex $i\in\calK$ is associated with a latent label $\bs{x}_i\sim\s{Unif}(\mathbb{S}^{d-1})$.
    \item For any $i,j\in\calK$, the vertices are connected by an edge if and only if $\langle \bs{x}_i,\bs{x}_j \rangle \ge t_{q,d}$, where $t_{q,d}\in[-1,1]$ is chosen so that $\P\left(\langle \bs{x}_i,\bs{x}_j \rangle \ge t_{q,d}\right)=q$.
    \item All remaining edges, namely those with at least one endpoint outside $\calK$, are added independently with probability $q$.
\end{enumerate}
We denote by $\calG_d(n,k,q)$ the ensemble of random graphs generated by the above procedure. That is, $\calG_d(n,k,q)$ consists of graphs on $n$ vertices in which a geometric random subgraph $\calG_d(k,p)$ is \emph{planted} inside an Erd\H{o}s--R\'enyi random graph $\calG(n,q)$. The vertices in the set $\calK$ thus form a \emph{geometric community} whose internal connectivity is higher than that of the background graph. In this paper, we focus on the regime in which both edge probability $q$ is a fixed constant independent of $n$. Our goal is to address the following \emph{planted random geometry} ($\s{PRG}$) \emph{detection problem}.
\begin{definition}[$\s{PRG}$ detection problem]\label{def:PRG}
The $\s{PRG}$ detection problem with parameters $(n,k,d,q)$, hereafter denoted by $\s{PRG}(n,k,d,q)$, refers to the problem of distinguishing between the following two hypotheses:
\begin{align}
    \calH_0: \s{G}_n\sim\calG(n,q)\quad\quad\s{vs.}\quad\quad
    \calH_1: \s{G}_n\sim\calG_d(n,k,q).\label{eqn:model}
\end{align}
\end{definition}
\begin{rmk} We restrict attention to the setting in which the null and alternative hypotheses have identical edge marginals, so that the distinction between the two models is carried entirely by latent geometric dependencies. This restriction is deliberate. In the dense regime where edge probabilities
are fixed constants, introducing a different edge probability on the planted vertex set creates an explicit density contrast. In that case, detection reduces to classical planted dense subgraph testing, and the ambient dimension does not govern the detection threshold. By contrast, when the edge marginals match, geometry becomes the sole source of signal, and the interaction between the subgraph size and the ambient dimension is fundamental. We note that when edge probabilities are allowed to vary with $n$, density contrast and geometry can interact nontrivially, even if they differ, but analyzing such sparse regimes is beyond the scope of this work.
\end{rmk}
Upon observing $\s{G}_n$, a detection algorithm $\calA_n(\s{G}_n)\in\{0,1\}$ for the above problem outputs a decision in $\{0,1\}$. We define the \emph{risk} of a detection algorithm $\calA_n$ as the sum of its $\s{Type}$-$\s{I}$ and $\s{Type}$-$\s{II}$ error probabilities, namely,
\begin{align}
\s{R}_n(\calA_n) = \pr_{\calH_0}\bigl(\calA_n(\s{G}_n)=1\bigr) + \pr_{\calH_1}\bigl(\calA_n(\s{G}_n)=0\bigr),
\end{align}
where $\pr_{\calH_0}$ and $\pr_{\calH_1}$ denote the probability distributions under the null and alternative hypotheses, respectively. The optimal risk is defined as $\s{R}_n^\star\triangleq\inf_{\calA_n}\s{R}_n (\calA_n)$, where the infimum is taken over all (possibly randomized) tests $\calA_n\colon \s{G}_n \mapsto \{0,1\}$. A sequence of tests $\calA_n(\s{G}_n)\in\{0,1\}$ is said to achieve \emph{strong detection} if $\limsup_{n\to\infty}\s{R}_n (\calA_n)=0$, and \emph{weak detection} if $\limsup_{n\to\infty}\s{R}_n (\calA_n)<1$. Conversely, we say that \emph{strong detection is impossible} if $\liminf_{n\to\infty}\s{R}_n^\star>0$, and that \emph{weak detection is impossible} if $\lim_{n\to\infty}\s{R}_n^\star=1$.

The algorithms we consider are either unconstrained—allowing for computationally expensive procedures—or restricted to run in polynomial time, corresponding to computationally efficient algorithms. Unconstrained algorithms are typically used to establish information-theoretic limits and to show the tightness of lower bounds. An algorithm is said to be polynomial-time if its running time is bounded by $\s{poly}(n)$, where $n$ denotes the size of the input. As discussed in the introduction, our goal is to characterize necessary and sufficient conditions under which detecting the planted geometric subgraph is possible or impossible, both in the absence of computational constraints and under polynomial-time restrictions.

\section{Main Results}\label{sec:main}

In this section, we present our main results. We begin with information-theoretic lower bounds establishing regimes in which detecting the planted geometric subgraph is impossible, regardless of computational constraints. We then present algorithmic upper bounds demonstrating that these limits are achievable, and conclude with computational lower bounds identifying regimes where detection is statistically possible but computationally hard. 

\paragraph{Statistical lower bounds.} We start with information-theoretic lower bounds for the $\s{PRG}(n,k,d,q)$ detection problem. By the characterization of optimal hypothesis testing in terms of total variation distance, the optimal risk satisfies (see, e.g.,~\cite[Theorem~2.2]{tsybakov2004introduction})
\begin{align}
\s{R}^\ast = 1 - d_{\s{TV}}(\P_{\calH_0},\P_{\calH_1}).
\label{eqn:optimalrisk}
\end{align}
The following theorem establishes that when the ambient dimension grows sufficiently fast relative to the size of the planted geometric subgraph, no test can reliably detect the presence of geometry.
\begin{theorem}[Impossibility of strong detection]\label{thm:ITlower}
Consider the detection problem defined in~\eqref{eqn:model}. If the following conditions hold simultaneously:
%\begin{align}
%\frac{d}{k^2\log^2 k}\to\infty, \quad    k=o(d\log n), \quad \frac{dn^3}{k^6\log^2 k}\to\infty,\quad k\leq C_{p,q}\log n,
%\end{align}
\begin{align}
\frac{d}{k^2\log^2 k}\to\infty\quad\s{and} \quad \frac{dn^3}{k^6\log^2 k}\to\infty,
\end{align}
then $d_{\s{TV}}(\P_{\calH_0},\P_{\calH_1})<1$, i.e., strong detection is impossible.
\end{theorem}
Moreover, the same proof technique shows that if the condition $\frac{dn^3}{k^6\log^2 k}\to\infty$ is replaced by $\frac{dn^3}{k^6 f(k)}\to\infty$ for some function $f(k)=\omega(\log^2 k)$, then $d_{\s{TV}}(\P_{\calH_0},\P_{\calH_1})\to 0$, implying that weak detection is impossible. 

It is instructive to compare Theorem~\ref{thm:ITlower} with previously known results in special cases. In the vanilla setting where $k=n$, the model reduces to testing for global high-dimensional geometry, and our lower bound essentially recovers the known phase transition at dimension $d\asymp n^{3}$ established in~\cite{bubeck2016testing}, up to a $\log^{2}k$ factor. We note that if the edge probability inside the planted set were taken to be $p\neq q$, then detection would be information-theoretically impossible whenever $k=O(\log n)$, regardless of the value of $d$. This impossibility is not driven by geometry, but instead reflects a fundamental limitation already present in classical planted clique and planted dense subgraph models: when the planted subgraph is too small, even a purely combinatorial increase in edge density cannot be reliably detected. Thus, only when $p=q$, geometry is the sole distinguishing feature, and the interplay between dimension and subgraph size becomes central.

\begin{comment}
    \begin{theorem}[Statistical lower bounds]\label{thm:ITlower}
    Consider the detection problem in \eqref{eqn:model}. 
    \begin{enumerate}
        \item  If $d/(k\log k)^2\to\infty$, $k=o(d\log n)$, $(dn^3)/(k^6\log^2k)\to\infty$, and $\frac{(p-q)^2}{p(1-p)}=o\left(\frac{\log n}{k}\right)$, all hold simultaneously, then $d_{\s{TV}}(\P_{\calH_0},\P_{\calH_1})<1$ implying that strong detection is impossible. \label{thm:strong-impossibility}
    \item If $d/(k^2\log k)\to\infty$, $k=o(d\log n)$, $(dn^3)/(k^6f(k))\to\infty$, where $f(k)=\omega(\log^2k)$ and $\frac{(p-q)^2}{p(1-p)}=o\left(\frac{\log n}{k}\right)$, all hold simultaneously, then $d_{\s{TV}}(\P_{\calH_0},\P_{\calH_1})\to0$ implying that weak detection is impossible. \label{thm:weak-impossibility}
    \end{enumerate}
\end{theorem}
\end{comment}

\paragraph{Upper bounds.} We now turn to algorithmic upper bounds. Specifically, we propose three detection algorithms and analyze their associated risks. To this end, we first introduce some notation. Let $\bs{A}$ denote the adjacency matrix of the observed graph $\s{G}_n$; we suppress the dependence of $\bs{A}$ on $\s{G}_n$, as the underlying graph will always be clear from the context. For distinct vertices $i,j,\ell\in[n]$, define
\begin{align}
\bar{\bs{A}}_{i,j}\triangleq \bs{A}_{i,j}-\E\!\left[\bs{A}_{i,j}\right],
\qquad
\s{T}_{\s{G}_n}(i,j,\ell)
\triangleq
\bar{\bs{A}}_{i,j}\bar{\bs{A}}_{i,\ell}\bar{\bs{A}}_{j,\ell}.
\end{align}
We consider the following statistics:
\begin{align}
    \s{T}_{\s{triangle}}(\s{G}_n)&\triangleq\sum_{\{i,j,\ell\}\subset\binom{[n]}{3}}\s{T}_{\s{G}_n}(i,j,\ell),\label{eqn:numberSigned}\\
    \s{T}_{\s{scan}}(\s{G}_n)&\triangleq\max_{\calS\subset[n]:\abs{\calS}=k}\sum_{\{i,j,\ell\}\subset\binom{\calS}{3}}\s{T}_{\s{G}_n}(i,j,\ell). \label{eqn:scanSigned}
    %\s{T}_{\s{dense}}(\s{G}_n)&\triangleq\max_{\calS\subset[n]:|\calS|=k}\sum_{i,j\in\calS:i<j}\mathbf{A}_{i,j}.\label{eqn:scandense}
\end{align}
The statistic $\s{T}_{\s{triangle}}(\s{G}_n)$ in \eqref{eqn:numberSigned} counts the total number of signed triangles in $\s{G}_n$. The scan statistic $\s{T}_{\s{scan}}(\s{G}_n)$ in \eqref{eqn:scanSigned} enumerates all $k$-vertex induced subgraphs and selects the one with the largest signed-triangle count. %Finally, $\s{T}_{\s{dense}}(\s{G}_n)$ in \eqref{eqn:scandense} scans over all $k$-vertex subgraphs and selects the densest one in terms of edge count. 
Based on these statistics, we define the following detection algorithms:
\begin{align}
    \calA_{\s{triangle}}(\s{G}_n)&\triangleq\mathds{1}\ppp{\s{T}_{\s{triangle}}(\s{G}_n)\geq \tau_{\s{triangle}}},\label{eqn:triangleTest}\\
    \calA_{\s{scan}}(\s{G}_n)&\triangleq\mathds{1}\ppp{\s{T}_{\s{scan}}(\s{G}_n)\geq \tau_{\s{scan}}},\label{eqn:scanTest}
    %\calA_{\s{dense}}(\s{G}_n) &\triangleq\Ind\ppp{\s{T}_{\s{dense}}(\s{G}_n)\geq \tau_{\s{dense}}},\label{eqn:denseTest}
\end{align}
where $\tau_{\s{triangle}},\tau_{\s{scan}},\tau_{\s{dense}}\in\mathbb{R}_+$ are thresholds specified below.

We briefly motivate this choice of statistics. Consider the conditional likelihood ratio $\calL_{n\mid\calK}\triangleq
\pr_{\calH_1\mid\calK}/\pr_{\calH_0}$. Expanding $\calL_{n\mid\calK}$ in the orthonormal basis associated with the entries of $\calG(n,q)$, the expansion up to degree three takes the form
\begin{align}
\calL_{n\mid\calK}
\approx
1+
\frac{1}{\sqrt{d}}
\sum_{\{i,j,\ell\}\subset\binom{\calK}{3}}
\s{T}_{\s{G}_n}(i,j,\ell).
\end{align}
Thus, signed triangles constitute the lowest-degree nontrivial term in the likelihood-ratio expansion, explaining why triangle-based statistics arise naturally. Since the likelihood-ratio test is optimal by the Neyman--Pearson lemma, it is reasonable to expect low-degree proxies of the likelihood ratio to be powerful as well. Our algorithmic guarantees are summarized in the following theorem.
\begin{theorem}[Strong detection upper bounds]\label{thm:upperBound}
    Consider the detection problem in \eqref{eqn:model}, and the signed triangle and triangle scan tests in \eqref{eqn:triangleTest}, and \eqref{eqn:scanTest}, respectively, with $\tau_{\s{scan}}=\tau_{\s{triangle}}\triangleq\binom{k}{3}\frac{\s{C}_q}{\sqrt{d}}$% and $\tau_{\s{dense}} = \binom{k}{2}\frac{p+q}{2}$
    , for a constant $\s{C}_q$ depending only on $q$.
    \begin{enumerate}
        \item If $dn^3/k^6\to0$, then $\s{R}(\calA_{\s{triangle}})\to0$, as $d,k,n\to\infty$. \label{thm:upperbound_triangle}
        \item If $k^2/d>C\log^2 n$, for some sufficiently large constant $C>0$, then $\s{R}(\calA_{\s{scan}})\to0$, as $k,n\to\infty$. \label{thm:upperbound_scan}
        %\item If $\frac{(p-q)^2}{q(1-q)} = \omega\p{\frac{\log n}{k}}$, then $\s{R}(\calA_{\s{dense}})\to0$, as $k,n\to\infty$. \label{thm:upperbound_dense}
    \end{enumerate}
\end{theorem}
Theorem~\ref{thm:upperBound} complements the information-theoretic lower bound in Theorem~\ref{thm:ITlower} up to polylogarithmic factors. In particular, Theorem~\ref{thm:ITlower} together with Theorem~\ref{thm:upperBound} identifies the phase transition for optimal detection (up to the stated logarithmic terms). In the case where the edge probability inside the planted set is taken to be $p\neq q$, the optimal test scans over all $k$-vertex subgraphs and selects the one with the largest number of edges. We show in Appendix~\ref{app:apppnotq} that strong detection is achievable whenever $k=\Omega(\log n)$.

We briefly remark on a key technical difficulty underlying Theorem~\ref{thm:upperBound}, which is the analysis of the signed triangle scan test. Unlike the vanilla signed triangle statistic, the scan test ranges over an exponential family of candidate vertex sets, which necessitates control of exponential tail probabilities rather than moment bounds. While first- and second-moment analyses suffice for fixed tests and have been the primary tools in prior work on testing high-dimensional geometry in random graphs, they are insufficient in the scan setting due to the competing combinatorial complexity. Addressing this challenge requires new techniques to characterize the relevant exponential rates for signed triangle counts, which we develop using decoupling methods for U-statistics.

\paragraph{Computational lower bounds.} We begin with a brief overview of the low-degree polynomial method. The central idea of this approach is to use low-degree multivariate polynomials in the entries of the observed data as surrogates for computationally efficient procedures. This heuristic is motivated by the observation that many known polynomial-time algorithms can be approximated, in an appropriate sense, by
polynomials of bounded degree. The foundations of this methodology originate from a sequence of works in the sum-of-squares optimization literature \cite{barak2016nearly,Hopkins18,hopkins2017bayesian,hopkins2017power}.

We follow the notation and formalism introduced in \cite{Hopkins18,Dmitriy19}. Let $\pr_{\calH_0}$ be a probability distribution on $\Omega_n=\{0,1\}^{\binom{n}{2}}$. This distribution induces an inner product on measurable functions $f,g:\Omega_n\to\mathbb{R}$ given by $\langle f,g\rangle_{\calH_0}\triangleq\E_{\calH_0}[f(\s{G})g(\s{G})]$, with associated norm $\|f\|_{\calH_0}=\langle f,f\rangle_{\calH_0}^{1/2}$. We denote by $L^2(\pr_{\calH_0})$ the Hilbert space of functions with finite $\|\cdot\|_{\calH_0}$ norm, equipped with this inner product.

In the absence of computational constraints, the Neyman--Pearson lemma implies that the likelihood ratio test achieves the optimal tradeoff between $\mathsf{Type}$-$\mathsf{I}$ and $\mathsf{Type}$-$\mathsf{II}$ error probabilities. Equivalently, the likelihood ratio is the optimal distinguisher between $\pr_{\calH_0}$ and $\pr_{\calH_1}$ in the $L^2(\pr_{\calH_0})$ sense. Writing $\calL_n\triangleq \pr_{\calH_1}/\pr_{\calH_0}$ for the likelihood ratio, standard second-moment arguments show that if $\|\calL_n\|_{\calH_0}^2$ remains bounded as $n\to\infty$, then $\pr_{\calH_1}$ is contiguous to $\pr_{\calH_0}$. In this case, the two distributions are statistically indistinguishable, in the sense that no test can simultaneously drive both error probabilities to zero.

The low-degree method asks whether a similar conclusion holds when one restricts attention to low-degree polynomials. To formalize this, let $\calV_{n,\le D}\subset L^2(\pr_{\calH_0})$ denote the subspace of polynomial functions $\Omega_n\to\mathbb{R}$ of total degree at most $D$. Let $\calP_{\le D}:L^2(\pr_{\calH_0})\to \calV_{n,\le D}$ be the corresponding orthogonal projection operator. The \emph{$D$-low-degree likelihood ratio} is defined as $\calL_{n,\le D}
\triangleq
\calP_{\le D}\calL_n$,
that is, the orthogonal projection of the likelihood ratio onto the space of degree-$D$ polynomials with respect to the inner product $\langle\cdot,\cdot\rangle_{\calH_0}$. Since the full likelihood ratio is the optimal $L^2$ distinguisher, the projected likelihood ratio plays an analogous optimal role among all degree-$D$ polynomials. The following lemma formalizes this property. 
\begin{lemma}[Optimality of the low-degree likelihood ratio
{\cite{hopkins2017bayesian,hopkins2017power,Dmitriy19}}]\label{lem:Dmitriy}
Consider the optimization problem
\begin{align}
\begin{aligned}
\mathrm{max}
\;\bE_{\calH_1}\pp{f(\s{G})}
\quad\mathrm{s.t.}
\quad\bE_{\calH_0}\pp{f^2(\s{G})} = 1,\; f\in\calV_{n,\leq D},
\end{aligned}\label{eqn:optimizationProblem}
\end{align}
%\begin{align}
%\begin{aligned}
%& \mathrm{maximize}
%& & \bE_{\calH_1}f(\s{G}) \\
%& \mathrm{subject}\;\mathrm{to}
%& & \bE_{\calH_0}f^2(\s{G}) = 1,\; %f\in\calV_{n,\leqD},
%\end{aligned}\label{eqn:optimizationProblem}
%\end{align}
The unique maximizer is $f^\star = \calL_{n,\leq D}/\norm{\calL_{n,\leq D}}_{\calH_0}$, and the optimal value equals $\norm{\calL_{n,\leq D}}_{\calH_0}$. 

\end{lemma}
In the computationally unbounded setting, boundedness of
$\|\calL_n\|_{\calH_0}$ implies statistical indistinguishability between
$\pr_{\calH_0}$ and $\pr_{\calH_1}$. The low-degree method asserts that an
analogous principle governs computational limits, with $\calL_{n,\le D}$
playing the role of the likelihood ratio when attention is restricted to
efficiently computable tests. This intuition is captured by the following informal version of the
low-degree conjecture (see~\cite{Hopkins18,hopkins2017bayesian,hopkins2017power}
and~\cite[Conj.~1.16]{Dmitriy19}).

\begin{conjecture}[Low-degree conjecture (informal)]\label{conj:1}
If there exist $\epsilon>0$ and
$D=D(n)\ge (\log n)^{1+\epsilon}$ such that
$\|\calL_{n,\le D}\|_{\calH_0}$ remains bounded as $n\to\infty$, then no
polynomial-time algorithm can distinguish $\pr_{\calH_0}$ from
$\pr_{\calH_1}$ (i.e., achieve strong detection).
\end{conjecture}

In what follows, we use Conjecture~\ref{conj:1} to provide evidence for the
statistical--computational gap observed above. We start with the following result.
\begin{theorem}[Statistical-computational gap]\label{thm:gap}
Consider the problem in \eqref{eqn:model}. Suppose there exists $\varepsilon>0$ such that: 1) $k\leq n^{1/2-\varepsilon}$ or 2) $\frac{k^6}{n^3d}\log^3d\leq n^{-\varepsilon}$ and $k=\Omega(\sqrt{n})$, for all large $n$. Then, there exists $C = C(\varepsilon)$ such that if $D\leq C\log n/\log\log n$, one has $\norm{\calL_{n,\leq D}}_{\calH_0}= O(1)$. On the other hand, if $\frac{k^6}{n^3d}=\omega(1)$, then $\norm{\calL_{n,\leq D}}_{\calH_0}= \omega(1)$.
%Then, if $1\vee\frac{k^6}{n^3}\ll d\ll k^2$, then $\norm{\calL_{n}^{\leq D}}_{\calH_0}\leq O(1)$, for any $D = \Omega(\log n)$. On the other hand, if $d\ll\frac{k^6}{n^3}$, then $\norm{\calL_{n}^{\leq D}}_{\calH_0}\geq \omega(1)$.
\end{theorem}
Theorem~\ref{thm:gap} establishes boundedness of the low-degree likelihood ratio up to $D \le \frac{\log n}{\log\log n}$, which is near-logarithmic and thus provides evidence for computational hardness within this framework. This interpretation is consistent with analogous results in related planted problems (e.g., \cite{BangachevBresler2024FourierRGG}, \cite{elimelech2025detecting}, \cite{wein2025computational}). Together with Conjecture~\ref{conj:1}, Theorem~\ref{thm:gap} suggests that, if low-degree polynomials are taken as a proxy for efficient computation, then no polynomial-time algorithm can distinguish the null and alternative hypotheses in the regime $1\vee\frac{k^6}{n^3}\ll d\ll k^2$. Equivalently, low-degree polynomials fail whenever either the problem is planted-clique hard, or the geometric signal is sufficiently weak that even the signed-triangle statistic is ineffective. This conclusion is summarized in the following corollary.

These predictions align precisely with the statistical--computational tradeoffs identified above. A more explicit characterization of the computational barrier, exhibiting its dependence on $D$, can be extracted from the proof of Theorem~\ref{thm:gap}; for clarity of exposition, we have chosen to present the simplified formulation stated above. A key ingredient in the proof of Theorem~\ref{thm:gap} is a recent powerful result \cite[Thm. 1.1]{BangachevBresler2024FourierRGG}, which bounds centered subgraph moments (equivalently, Fourier coefficients) of spherical random geometric graphs by showing they decay polynomially in $d$. This decay, together with the elementary observation that tree-like subgraphs contribute nothing to the low-degree likelihood, ensures that only cyclic subgraphs contribute to the low-degree likelihood; among these their total contribution can be controlled when $\frac{k^6}{n^3}\ll d$.

\section{Proofs of Upper Bounds}

In this section, we prove Theorem~\ref{thm:upperBound}. We begin by analyzing the $\s{Type}$-$\s{I}$ and $\s{II}$ error probabilities of the signed triangle test, and then turn to the triangle scan test.

\subsection{Signed triangle test}

Recall the signed triangle test in \eqref{eqn:triangleTest}. To analyze its performance we use the second moment technique. To that end, we start by finding the first and second moments of $\s{T}_{\s{triangle}}(\s{G}_n)$ in \eqref{eqn:numberSigned}, under $\calH_0$ and $\calH_1$. Specifically, it is easy to show that \cite[Section 3.1]{bubeck2016testing}
    \begin{align}
        \bE_{\calH_0}\pp{\s{T}_{\s{triangle}}(\s{G}_n)} &= 0,\\
        \s{Var}_{\calH_0}\p{\s{T}_{\s{triangle}}(\s{G}_n)}&=
        \binom{n}{3} q^3 (1-q)^3\leq n^3.\label{eqn:H0Variance}
    \end{align}
    On the other hand, define $\calT_\calK \triangleq \left\{\{i,j,\ell\}\subset\binom{[n]}{3}:i,j,\ell\in\calK\right\}$, that is, the collection of all triples whose indices lie entirely within the planted set $\calK$. Then, observe that
    \begin{align}
        \bE_{\calH_1}\pp{\s{T}_{\s{triangle}}(\s{G}_n)}& =  \bE_{\calH_1}\pp{\sum_{\{i,j,\ell\}\subset\binom{[n]}{3}}\s{T}_{\s{G}_n}(i,j,\ell)}\\
        & = \sum_{\{i,j,\ell\}\in\calT_\calK}\bE_{\calH_1}\pp{\s{T}_{\s{G}_n}(i,j,\ell)}\\
        &\geq\binom{k}{3}\frac{\s{C}_q}{\sqrt{d}},
    \end{align}
    where the second equality follows from the fact that the expectation of a signed triangle $\s{T}_{\s{G}_n}(i,j,\ell)$ is zero unless all three indices $(i,j,\ell)$ belong to the planted set $\calK$. The inequality then follows from \cite[Lemma~3]{bubeck2016testing}, where $\s{C}_q$ denotes a constant depending only on $q$. Finally,
    \begin{align}
        \s{Var}_{\calH_1}\p{\s{T}_{\s{triangle}}(\s{G}_n)}& = \s{Var}_{\calH_1}\p{\sum_{\{i,j,\ell\}\subset\binom{[n]}{3}}\s{T}_{\s{G}_n}(i,j,\ell)}\\
        & \leq 2 \cdot\s{Var}_{\calH_1}\p{\sum_{\{i,j,\ell\}\in\calT_\calK}\s{T}_{\s{G}_n}(i,j,\ell)}+2\cdot\s{Var}_{\calH_1}\p{\sum_{\{i,j,\ell\}\in\calT_\calK^c}\s{T}_{\s{G}_n}(i,j,\ell)}\label{eqn:twoVarTerms}\\
        &\leq 2k^3+\frac{6k^4}{d}+2n^3,
    \end{align}
    where first inequality follows from the fact that $\s{Var}(X+Y)\leq 2\cdot\s{Var}(X)+2\cdot\s{Var}(Y)$, for any pair of random variables $(X,Y)$. The second inequality is obtained by bounding the first variance term in \eqref{eqn:twoVarTerms} using~\cite[eq.~(29)]{bubeck2016testing}, and the second
variance term using~\eqref{eqn:H0Variance}. Consequently, by Chebyshev's inequality, we obtain
    \begin{align}
        \P_{\calH_0}\p{\s{T}_{\s{triangle}}(\s{G}_n) \geq \frac{\bE_{\calH_1}\pp{\s{T}(\s{G}_n)}}{2}}\leq  \s{C}_1\cdot\frac{dn^3}{k^6},
    \end{align}
    for some constant $\s{C}_1>0$, and
    \begin{align}
        \P_{\calH_1}\p{\s{T}_{\s{triangle}}(\s{G}_n)<\frac{\bE_{\calH_1}\pp{\s{T}(\s{G}_n)}}{2}}\leq \s{C}_2\frac{dk^3+3k^4+dn^3}{k^6},
    \end{align}
    for some constant $\s{C}_2>0$. Therefore, we conclude that $\s{R}(\calA_{\s{triangle}})\to 0$ as $n\to\infty$ provided that $dn^3/k^6\to 0$ and $d/k^3\to 0$. Since the latter condition is implied by the former, this reduces to the requirement $dn^3/k^6\to 0$, as stated in item~\ref{thm:upperbound_triangle} of Theorem~\ref{thm:upperBound}.
%%%%%%%%%%%%%%%%

\subsection{Triangle scan test}

%\wasim{The analysis below is true only for the special case where $p=q$, right? Am I missing something; when you define the signed triangles, you subtract the mean w.r.t. $\calH_0$, right?}
%\mor{In the proof of Type $\s{II}$ error its just the bound of the variance of the number of signed triangles in $\calH_1$. As for the proof of Type $\s{I}$, I subtract the mean $\E\pp{A_{ij}}=q$ since this is $\calH_1$.}

We analyze the $\s{Type}$-$\s{I}$ and $\s{Type}$-$\s{II}$ error probabilities
associated with the triangle scan test defined in \eqref{eqn:scanSigned} and \eqref{eqn:scanTest}. We define $\tau_{\s{scan}} = \frac{1}{2}\binom{k}{3}\frac{\s{C}_q}{\sqrt{d}}$, and denote the underlying planted set by $\calK$. Note that
\begin{align}
    \zeta\triangleq\bE_{\calH_1}\pp{\sum_{\{i,j,\ell\}\subset\binom{\calK}{3}}\s{T}_{\s{G}_n}(i,j,\ell)}\geq \binom{k}{3}\frac{\s{C}_q}{\sqrt{d}},
\end{align}
where the inequality follows from~\cite[Lemma~3]{bubeck2016testing}, and $\s{C}_q$, as before, denotes a constant depending only on $q$. We now turn to the $\s{Type}$-$\s{II}$ error probability; applying Chebyshev's inequality yields
\begin{align}
    \P_{\calH_1}\p{\s{T}_{\s{scan}}(\s{G}_n)<\tau_{\s{scan}}}& = \P_{\calH_1}\p{\max_{\calS\subset[n]:\abs{\calS}=k}\sum_{\{i,j,\ell\}\subset\binom{\calS}{3}}\s{T}_{\s{G}_n}(i,j,\ell)<\tau_{\s{scan}}}\\
    &\leq \P_{\calH_1}\p{\sum_{\{i,j,\ell\}\subset\binom{\calK}{3}}\s{T}_{\s{G}_n}(i,j,\ell)<\tau_{\s{scan}}}\\
    &\leq \frac{\s{Var}_{\calH_1}\p{\sum_{\{i,j,\ell\}\subset\binom{\calK}{3}}\s{T}_{\s{G}_n}(i,j,\ell)}}{(\zeta-\tau_{\s{scan}})^2}\\
    &\leq \s{C}_3\frac{dk^3+3k^4}{k^6},
\end{align}
for some constant $\s{C}_3>0$, and the last inequality follows from \cite[eq. (29)]{bubeck2016testing}. Thus, we see that the $\s{Type}$-$\s{II}$ converges to zero if $d/k^3\to0$ and $k\to\infty$. 

Next, we analyze the $\s{Type}$-$\s{I}$ error probability. By the union bound, we have
\begin{align}
    \P_{\calH_0}\p{\s{T}_{\s{scan}}(\s{G}_n)\geq\tau_{\s{scan}}}\leq \binom{n}{k}\cdot\P_{\calH_0}\p{\sum_{\{i,j,\ell\}\subset\binom{\calK}{3}}\s{T}_{\s{G}_n}(i,j,\ell)\geq\tau_{\s{scan}}},\label{eqn:probabScan}
\end{align}
and thus, it remains to upper bound the probability term on the right-hand side of \eqref{eqn:probabScan}. The following result provides the required bound.
\begin{lemma}\label{lem:probabScanUpper}
Fix $\calK\subset[n]$ with $|\calK|=k\ge 3$. Let $\tau_{\s{scan}}$ be any threshold of the form $\tau_{\s{scan}}\geq c_0 k^3/\sqrt d$ for a constant $c_0=c_0(q)>0$. Then, for some $C=C(q)>0$ and $C'=C'(q)>0$,
\begin{align}
    \P_{\calH_0}\p{\sum_{\{i,j,\ell\}\subset\binom{\calK}{3}}\s{T}_{\s{G}_n}(i,j,\ell) \geq \tau_{\s{scan}}}\leq \exp\p{-C\frac{k^3}{d\alpha \log n}}+\exp\p{-\alpha C'k\log n},
    \label{eq:SignedBound}
\end{align}
for any $\alpha>0$.
\end{lemma}
Assuming the validity of Lemma~\ref{lem:probabScanUpper}, we can bound \eqref{eqn:probabScan} as follows:
\begin{align}
    \P_{\calH_0}\p{\s{T}_{\s{scan}}(\s{G}_n)\geq\tau_{\s{scan}}}&\leq \binom{n}{k} \cdot \pp{\exp\p{-C\frac{k^3}{\alpha d\log n}}+\exp\p{-C'k\log n}}\\
    &\leq \exp\pp{k\cdot\p{\log n-C\frac{k^2}{\alpha d\log n}}}+\exp\p{-(\alpha C'-1)\log n}.\label{eqn:upperTailFinal}
\end{align}
Thus, we see that the first term at the right-hand side of \eqref{eqn:upperTailFinal} converges to zero provided that $k^2 > \frac{\alpha d\log^2 n}{C}$ and $k\to\infty$, while the second term converges to zero by taking any fixed constant $\alpha>1/C'$, and $k\to\infty$. It therefore remains to establish Lemma~\ref{lem:probabScanUpper}. To this end, we prove a sequence of auxiliary results, beginning with the following lemma.
\begin{lemma}\label{lem:mgf-decoupled}
Fix $\calK\subset[n]$ with $|\calK|=k\ge 3$, and let $\s{G}\sim \calG(n,q)$ under $\calH_0$. For an edge $e=(i,j)\subset \binom{\calK}{2}$, set
\begin{align}
\bar{\mathbf{A}}_{e}&\triangleq \mathbf{A}_{e}-q,\\
W_e&\triangleq \sum_{\ell\in \calK\setminus\{i,j\}} \bar{\mathbf{A}}_{i\ell} \bar{\mathbf{A}}_{j\ell},\\
\s{T}_{\calK}(\s{G})&\triangleq \sum_{\{i,j,\ell\}\subset \binom{\calK}{3}}\bar{\mathbf{A}}_{ij} \bar{\mathbf{A}}_{i\ell} \bar{\mathbf{A}}_{j\ell}
=\frac{1}{3}\sum_{e\subset \binom{\calK}{2}}\bar{\mathbf{A}}_e W_e.
\end{align}
There exists an absolute constant $C_{\mathrm{dec}}\in(0,\infty)$ (one may take $C_{\mathrm{dec}}=8$) such that for every $\theta\in\mathbb R$,
\begin{align}\label{eq:mgf-main}
\bE\pp{\exp\p{\theta \s{T}_{\calK}(\s{G})}}\leq\bE\pp{\exp\p{\frac{C_{\mathrm{dec}}^2 \theta^2}{72}\sum_{e\subset \binom{\calK}{2}} W_e^2}}.
\end{align}
\end{lemma}

\begin{proof}[Proof of Lemma~\ref{lem:mgf-decoupled}]
Introduce three \emph{independent} copies of the edges $\{\mathbf{A}_{ij}\}$ under $\calH_0$, denoted by $\{\mathbf{A}^{(1)}_{ij}\}, \{\mathbf{A}^{(2)}_{ij}\}, \{\mathbf{A}^{(3)}_{ij}\}$, each distributed as $G(n,q)$, and write $\bar{\mathbf{A}}^{(r)}_{ij}\triangleq \mathbf{A}^{(r)}_{ij}-q$. Define the fully decoupled statistic
\begin{align}
\s{T}_{\calK}^{\mathrm{dec}}(\s{G})
 \triangleq  
\sum_{\{i,j,\ell\}\subset \binom{\calK}{3}}\bar{\mathbf{A}}^{(1)}_{ij} \bar{\mathbf{A}}^{(2)}_{i\ell} \bar{\mathbf{A}}^{(3)}_{j\ell}.
\end{align}
We use the following classical decoupling inequality for \emph{canonical} $U$-statistics of order $3$. We start with the following definition.
\begin{definition}
Let $X_1,X_2,\dots$ be i.i.d. random variables with values in a measurable space $(\mathcal X,\mathcal A)$.  For an integer $m\geq 1$, let $h:\mathcal X^m\to\mathbb R$ be a symmetric
measurable function (the kernel). The kernel $h$ is said to be \emph{canonical} if
\begin{align}\label{eq:canonical}
\bE\pp{\left.h(X_1,\dots,X_m)\right|X_r}=0
\qquad \text{almost surely for every } r=1,\dots,m.
\end{align}
A U-statistic of order $m$ built from a canonical kernel is called a \emph{canonical U-statistic}.
\end{definition}
\begin{lemma}[Decoupling for U-statistics {\cite[Thm.~3.1.1]{delapena1999decoupling}}]\label{thm:decoupling}
Let $U$ be a canonical U-statistic of order $m$ built from i.i.d.\ inputs.
Let $U^{\mathrm{dec}}$ denote the corresponding fully decoupled version,
constructed by evaluating the kernel on $m$ independent copies of the input sequence.
Then there exists a constant $C_m$ depending only on the order $m$ such that,
for every convex nondecreasing function $\Phi:\mathbb{R}\to\mathbb{R}$,
\begin{align}
\bE \pp{\Phi(U)} \leq \bE \pp{\Phi\p{C_m U^{\mathrm{dec}}}}.
\end{align}
One may take $C_m=2^{ m-1}$.
\end{lemma}
Applying this lemma with $\Phi(x)=e^{\theta x}$ and $U=\s{T}_{\calK}(\s{G})$ we get
\begin{align}\label{eq:decouple}
\bE\pp{\exp\p{\theta \s{T}_{\calK}(\s{G})}}\leq\bE\pp{\exp\p{{C_{\mathrm{dec}}\theta \s{T}_{\calK}^{\mathrm{dec}}(\s{G})}}},
\end{align}
where $C_{\mathrm{dec}}$ is an absolute constant. Now, for each edge $e=(i,j)$, define the \emph{decoupled coefficient}
\begin{align}
S_e\triangleq\sum_{\ell\in K\setminus\{i,j\}} \bar{\mathbf{A}}^{(2)}_{i\ell} \bar{\mathbf{A}}^{(3)}_{j\ell}.
\end{align}
Then we can rewrite
\begin{align}
\s{T}_{\calK}^{\mathrm{dec}}(\s{G})=\frac{1}{3}\sum_{e\subset \binom{\calK}{2}}\bar{\mathbf{A}}^{(1)}_e S_e.
\end{align}
Crucially, conditionally on the arrays $\{\mathbf{A}^{(2)}\},\{\mathbf{A}^{(3)}\}$, the family $\{\bar{\mathbf{A}}^{(1)}_e:e\subset \binom{\calK}{2}\}$ is independent, mean-zero, and each $\bar{\mathbf{A}}^{(1)}_e\in[-q,1-q]$, while the coefficients $\{S_e\}$ are deterministic constants. Therefore, by Hoeffding's lemma\footnote{For mean-zero $X\in[a,b]$: $\bE e^{\lambda X}\le\exp\{\lambda^2(b-a)^2/8\}$.} applied edgewise and then multiplied,
\begin{align}
\bE\pp{\left.\exp\p{\frac{C_{\mathrm{dec}}\theta}{3}\sum_{e\subset \binom{\calK}{2}} \bar{\mathbf{A}}^{(1)}_e S_e}\right| \mathbf{A}^{(2)},\mathbf{A}^{(3)}}\leq \exp\p{\frac{(C_{\mathrm{dec}}\theta)^2}{72}\sum_{e\subset \binom{\calK}{2}} S_e^2}.
\end{align}
Taking expectation over $(\mathbf{A}^{(2)},\mathbf{A}^{(3)})$ and using \eqref{eq:decouple},
\begin{align}\label{eq:mgf-dec-S}
\bE\pp{\exp\p{\theta \s{T}_{\calK}(\s{G})}}\leq \bE\pp{\exp\p{\frac{C_{\mathrm{dec}}^2 \theta^2}{72}\sum_{e\subset \binom{\calK}{2}} S_e^2}}.
\end{align}
For fixed $e=(i,j)$, the random sum $S_e=\sum_{\ell\ne i,j}\bar{\mathbf{A}}^{(2)}_{i\ell} \bar{\mathbf{A}}^{(3)}_{j\ell}$ is a sum of independent products of independent, centered, bounded Bernoulli variables; hence $S_e$ has the \emph{same distribution} as $W_e=\sum_{\ell\ne i,j}\bar{\mathbf{A}}_{i\ell} \bar{\mathbf{A}}_{j\ell}$ (the only difference is that the two factors in each product come from independent copies, which does not change the one-dimensional law, since $\bar{\mathbf{A}}_{i\ell}$ and $\bar{\mathbf{A}}_{j\ell}$ are independent already). Consequently, the random vectors $(S_e)_{e\subset \binom{\calK}{2}}$ and $(W_e)_{e\subset \binom{\calK}{2}}$ have the same joint law under the product Erd\H{o}s--R\'enyi measure on arrays $\{\mathbf{A}^{(2)}\}\times\{\mathbf{A}^{(3)}\}$ and on $\{\mathbf{A}\}$, respectively, \emph{up to relabeling of underlying independent coordinates}. In particular,
\begin{align}
\sum_{e\subset \binom{\calK}{2}} S_e^2\stackrel{d}{=}\sum_{e\subset \binom{\calK}{2}} W_e^2,
\end{align}
and the right-hand sides of \eqref{eq:mgf-dec-S} and \eqref{eq:mgf-main} coincide in distribution. This proves \eqref{eq:mgf-main}.
\end{proof}
Next, we have the following high probability upper bound on $\sum_{e\subset \binom{\calK}{2}} W_e^2$.
\begin{lemma}\label{lem:energy}
Let $\calK\subset[n]$ with $|\calK|=k\ge 3$, and for each $e=(i,j)\subset \binom{\calK}{2}$ define
\begin{align}
\s{W}_{e}=\sum_{\ell\in \calK\setminus\{i,j\}}({\mathbf{A}}_{i\ell}-q)({\mathbf{A}}_{j\ell}-q).
\end{align}
There exist constants $C=C(q)>0$ $C=C'(q)>0$ such that for sufficiently large $n$ and $k=\Omega(\log n)$
\begin{align}
\P\p{\sum_{e\in\binom{\calK}{2}} W_e^2> \alpha Ck^3\log n}\leq 3\exp\p{-\alpha C'k\log n},\label{eq:energy}
\end{align}
for any $\alpha>0$.
\end{lemma}
\begin{proof}[Proof of Lemma~\ref{lem:energy}]
%Recall that $\bar\mathbf{A}_{ij}=\mathbf{A}_{ij}-q$. 
For simplicity of notations, let $\mathbf{B}\in\mathbb R^{k\times k}$ be the symmetric matrix indexed by $\calK$ with 
\begin{align}
\mathbf{B}_{ij}\triangleq\mathbf{A}_{ij}-q,
\end{align}
for $i\neq j$ and $\mathbf{B}_{ij}=0$, for $i=j$. Now for $i\neq j$,
\begin{align}
[\mathbf{B}^2]_{ij}=\sum_{m\in \calK}\mathbf{B}_{im}\mathbf{B}_{mj}.
\end{align}
Since $\mathbf{B}_{ii}=\mathbf{B}_{jj}=0$ and $\mathbf{B}_{mj}=\mathbf{B}_{jm}$, this becomes
\begin{align}
[\mathbf{B}^2]_{ij}=\sum_{m\in \calK\setminus\{i,j\}}(\mathbf{A}_{im}-q)(\mathbf{A}_{jm}-q)=W_{ij}.
\end{align}
So $W_{ij}=[\mathbf{B}^2]_{ij}$ for all $i\neq j$. Therefore the ``energy"
\begin{align}
\calE \triangleq \sum_{e\subset \binom{\calK}{2}}W_e^2
\end{align}
satisfies
\begin{align}
\calE = \sum_{i<j}([\mathbf{B}^2]_{ij})^2= \frac{1}{2}\sum_{i\neq j}([\mathbf{B}^2]_{ij})^2
\leq \frac12\sum_{i,j}([\mathbf{B}^2]_{ij})^2
=\frac{1}{2}\|\mathbf{B}^2\|_F^2.
\end{align}
But $\|\mathbf{B}^2\|_F^2 = \s{trace}([\mathbf{B}^2]^\top[\mathbf{B}^2])=\s{trace}(\mathbf{B}^4)$ since $\mathbf{B}$ is symmetric. Hence
\begin{align}
\calE \leq \frac{1}{2}\s{trace}(\mathbf{B}^4)=\frac{1}{2}\|\mathbf{B}^2\|_F^2\leq \frac{1}{2}\|\mathbf{B}\|_{\mathrm{op}}^2\,\|\mathbf{B}\|_F^2,\label{eqn:Energgy}
\end{align}
where we use the Frobenius-operator norm inequality.

Next, we derive high probability upper bounds for the above Frobenius and norm operators, starting with the former. Note that for $i\neq j$,
\begin{align}
\mathbf{B}_{ij}\in[-q,1-q]\subset[-1,1],\quad \bE[\mathbf{B}_{ij}]=0,\quad \bE[\mathbf{B}_{ij}^2]=\mu,
\end{align}
where $\mu\triangleq q(1-q)$. We have
\begin{align}
\|\mathbf{B}\|_F^2 = \sum_{i,j}\mathbf{B}_{ij}^2 = 2\sum_{1\leq i<j\leq k}\mathbf{B}_{ij}^2
\triangleq 2\sum_{i<j}Z_{ij},
\end{align}
where $Z_{ij}\triangleq\mathbf{B}_{ij}^2\in[0,1]$, i.i.d. over $i<j$, and $\bE[Z_{ij}]=\mu$. Let $N\triangleq\binom{k}{2}$. Then $\sum_{i<j}Z_{ij}$ is a sum of $N$ i.i.d. $[0,1]$-bounded variables with mean $\mu$. Let us apply Bernstein's inequality to $X_r\triangleq Z_r-\mu$. We note that $|X_r|\leq 1$ and
$\s{Var}(X_r)\leq \bE[Z_r^2]\leq \bE[Z_r]=\mu$ since $0\le Z_r\le 1$. Thus, Bernstein's inequality yields
\begin{align}
\P\p{\sum_{i<j}Z_{ij} \geq 2N\mu}&=\P\p{\sum_{i<j}(Z_{ij}-\mu)\geq N\mu}\\
&\leq \exp\p{-\frac{(N\mu)^2}{2N\mu+\frac23 N\mu}}\\
&=\exp\p{-\frac{3}{8}N\mu}\\
&\leq \exp\p{-C\mu k^2}
\end{align}
for an absolute $C>0$. Therefore
\begin{align}
\P\p{\|\mathbf{B}\|_F^2 > 2\mu k^2}\leq \exp(-C\mu k^2).\label{eqn:ForbTail}
\end{align}
As for $\|\mathbf{B}\|_{\mathrm{op}}$ we can use classical results, e.g., \cite[Thm. 4.4.3]{vershynin2018high}, and get for any $\alpha>0$
\begin{align}
    \pr\p{\|\mathbf{B}\|_{\mathrm{op}}\geq  C\sqrt{\alpha k\log n}}\leq2\exp(-C'\alpha k\log n),\label{eqn:OpTail}
\end{align}
for universal constants $C,C'>0$. Combining \eqref{eqn:Energgy}, \eqref{eqn:ForbTail}, \eqref{eqn:OpTail}, and the union bound, we get
\begin{align}
\P\p{\sum_{e\subset \binom{\calK}{2}} W_e^2
> \alpha C''(q)k^3\log n}\leq e^{-C\mu k^2}  +  2e^{-C'\alpha k\log n}\leq 3e^{-C'\alpha k\log n},
\end{align}
for an constant $C''=C''(q)>0$, and in the last inequality we used that fact that $k=\Omega(\log n)$.
\end{proof}
Using the above lemmata we prove the following general tail bound on $\s{T}_{\calK}(\s{G})$.
\begin{prop}[Upper tail for $\s{T}_{\calK}(\s{G})$]\label{prop:tail}
Let $\s{T}_{\calK}(\s{G})$ be as above. Then for all $t>0$ and $u>0$,
\begin{align}
\pr\p{\s{T}_{\calK}(\s{G})\ge t,\sum_{e\subset K} W_e^2\le u}
\leq\exp\p{-\theta t + \frac{C_{\mathrm{dec}}^2 \theta^2}{72} u}
\qquad\text{for all }\theta>0.
\end{align}
Optimizing in $\theta$ gives
\begin{align}
\pr\p{\s{T}_{\calK}(\s{G})\geq t,\sum_{e\subset K} W_e^2\leq u}\leq\exp\p{- \frac{18 t^2}{C_{\mathrm{dec}}^2 u}}.
\end{align}
Consequently,
\begin{align}
\pr\p{\s{T}_{\calK}(\s{G})\ge t}\leq\exp\p{-\frac{18 t^2}{C_{\mathrm{dec}}^2 u}}+\pr\p{\sum_{e\subset K} W_e^2>u}.
\end{align}
\end{prop}

\begin{proof}[Proof of Proposition~\ref{prop:tail}]
By Markov's inequality and Lemma~\ref{lem:mgf-decoupled},
\begin{align}
\pr\p{\s{T}_{\calK}(\s{G})\ge t\sum_{e\subset K} W_e^2\le u}&=\bE\pp{\Ind\ppp{\s{T}_{\calK}(\s{G})\ge t,\sum_{e\subset K} W_e^2\leq u}}\\
&\leq e^{-\theta t} \bE\pp{e^{\theta \s{T}_{\calK}(\s{G})}\Ind\ppp{\sum_{e\subset K} W_e^2\leq u}}\\
&\leq \exp\p{-\theta t + \tfrac{C_{\mathrm{dec}}^2 \theta^2}{72} u}.
\end{align}
The choice $\theta^\star = 36 t/(C_{\mathrm{dec}}^2 u)$ minimizes the exponent and yields the result.
\end{proof}

Finally, we are in a position to prove Lemma~\ref{lem:probabScanUpper}. Specifically, apply Proposition~\ref{prop:tail} with $u= \alpha C k^3\log n$ from Lemma~\ref{lem:energy}. Then
\begin{align}
\pr\p{\s{T}_{\calK}(\s{G})\ge \tau_{\s{scan}}}\leq\exp\p{-\frac{18 \tau_{\s{scan}}^2}{C_{\mathrm{dec}}^2 \alpha C  k^3\log n}}+ 3\exp\p{-\alpha C'k\log n}.
\end{align}
For $\tau_{\s{scan}}\ge c_0 k^3/\sqrt d$ and $k\ge 3$, the first term is $\leq \exp\{-c k^3/(\alpha d\log n)\}$ for some $c=c(q)>0$. This concludes the proof.

%%%%%%%%%%%%%%%%%%%%%%%%%%%%%%%%%%

\section{Proofs of Lower Bounds}

\subsection{Information-Theoretic lower bounds}
In this subsection we prove Theorem~\ref{thm:ITlower}. \cite{bubeck2016testing} can be applied.
Our aim is to derive a lower bound on the optimal risk, thereby precluding the possibility of successful detection. Define the likelihood ratio as
\begin{align}
    \calL(\bs{A})\triangleq\frac{\mathrm{d}\P_{\calH_1}}{\mathrm{d}\P_{\calH_0}}(\bs{A}),
    \label{eqn:LIKlihood}
\end{align}
namely, the Radon--Nikodym derivative of $\P_{\calH_1}$ relative to the measure $\P_{\calH_0}$, and we recall that $\bs{A}$ denotes the graph adjacency matrix. It is classical (see, e.g.,~\cite[Theorem~2.2]{tsybakov2004introduction}) that the test minimizing the risk is the likelihood ratio test
\begin{align}
\calA^\ast(\bs{A}) \triangleq \Ind\{\calL(\bs{A}) \geq 1\},
\end{align}
and that the corresponding optimal risk satisfies
\begin{align}
\s{R}^\ast
= 1-d_{\TV}(\P_{\calH_0}, \P_{\calH_1}).
\end{align}
Recalling that the chi-square divergence admits the representation $\chi^2(\P_{\calH_1}\|\P_{\calH_0})
= \E_{\calH_0}[\calL^2] - 1$, it follows from standard inequalities relating total variation and chi-square divergences (see, e.g.,~\cite[Sec.~2]{tsybakov2004introduction}, \cite[Prop.~3]{sason2014bounds}) that
\begin{align}
\chi^2(\P_{\calH_1}\|\P_{\calH_0})\geq\max\left\{\frac{1}{2(1-d_{\TV}(\P_{\calH_0},\P_{\calH_1}))}-1,\bigl(2d_{\TV}(\P_{\calH_0},\P_{\calH_1})\bigr)^2\right\}.
\end{align}
Consequently, the optimal risk admits the lower bound
\begin{align}
\s{R}^\ast\geq\max\left\{1-\frac{1}{2}\sqrt{\chi^2(\P_{\calH_1}\|\P_{\calH_0})},\frac{1}{2\bigl(1+\chi^2(\P_{\calH_1}\|\P_{\calH_0})\bigr)}\right\}.
\label{eqn:lowerBoundSecond}
\end{align}
In particular, the optimal risk remains bounded away from zero whenever $\E_{\calH_0}[\calL^2]$ is bounded, and converges to one if $\E_{\calH_0}[\calL^2]=1+o(1)$. Therefore, to rule out detection it suffices to control the second moment of the likelihood ratio under $\calH_0$. To conclude
\begin{align}
        &\bE_{\calH_0}[\calL^2]=1+o(1)\implies d_{\s{TV}}(\P_{\calH_0}\|\P_{\calH_1})=o(1),\label{eqn:weakDet}\\
    &\bE_{\calH_0}[\calL^2]=O(1)\implies d_{\s{TV}}(\P_{\calH_0}\|\P_{\calH_1})\leq1-\Omega(1).\label{eqn:strongDet}
\end{align}

Let us introduce some notation. Let $\bs{X}$ denote the $n\times d$ matrix formed by stacking $n$ i.i.d. standard Gaussian vectors, and let $\bs{W}=\bs{X}\bs{X}^T$ be the associated $n\times n$ Wishart matrix. Observe that $\bs{W}_{ii}=\|\bs{x}_i\|^2$, and that
\begin{align}
\big\langle \bs{x}_i/\|\bs{x}_i\|,\ \bs{x}_j/\|\bs{x}_j\| \big\rangle
= \bs{W}_{ij}/\sqrt{\bs{W}_{ii}\bs{W}_{jj}}.
\end{align}
Define
\begin{align}
\sigma_{ij} \triangleq \Ind\{\langle \bs{x}_i,\bs{x}_j\rangle \ge t_{q,d}\},
\end{align}
or equivalently,
\begin{align}
\sigma_{ij}= \Ind\left\{\bs{W}_{ij}/\sqrt{\bs{W}_{ii}\bs{W}_{jj}} \ge t_{q,d}\right\}.
\end{align}
For $\calK\in\binom{[n]}{k}$, we write $E(\calK)=\{\{i,j\}: i,j\in\calK,\ i\neq j\}$ and $E^c(\calK)=E([n])\setminus E(\calK)$. With these notations in place, we may write the likelihood ratio explicitly. Recall the detection problem in \eqref{eqn:model}. Under the null hypothesis
\begin{align}
    \P_{\calH_0}(\bs{A}) = \prod\limits_{i<j} q^{\bs{A}_{i,j}}(1 - q)^{1 - \bs{A}_{i,j}},\label{eqn:H0Measure}
\end{align}
while under the alternative’
\begin{align}
    \P_{\calH_1\vert \calK} \p{ \bs{A}\vert \calK} 
    &= \prod\limits_{\{i,j\} \in E^c\p{\calK}} 
        q^{\bs{A}_{i,j}} (1 - q)^{1 - \bs{A}_{i,j}} 
        \cdot \E_{\bs{W}\vert \calK} 
            \pp{
                \prod\limits_{\{i,j\} \in E\p{\calK}} 
                \sigma_{ij}^{\bs{A}_{i,j}}
                \p{1 - \sigma_{ij}} ^ {1 - \bs{A}_{i,j}}
            }.\label{eqn:H1Measure}
\end{align}
It follows that the likelihood ratio takes the form
\begin{align}
\label{eqn:L-exenpded}
    \calL(\bs{A})=\frac
        {\E_\calK\P_{\calH_1\vert \calK} (\mathbf{A}\vert \calK)}
        {\P_{\calH_0} (\mathbf{A})}
        &= \E_\calK\E_{\bs{W}\vert \calK}
            \pp{
                \prod\limits_{\{i,j\} \in E(\calK)}
                \p{
                    \frac{\sigma_{ij}}{q}
                } ^ {\bs{A}_{i,j}}
                \p{
                    \frac{1 - \sigma_{ij}}{1 - q}
                } ^ {1 - \bs{A}_{i,j}}
            }.
\end{align}
In this expression, the outer expectation over $\calK$ is uniform over $\binom{[n]}{k}$, and the inner expectation is with respect to the Wishart distribution. 

As discussed above, our objective is to upper bound the second moment of the likelihood ratio under $\calH_0$. Nevertheless, in the present setting, certain rare events under $\calH_1$ may cause this second moment to diverge, even though $d_{\s{TV}}(\pr_{\calH_0},\pr_{\calH_1})$ remains bounded away from one. To circumvent the effect of such atypical events, we instead analyze the second moment conditioned on events that are typical under $\calH_1$. This refined approach, originally proposed by \cite{Ingster1997HypothesisTesting,butucea2013detection}, is based on controlling the first and second moments of a truncated likelihood ratio \cite{arias2014community,WuXuYu2023}.

We begin by outlining the truncated second-moment method in general. Let $\calF$ be an event such that $\pr(\calF) = 1-o(1)$. Define the \emph{truncated/conditional planted model} as
\begin{align}
    \tilde{\pr}_{\calH_1,\calK,\bs{W}}(\bs{A},\calK,\bs{W}) &= \frac{\pr_{\calH_1,\calK,\bs{W}}(\bs{A},\calK,\bs{W})\Ind\ppp{(\calK,\bs{W})\in\calF}}{\pr(\calF)}\nonumber\\
    &= (1+o(1))\cdot\pr_{\calH_1,\calK,\bs{W}}(\bs{A},\calK,\bs{W})\Ind\ppp{(\calK,\bs{W})\in\calF},\label{eqn:truncatedModel}
\end{align}
and note that this is a legitimate probability measure. Then, define the \emph{truncated likelihood ratio} as 
\begin{align}
    \tilde{\calL}(\bs{A})\triangleq\frac{\tilde{\pr}_{\calH_1}(\bs{A})}{\pr_{\calH_0}(\bs{A})} &= \frac{1}{\pr(\calF)}\bE_{\calK,\bs{W}}\pp{\frac{\pr_{\calH_1\vert\calK,\bs{W}}(\bs{A}\vert\calK,\bs{W})\Ind\ppp{(\calK,\bs{W})\in\calF}}{{\pr_{\calH_0}(\bs{A})}}}\label{eqn:trunLikeGen}\\
    & = (1+o(1))\cdot\bE_{\calK,\bs{W}}\pp{\frac{\pr_{\calH_1\vert\calK,\bs{W}}(\bs{A}\vert\calK,\bs{W})\Ind\ppp{(\calK,\bs{W})\in\calF}}{\pr_{\calH_0}(\bs{A})}}.
\end{align}
Now, by the data processing inequality for the total variation distance, we know that
\begin{align}
    d_{\s{TV}}(\pr_{\calH_1},\tilde{\pr}_{\calH_1})&\leq d_{\s{TV}}(\pr_{\calH_1,\calK,\bs{W}},\tilde{\pr}_{\calH_1,\calK,\bs{W}})= \pr\pp{\calF^c}=o(1).\label{eqn:highprobF}
\end{align}
Accordingly, combining \eqref{eqn:weakDet}--\eqref{eqn:strongDet} with \eqref{eqn:highprobF}, the triangle inequality implies that
\begin{align}
        &\bE_{\calH_0}[\tilde{\calL}^2]=1+o(1)\implies d_{\s{TV}}(\P_{\calH_0}\|\P_{\calH_1})=o(1),\label{eqn:weakDetTrun}\\
    &\bE_{\calH_0}[\tilde{\calL}^2]=O(1)\implies d_{\s{TV}}(\P_{\calH_0}\|\P_{\calH_1})\leq1-\Omega(1).\label{eqn:strongDetTrun}
\end{align}
Therefore, with a carefully chosen high probable truncation set $\calF$, it suffices to analyze
the second moment of the truncated likelihood ratio. 
\begin{comment}
    Let $\tilde{\calL}_n$ denote a truncated likelihood ratio satisfying $\tilde{\calL}_n \le \calL_n$, to be specified shortly. By the triangle inequality and the Cauchy--Schwarz inequality, we have
\begin{align}
\mathbb{E}_{\calH_0}|\calL-1|
&\le
\mathbb{E}_{\calH_0}|\tilde{\calL}-1|
+
\mathbb{E}_{\calH_0}\bigl(\calL-\tilde{\calL}\bigr) \\
&\le
\sqrt{
\mathbb{E}_{\calH_0}[\tilde{\calL}^2]
- 1
+ 2\bigl(1-\mathbb{E}_{\calH_0}[\tilde{\calL}]\bigr)
}
+
\bigl(1-\mathbb{E}_{\calH_0}[\tilde{\calL}]\bigr).
\end{align}
Consequently, if $\mathbb{E}_{\calH_0}[\tilde{\calL}]\to 1$ and $\mathbb{E}_{\calH_0}[\tilde{\calL}^{2}]\to 1$, then $\s{R}^\ast \to 1$, and weak detection is impossible. On the other hand, if both quantities remain bounded, then $\s{R}^\ast$ is bounded away from zero, and strong detection is impossible. 
\end{comment}

We now proceed to define the truncation set and likelihood ratio in our setting. For any given $\calK\subseteq[n]$, define the truncation set as
\begin{align}\label{eqn:truncationset1}
    \Gamma_\calK\triangleq\bigcap_{z=2}^{|\calK|}\ppp{\Gamma^{\s{fro}}_{z,\calK}\cap\Gamma^{\s{op}}_{z,\calK}},
\end{align}
where for $z\geq 2$
\begin{align}\label{eqn:truncationset2}
    \Gamma^{\s{fro}}_{z,\calK}\triangleq\bigcap_{T\subseteq \calK,\;|T|=z}\bigcap_{i\in\calT}\ppp{\sum_{j\in T\setminus\{i\}}A_{ij}^2\leq L^{\s{fro}}_{k,z,d}},
\end{align}
and
\begin{align}\label{eqn:truncationset3}
    \Gamma^{\s{op}}_{z,\calK}\triangleq\bigcap_{T\subseteq \calK,\;|T|=z}\ppp{\norm{A|_T-d\bs{I}_z}_{\s{op}}\leq L^{\s{op}}_{k,z,d}},
\end{align}
where we recall that $A|_T$ is the restriction of $A$ to the indices in $T$, $L^{\s{fro}}_{k,z,d}\triangleq(1+C_1)d(z-1)\log k$, and $L^{\s{op}}_{k,z,d}\triangleq C_2(\sqrt{dz}+\sqrt{d\log \binom{k}{z}})$, for some $C_1,C_2>0$, specified later on. Then, using \eqref{eqn:H0Measure}--\eqref{eqn:L-exenpded} and \eqref{eqn:truncatedModel}--\eqref{eqn:trunLikeGen}, we see that the truncated planted model on $\Gamma_\calK$ is given by
\begin{comment}
    \footnote{
If $X$ is a random object with law $\P_X$ and $B$ is a measurable event,
we define the truncated law
\begin{align}
\tilde P_X(A) \triangleq P_X(A\cap B).
\end{align}
Thus $\tilde P_X$ is the restriction of $P_X$ to $B$. We continue to use
the terms “law”, “distribution”, and “probability” for such truncated
measures, even though they need not be normalized. If $Y=h(X)$ for a
measurable map $h$, we write $P_Y$ and $\tilde\P_Y$ for the
(pushforward) laws of $Y$ and its truncation. 
}
\end{comment}
\begin{align}
    \tilde{\pr}_{\calH_1\vert \calK} \p{ \bs{A}\vert \calK} 
    &= \frac{1}{\pr[\Gamma_{\calK}]}\prod\limits_{\{i,j\} \in E^c\p{\calK}} 
        q^{\bs{A}_{i,j}} (1 - q)^{1 - \bs{A}_{i,j}} 
        \nonumber\\
        &\qquad\qquad\cdot\E_{\bs{W}\vert \calK} 
            \pp{
                \prod\limits_{\{i,j\} \in E\p{\calK}} 
                \sigma_{ij}^{\bs{A}_{i,j}}
                \p{1 - \sigma_{ij}} ^ {1 - \bs{A}_{i,j}},
            \mathds{1}_{\Gamma_\calK}(\bs{W})},\label{eqn:TruncatedMoment0}
\end{align}
and thus
\begin{align}
\tilde{\calL}(\bs{A})=\frac{1}{\pr[\Gamma_{\calK}]}\E_\calK\E_{\bs{W}\vert \calK}
            \pp{
                \prod\limits_{\{i,j\} \in E(\calK)}
                \p{
                    \frac{\sigma_{ij}}{q}
                } ^ {\bs{A}_{i,j}}
                \p{
                    \frac{1 - \sigma_{ij}}{1 - q}
                } ^ {1 - \bs{A}_{i,j}}\mathds{1}_{\Gamma_\calK}(\bs{W})
            }.\label{eqn:TruncatedMoment}
\end{align}
We also denote the conditional likelihood
\begin{align}
    \tilde{\calL}_{\calK}(\bs{A}\vert\calK)&\triangleq\frac{\tilde{\pr}_{\calH_1\vert \calK} \p{ \bs{A}\vert \calK}}{\pr_{\calH_0}(\bs{A})}\\
    &= \frac{\tilde{\pr}_{\calH_1\vert \calK} \p{ \bs{A}_{\calK}}}{\pr_{\calH_0}(\bs{A}_{\calK})}\\
    &=\frac{1}{\pr[\Gamma_{\calK}]}\E_{\bs{W}\vert \calK}
            \pp{
                \prod\limits_{\{i,j\} \in E(\calK)}
                \p{
                    \frac{\sigma_{ij}}{q}
                } ^ {\bs{A}_{i,j}}
                \p{
                    \frac{1 - \sigma_{ij}}{1 - q}
                } ^ {1 - \bs{A}_{i,j}}\mathds{1}_{\Gamma_\calK}(\bs{W})
            },\label{eqn:TruncConcdLike}
\end{align}
and we note that $\tilde{\calL}(\bs{A}) = \bE_{\calK}\pp{\tilde{\calL}_{\calK}(\bs{A}\vert\calK)}$. 
Next, we show that the truncation set defined above has high probability, and then derive conditions under which $\bE_{\calH_0}[\tilde{\calL}^2]$ is bounded or converges to unity.

\subsubsection{High probability truncation set}
\begin{comment}
    We begin by proving that $\mathbb{E}_{\calH_0}[\tilde{\calL}]\to1$. Using Fubini's theorem, we have
\begin{align}
    \E_{\calH_0}[\tilde{\calL}]
    &= 
    \E_{\calK}
    \E_{\bs{W}\mid\calK}
    \left[
        \prod_{\{i,j\}\in E(\calK)}
        \E_{\calH_0}\left[
            \left(\frac{\sigma_{ij}}{q}\right)^{\bs{A}_{ij}}
            \left(\frac{1-\sigma_{ij}}{1-q}\right)^{1-\bs{A}_{ij}}
        \right]
        \mathds{1}_{\Gamma_{\calK}}(\bs{W})
    \right] \\
    &=
    \E_{\calK}
    \E_{\bs{W}\mid\calK}
    \left[
        \mathds{1}_{\Gamma_{\calK}}(\bs{W})
    \right].
\end{align}
\end{comment}
We show that $\pr[\Gamma_{\calK}]=1-o(1)$. We have
\begin{align}
    \pr[\Gamma_{\calK}] = \bE_{\calK}\pp{\P_{\bs{W}\vert\calK}[\Gamma_{\calK}]} = \P_{\bs{W}}[\Gamma_{\calK_0}]
\end{align}
where $\calK_0$ is any fixed subset of size $k$ in $[n]$, the last equality is because $\P_{\bs{W}\vert\calK}[\Gamma_{\calK}]$ does not depend on the particular choice of $\calK$ by symmetry, and $\P_{\bs{W}}$ denotes the Wishart distribution. By the union bound we have
\begin{align}
   \P_{\bs{W}}(\Gamma_{\calK_0}^c)\leq \P_{\bs{W}}\pp{\p{\Gamma_{\calK_0}^{\s{fro}}}^c}+\P_{\bs{W}}\pp{\p{\Gamma_{\calK_0}^{\s{op}}}^c},\label{eqn:twoProbToUpperBound}
\end{align}
where $\Gamma_{\calK_0}^{\s{fro}}\triangleq\cap_{z=2}^k\Gamma^{\s{fro}}_{z,\calK}$ and $\Gamma_{\calK_0}^{\s{op}}\triangleq\cap_{z=2}^k\Gamma^{\s{op}}_{z,\calK}$, and we next prove that both terms at the right-hand side of \eqref{eqn:twoProbToUpperBound} converge to zero as $n\to\infty$. 

We start with $\P_{\bs{W}}\pp{\p{\Gamma_{\calK_0}^{\s{fro}}}^c}$. Fix $z\in[k]$ and $T\subseteq\calK_0$ with $|T|=z$. Recall that $\bs{W}_{ij}=\langle \bs{x}_i,\bs{x}_j\rangle$, where $\bs{x}_1,\dots,\bs{x}_n$ are i.i.d. $d$-dimensional standard Gaussian vectors. Conditioning on $\bs{x}_i = x$, we have for any $j\neq i$,
\begin{align}
\langle x,\bs{x}_j\rangle
=
\sum_{\ell=1}^d x_\ell \bs{x}_{j\ell}
\sim \mathcal{N}(0,\|x\|^{2}),
\end{align}
and these random variables are independent across $j$. Consequently, given
$\bs{x}_i = x$,
\begin{align}
\sum_{j\in T\setminus\{i\}} \bs{W}_{ij}^{2}
\overset{d}{=}
\|x\|^{2}\cdot\s{A},
\end{align}
where $\s{A}\sim\chi^{2}_{z-1}$. Since $\|\bs{x}_i\|^{2}\sim\chi^{2}_{d}$ and
$\bs{x}_i$ is independent of $\{\bs{x}_j\}_{j\neq i}$, it follows that, unconditionally,
\begin{align}
\sum_{j\in T\setminus\{i\}} \bs{W}_{ij}^{2}
\overset{d}{=}
\s{A}\cdot\s{B},
\end{align}
where $\s{B}\sim\chi^{2}_{d}$, and $\s{A}\independent\s{B}$. Recall that for a chi-square random variable $Y\sim\chi^2_m$, the following standard Chernoff bound holds (see, e.g.,~\cite[Ch.~2]{BoucheronLugosiMassart2013}):
\begin{align}
\pr\left(Y \ge (1+\varepsilon)m\right)
\leq 
\exp\left(-\frac{m}{2}\bigl(\varepsilon - \log(1+\varepsilon)\bigr)\right),
\label{eqn:standardChrenoffChi0}
\end{align}
for any $\varepsilon>0$. Furthermore, note that $\varepsilon - \log(1+\varepsilon)\geq \frac{\varepsilon^2}{2(1+\varepsilon)}$, and that for $\varepsilon\geq 1$ we have $\frac{\varepsilon^2}{2(1+\varepsilon)}\geq \frac{\varepsilon}{2}$. Thus, in this regime, \eqref{eqn:standardChrenoffChi0} simplifies to
\begin{align}
\pr\left(Y \ge (1+\varepsilon)m\right)
\leq 
\exp\left(-\frac{m\varepsilon}{8}\right).
\label{eqn:standardChrenoffChi}
\end{align}
Therefore, applying the union bound together with \eqref{eqn:standardChrenoffChi}, we obtain
\begin{align}
\P_{\bs{W}}\pp{\p{\Gamma_{\calK_0}^{\s{fro}}}^c}
&\leq \sum_{z=2}^{k}\sum_{\substack{T\subseteq\calK_0\\ |T|=z}}\sum_{i\in T}\mathbb{P}\left(\s{A}\cdot\s{B} \ge L^{\s{fro}}_{k,z,d} \right)\\
&\leq \sum_{z=2}^{k} z\binom{k}{z} \cdot\pp{\mathbb{P}\left( \s{A} \ge (1+C_1')(z-1)\log k \right)+\mathbb{P}\left( \s{B} \geq (1+C_1'')d \right)} \\
&\leq \sum_{z=2}^{k}\exp\left[z\left(2\log\left(\frac{ek}{z}\right)-C_1'\log k\right)\right] + \sum_{z=2}^{k} \exp\left[
z\left( \log\left(\frac{ek}{z}\right)-C_{1}''\frac{d}{8z}\right)\right]\\
&\leq \sum_{z=2}^{k}\exp\left[z\left(2\log (ek)-C_1'\log k\right)\right] + \sum_{z=2}^{k} \exp\left[
z\left( \log (ek)-C_{1}''\frac{d}{8k}\right)\right].
\label{eq:final_prob_bound_1}
\end{align}
where the second inequality follows from the trivial inclusion $\{\s{A}\cdot\s{B} \ge (1+C_1)d(z-1)\log k\}\subseteq \{\s{A} \ge (1+C_1')(z-1)\log k\}\cup \{\s{B} \ge (1+C_1'')d\}$ where $C_1',C_1''>0$ are such that $(1+C_1')(1+C_1'')\leq(1+C_1)$. Now, define $\zeta\triangleq\max\ppp{2\log (ek)-C_{1}'\log k,\log (ek)-C_1''\frac{d}{8k}}$. Then, for  $C_1'>2$ (achieved by choosing $C_1$ sufficiently large) the first term in the maximum converges to $-\infty$ as $k\to\infty$. Furthermore, in the impossibility regime of Theorem~\ref{thm:ITlower}, the condition $d/(k^{2}\log^2k)\to\infty$ clearly implies that the second term in the maximum also converges to $-\infty$. Thus, $\zeta\to-\infty$, and
\begin{align}
    \mathbb{P}_{\bs{W}}((\Gamma^1_{\calK_0})^{c})\leq 2\sum_{z=2}^\infty e^{z\zeta}\leq 2\frac{e^{\zeta}}{1-e^{\zeta}}\to0,\label{eqn:probOp0}
\end{align}
eventually.

Next, we bound $\P_{\bs{W}}\pp{\p{\Gamma_{\calK_0}^{\s{op}}}^c}$. By the union bound,
\begin{align}
  \P_{\bs{W}}\pp{\p{\Gamma_{\calK_0}^{\s{op}}}^c}&\leq \sum_{z=2}^{k}\sum_{\substack{T\subseteq\calK_0\\ |T|=z}}\pr\pp{\norm{\bs{W}|_T-d\bs{I}_z}_{\s{op}}> L^{\s{op}}_{k,z,d}}\\
  &\leq \sum_{z=2}^{k}\binom{k}{z}\pr\pp{\norm{\bs{W}|_T-d\bs{I}_z}_{\s{op}}> L^{\s{op}}_{k,z,d}}\\
  &\leq \sum_{z=2}^{k}e^{(1-C_2')\log \binom{k}{z}},\label{eqn:probOp}
  \end{align}
where in the last inequality we applied standard concentration inequalities for sample covariance matrices (see, e.g.,~\cite[Remark~4.7.3]{vershynin2018high}. Choosing $C_2'>1$ (which can be ensured by taking $C_2$ sufficiently large) makes the exponent $(1-C_2')\log\binom{k}{z}$ negative for every $z$, and therefore the sum in \eqref{eqn:probOp} converges to zero. Combining \eqref{eqn:twoProbToUpperBound}, \eqref{eqn:probOp0}, and \eqref{eqn:probOp}, we conclude that $\pr[\Gamma_{\calK}]=1-o(1)$, as required.

\begin{comment}
    \begin{lemma}[\cite{vershynin2011} Corollary 3.35]
    \label{lem:gaussian-matrix-concentration}
    Let $\bs{X}$ be a $d\times z$ matrix with iid gaussian entries. Then for any $t\ge 0$, with probability at least $1-2\exp(t^2/2)$ one has for all $i=1,...,\min\{d,z\}$,
\begin{align}
    \sqrt{d}-\sqrt{z}-t\le s_i(\bs{X})\le\sqrt{d}+\sqrt{z}+t,
\end{align}
where $s_i$ is the i'th singular value of $\bs{X}$.
\end{lemma}
From the above lemma we can deduce that with probability $1-2\exp(C_2z\log k/2)$ the eigenvalues of $\bs{W}=\bs{X}\bs{X}^T$ satisfy 
\begin{align}
    d-2\sqrt{C_2dz\log k}+C_2\log k\le \lambda_i(\bs{W})\le d+2\sqrt{C_2dz\log k}+C_2\log k,
\end{align}
and since $d\ge k\ge z$ we can infer that this event is contained in  $\abs{\bs{W}-d}\le(1+C_2)\sqrt{dz\log k}$. Hence,
\begin{align}
\P_{\bs{W}}\pp{\p{\Gamma_{\calK_0}^{\s{op}}}^c}
&\leq \sum_{z=2}^{k}\sum_{\substack{T\subseteq\calK_0\\ |T|=z}}\sum_{i\in T}\mathbb{P}\left(\abs{\lambda_i(\bs{W}|_T)-d}\ge (1+C_2)\sqrt{dz\log k} \right)\\
&\leq 2\sum_{z=2}^{k} z\binom{k}{z}\exp\left(\frac{C_2z\log k}{2}\right) \\
&\leq \sum_{z=2}^{k}\exp\left[z\left(\log\left(\frac{ke}{z}\right)-C_{6}\log k\right)\right]\\
&\leq \sum_{z=2}^{k}\exp\left[z\left(\log(ek)-C_{6}\log k\right)\right].
\label{eq:final_prob_bound_2}
\end{align}
The same considerations as in (\ref{eq:final_prob_bound_1}) show that $\mathbb{P}_{\bs{W}}((\Gamma^2_{\calK_0})^{c})\to0$.
\end{comment}

\subsubsection{Truncated second moment}

We begin by deriving a simplified formula for the truncated second moment $\E_{\calH_0}[\tilde{\calL}^{2}]$, and then derive conditions under which  $\E_{\calH_0}[\tilde{\calL}^{2}] = O(1)$ and $\E_{\calH_0}[\tilde{\calL}^{2}] = 1+o(1)$. Let $\calK$ and $\bar\calK$ be two independent copies drawn uniformly at random over $\binom{[n]}{k}$, i.e., $\calK\independent\bar{\calK}\overset{\s{i.i.d.}}{\sim}\s{Unif}\left(\binom{[n]}{k}\right)$. Similarly let  $\bs{W}\independent\bar{\bs{W}}\overset{\s{iid}}{\sim}\mu_{\bs{W_k}}$ be two independent copies of a Wishart matrix of size $n\times n$.  Recall \eqref{eqn:TruncatedMoment}, and define
\begin{align}
    \theta(\bs{A}_{ij},\bs{W}_{ij})\triangleq\p{
                    \frac{\sigma_{ij}}{q}
                } ^ {\bs{A}_{i,j}}
                \p{
                    \frac{1 - \sigma_{ij}}{1 - q}
                } ^ {1 - \bs{A}_{i,j}},
\end{align}
for $(i,j)\in\binom{[n]}{2}$. Below we denote by $\mathbb{E}_{(\calK,\bar\calK,\bs{W},\bar{\bs{W}})}$ the expectation with respect to the distribution $\mathbb{P}_\calK\times\mathbb{P}_{\bar{\calK}}\times\mathbb{P}_{\bs{W}\independent\bar{\bs{W}}\vert(\calK,\bar\calK)}$, as defined above. Also, let $\bar\sigma_{ij}$ denote the analogue of $\sigma_{ij}$ constructed from the independent copy $\bar{\bs{W}}$, for all $(i,j)\in\binom{[n]}{2}$. Then, by Fubini's theorem, we have
\begin{align}
    \mathbb{E}_{\mathcal{H}_0}\bigl[\tilde{\calL}^{2}\bigr] &= \frac{1}{(\pr[\Gamma_{\calK}])^2}\mathbb{E}_{(\calK,\bar\calK,\bs{W},\bar{\bs{W}})}\bE_{\calH_0}\left[\prod\limits_{\{i,j\} \in E(\calK)}
               \theta(\bs{A}_{ij},\bs{W}_{ij}) \mathds{1}_{\Gamma_{\calK}}(\bs{W})\right.\nonumber\\
               &\hspace{5.5cm}\left.\prod\limits_{\{i,j\} \in E(\bar\calK)}
                \theta(\bs{A}_{ij},\bar{\bs{W}}_{ij})
          \mathds{1}_{\Gamma_{\bar\calK}}(\bar{\bs{W}})\right]\nonumber\\
          & = \frac{1}{(\pr[\Gamma_{\calK}])^2}\mathbb{E}_{(\calK,\bar\calK,\bs{W},\bar{\bs{W}})}\left[\mathds{1}_{\Gamma_{\calK}}(\bs{W})
          \mathds{1}_{\Gamma_{\bar\calK}}(\bar{\bs{W}})\eta(\calK,\bar\calK,\bs{W},\bar{\bs{W}})\right],\label{eqn:eta0}
\end{align}
where
\begin{align}
    \eta(\calK,\bar\calK,\bs{W},\bar{\bs{W}})\triangleq\bE_{\calH_0}\pp{\prod\limits_{\{i,j\} \in E(\calK)}
               \theta(\bs{A}_{ij},\bs{W}_{ij})\prod\limits_{\{i,j\} \in E(\bar\calK)}
                \theta(\bs{A}_{ij},\bar{\bs{W}}_{ij})}.
\end{align}
If $|\calK\cap\bar\calK|\leq1$, namely, at most one vertex shared between $\calK$ and $\bar\calK$, then the set of edges formed among nodes in $\calK$ are disjoint from the set of edges formed among nodes in $\bar\calK$. Using the fact that $\bE_{\calH_0}[\theta(\bs{A}_{ij},\bs{W}_{ij})]=1$,  for all $(i,j)\in\binom{[n]}{2}$, we get that
\begin{align}
    \eta(\calK,\bar\calK,\bs{W},\bar{\bs{W}}) = \Ind\ppp{|\calK\cap\bar\calK|\leq1}.\label{eqn:eta1}
\end{align}
Otherwise, if $|\calK\cap\bar\calK|\geq2$, we have
\begin{align}
    \eta(\calK,\bar\calK,\bs{W},\bar{\bs{W}}) &= \bE_{\calH_0}\left[\Ind\ppp{|\calK\cap\bar\calK|\geq2}\prod\limits_{\{i,j\} \in E(\calK\cap\bar\calK)}
               \theta(\bs{A}_{ij},\bs{W}_{ij})\theta(\bs{A}_{ij},\bar{\bs{W}}_{ij})\right.\nonumber\\
               &\qquad\qquad\quad\left.\prod\limits_{\{i,j\} \in E(\calK)\setminus E(\calK\cap\bar\calK)}
                \theta(\bs{A}_{ij},{\bs{W}}_{ij})\prod\limits_{\{i,j\} \in E(\bar\calK)\setminus E(\calK\cap\bar\calK)}
                \theta(\bs{A}_{ij},\bar{\bs{W}}_{ij})\right]\\
                & = \bE_{\calH_0}\left[\prod\limits_{\{i,j\} \in E(\calK\cap\bar\calK)}
               \theta(\bs{A}_{ij},\bs{W}_{ij})\theta(\bs{A}_{ij},\bar{\bs{W}}_{ij}) \Ind\ppp{|\calK\cap\bar\calK|\geq2}\right],\label{eqn:eta2}
\end{align}
where the the second equality follows from the fact that $\bE_{\calH_0}[\theta(\bs{A}_{ij},\bs{W}_{ij})]=1$ and $\bE_{\calH_0}[\theta(\bs{A}_{ij},\bar{\bs{W}}_{ij})]=1$ for all $(i,j)\in\binom{[n]}{2}$. Combining \eqref{eqn:eta0}, \eqref{eqn:eta1}, and \eqref{eqn:eta2}, we get
\begin{align}
    \mathbb{E}_{\mathcal{H}_0}\bigl[\tilde{\calL}^{2}\bigr] &= \frac{\pr\pp{\bs{W}\in\Gamma_{\calK},\bar{\bs{W}}\in\Gamma_{\bar\calK},|\calK\cap\bar\calK|\leq1}}{(\pr[\Gamma_{\calK}])^2}\nonumber\\
    &\qquad+\mathbb{E}_{\calK\independent\bar\calK}\bE_{\calH_0}\left[g(\bs{A},\calK,\bar\calK)\Ind\ppp{|\calK\cap\bar\calK|\geq2}\right]\\
    & = \pr\pp{\left.|\calK\cap\bar\calK|\leq1\right|\bs{W}\in\Gamma_{\calK},\bar{\bs{W}}\in\Gamma_{\bar\calK}}\nonumber\\
    &\qquad+\mathbb{E}_{\calK\independent\bar\calK}\bE_{\calH_0}\left[g(\bs{A},\calK,\bar\calK)\Ind\ppp{|\calK\cap\bar\calK|\geq2}\right],
\end{align}
where
\begin{align}
    g(\bs{A},\calK,\bar\calK)&\triangleq\frac{1}{(\pr[\Gamma_{\calK}])^2}\bE_{\bs{W}\independent\bar{\bs{W}}\vert\calK,\bar\calK}\left(\prod\limits_{\{i,j\} \in E(\calK\cap\bar\calK)}
               \theta(\bs{A}_{ij},\bs{W}_{ij})\mathds{1}_{\Gamma_{\calK}}(\bs{W})\right.\nonumber\\
               &\left.\hspace{5.2cm}\prod\limits_{\{i,j\} \in E(\calK\cap\bar\calK)}
               \theta(\bs{A}_{ij},\bar{\bs{W}}_{ij})
          \mathds{1}_{\Gamma_{\bar\calK}}(\bar{\bs{W}})\right).
\end{align}
Recall the definition of $\Gamma_{\calK}$ in \eqref{eqn:truncationset1}--\eqref{eqn:truncationset3}, and for brevity, for any $T\subseteq[n]$, define the sets $\calQ_{T}^{\s{fro}}\triangleq\bigcap_{i\in T}\ppp{A:\sum_{j\in T\setminus\{i\}}A_{ij}^2\leq L^{\s{fro}}_{k,|T|,d}}$, $\calQ_{T}^{\s{op}}\triangleq\ppp{A:\norm{A|_T-d\bs{I}_{|T|}}_{\s{op}}\leq L^{\s{op}}_{k,||T||,d}}$, and $\calQ_{T}\triangleq\calQ_{T}^{\s{fro}}\cap\calQ_{T}^{\s{op}}$. Then, by definition, we note that $\Gamma_{\calK}\subseteq\calQ_{\calK\cap\bar\calK}$ and $\Gamma_{\bar\calK}\subseteq\calQ_{\calK\cap\bar\calK}$, and as so,
$\mathds{1}_{\Gamma_{\calK}}({\bs{W}})\leq \mathds{1}_{\calQ_{\calK\cap\bar\calK}}({\bs{W}})$, and similarly $\mathds{1}_{\Gamma_{\bar\calK}}(\bar{\bs{W}})\leq \mathds{1}_{\calQ_{\calK\cap\bar\calK}}(\bar{\bs{W}})$. Denoting $\calU\triangleq \calK\cap\bar\calK$, we get
%is monotonically decreasing with respect to $\calK$, namely, if $\calK_1\subseteq\calK_2$, then $\Gamma_{\calK_2}\subset\Gamma_{\calK_1}$. Indeed, $\calK_2$ contains all subsets $T\subseteq\calK_1$ and more; thus $\Gamma_{\calK_2}$ imposes all constraints of $\Gamma_{\calK_1}$ and additional ones. Furthermore, for any $\calK_1$ and $\calK_2$, we have $\Gamma_{\calK_1}\cap\Gamma_{\calK_2}\subseteq\Gamma_{\calK_1\cap\calK_2}$. Indeed, because of monotonicity we have $\Gamma_{\calK_1}\subseteq\Gamma_{\calK_1\cap\calK_2}$ and $\Gamma_{\calK_2}\subseteq\Gamma_{\calK_1\cap\calK_2}$, and so, $\Gamma_{\calK_1}\cap\Gamma_{\calK_2}\subseteq\Gamma_{\calK_1\cap\calK_2}$. Therefore, $\mathds{1}_{\Gamma_{\calK}}({\bs{W}})\leq \mathds{1}_{\Gamma_{\calK}\cap\Gamma_{\bar\calK}}({\bs{W}})\leq \mathds{1}_{\Gamma_{\calK\cap\bar\calK}}({\bs{W}})$, and similarly $\mathds{1}_{\Gamma_{\bar\calK}}(\bar{\bs{W}})\leq \mathds{1}_{\Gamma_{\calK\cap\bar\calK}}(\bar{\bs{W}})$. Denoting $\calU\triangleq \calK\cap\bar\calK$, we get
\begin{align}
    g(\bs{A},\calK,\bar\calK)&\leq\frac{1}{(\pr[\Gamma_{\calK}])^2}\bE_{\bs{W}\independent\bar{\bs{W}}\vert\calK,\bar\calK}\left(\prod\limits_{\{i,j\} \in E(\calU)}
               \theta(\bs{A}_{ij},\bs{W}_{ij})\mathds{1}_{\calQ_{\calU}}(\bs{W})\right.\nonumber\\
               &\left.\hspace{5.2cm}\prod\limits_{\{i,j\} \in E(\calU)}
               \theta(\bs{A}_{ij},\bar{\bs{W}}_{ij})
          \mathds{1}_{\calQ_{\calU}}(\bar{\bs{W}})\right)\\
          & = \frac{(\pr[\calQ_{\calU}])^2}{(\pr[\Gamma_{\calK}])^2}\pp{\frac{1}{\pr[\calQ_{\calU}]}\bE_{\bs{W}\vert\calU}\p{\prod\limits_{\{i,j\} \in E(\calU)}
               \theta(\bs{A}_{ij},\bs{W}_{ij})\mathds{1}_{\calQ_{\calU}}(\bs{W})}}^2\\
          &\leq(1+o(1))\cdot\pp{\frac{1}{\pr[\calQ_{\calU}]}\bE_{\bs{W}\vert\calU}\p{\prod\limits_{\{i,j\} \in E(\calU)}
               \theta(\bs{A}_{ij},\bs{W}_{ij})\mathds{1}_{\calQ_{\calU}}(\bs{W})}}^2,\label{eqn:gdef}
\end{align}
where the last inequality is because we have proved that $1-o(1)=\pr[\Gamma_{\calK}]\leq\pr[\calQ_{\calU}]\leq1$. 
Now recalling \eqref{eqn:TruncatedMoment0} and \eqref{eqn:TruncConcdLike} we readily see that the the right-hand side of \eqref{eqn:gdef} can be written as
\begin{align}
    g(\bs{A},\calK,\bar\calK) \leq (1+o(1))\cdot\pp{\frac{\tilde{\P}_{\calH_1\vert \calK\cap\bar\calK} (\mathbf{A}|_{\calK\cap\bar\calK})}{\P_{\calH_0} (\mathbf{A}|_{\calK\cap\bar\calK})}}^2 = (1+o(1))\tilde{\calL}^2_{\calK\cap\bar\calK}(\mathbf{A}|_{\calK\cap\bar\calK}).\label{eqn:condLikeGivenInter}
\end{align}
Combining the above we get that
\begin{align}
\label{eq:L-2-bound0}
    \mathbb{E}_{\mathcal{H}_0}\bigl[\tilde{\calL}^2\bigr]& \leq \pr\pp{\left.|\calK\cap\bar\calK|\leq1\right|\bs{W}\in\Gamma_{\calK},\bar{\bs{W}}\in\Gamma_{\bar\calK}}\nonumber\\
    &\qquad+\mathbb{E}_{\calK\independent\bar\calK}\bE_{\calH_0}\left[\p{\frac{\tilde{\P}_{\calH_1\vert \calK\cap\bar\calK} (\mathbf{A}|_{\calK\cap\bar\calK})}{\P_{\calH_0} (\mathbf{A}|_{\calK\cap\bar\calK})}}^2\Ind\ppp{|\calK\cap\bar\calK|\geq2}\right].
\end{align}
At this point, we observe that the conditional likelihood in \eqref{eqn:condLikeGivenInter} corresponds exactly to a hypothesis test between the Erd\H{o}s--R\'enyi distribution $\calG(|\calK\cap\bar\calK|,q)$ and the \emph{truncated geometric random graph distribution} $\tilde{\calG}_d(|\calK\cap\bar\calK|,q)$, which is defined as the standard geometric random graph model with the additional constraint that the latent vectors lie in the intersection of $\mathbb{S}^{d-1}$ and $\calQ_{\calU}$. Accordingly, let $Z\triangleq|\calK\cap\bar\calK|$, and note that $Z\sim\s{Hypergeometric}(n,k,k)$. Then,
\begin{align}
    \mathbb{E}_{\mathcal{H}_0}\bigl[\tilde{\calL}^2\bigr]& \leq \pr\pp{\left.Z\leq1\right|\bs{W}\in\Gamma_{\calK},\bar{\bs{W}}\in\Gamma_{\bar\calK}}\nonumber\\
    &\qquad+\mathbb{E}_{\calU}\pp{(1+\chi^2(\tilde{\calG}_d(Z,q)\|\calG(Z,q)))\Ind\ppp{Z\geq2}}\\
    &\leq 1+o(1)+\mathbb{E}_{\calU}\pp{\chi^2(\tilde{\calG}_d(Z,q)\|\calG(Z,q))\Ind\ppp{Z\geq2}}.\label{eq:L-2-bound}
\end{align}
Next, for a given $Z=z\geq 2$, we uniformly upper bound $\chi^2\bigl(\tilde{\calG}_d(z,q)\|\calG(z,q)\bigr)$ over $z$. To this end, following the approach of~\cite{bubeck2016testing}, it is convenient to view the graph ensemble distributions $\tilde{\calG}_d(z,q)$ and $\calG(z,q)$ as being generated through a truncation procedure applied to Wishart and Gaussian orthogonal ensemble (GOE) random matrices, as described below.

Specifically, recall that if $\bs{x}$ is a $d$-dimensional standard Gaussian vector, then $\bs{x}/||\bs{x}||_2$ is uniformly distributed on the sphere $\mathbb{S}^{d-1}$. Let $\bs{X}$ be a $z\times d$ matrix obtained from stacking $z$ i.i.d. standard Gaussian vectors in $\mathbb{R}^d$, and let $\bs{W}=\bs{X}\bs{X}^T$ be the corresponding Wishart matrix. Define the matrix $\bs{A}$ as
\begin{align}
\bs{A}_{ij} =
\begin{cases}
1,   & \bs{W}_{ij}/\sqrt{\bs{W}_{ii}\bs{W}_{jj}}\geq t_{q,d} \;\text{, and }\;i\neq j,\\
0, & i=j.
\end{cases}
\end{align}
Then $\bs{A}$ has the same law as the adjacency matrix of $\calG_d(z,q)$. We denote by $H_{q,d}$ the map defined by $H_{q,d}(\bs{W})=\bs{A}$. In the truncated case, we sample $\tilde{\bs{W}}$ from the Wishart distribution conditioned on the event $\calQ_{\calU}$ (that is, we repeatedly sample from the Wishart distribution until $\tilde{\bs{W}}\in\calQ_\calU$), and write $H_{q,d}(\tilde{\bs{W}})=\tilde{\bs{A}}\sim \tilde{\calG}_d(z,q)$.

We next describe an analogous construction for Erdős--Rényi graphs using Gaussian orthogonal ensemble matrices. Let $\bs{M}(z)$ be a GOE matrix: a symmetric $z\times z$ matrix with diagonal entries i.i.d.\ $\mathcal{N}(0,2)$ and off-diagonal entries i.i.d.\ $\mathcal{N}(0,1)$. Define 
\begin{align}
\bs{B}_{ij} =
\begin{cases}
1, & [\bs{M}(z)]_{ij} \ge \overline{\Phi}^{-1}(q) \text{, and } i\neq j,\\
0, & i=j.
\end{cases}
\end{align}
Then $\bs{B}$ has the same law as the adjacency matrix of $\calG(z,q)$. Since $\bs{B}$ depends only on off-diagonal entries, it is convenient to rescale and shift $\bs{M}(z)$ so that it matches the normalization of the Wishart ensemble. Define 
\begin{align}
\bs{M}(z,d) \triangleq  \sqrt{d}\bs{M}(z) + d\bs{I}_{z},\label{eqn:translationScaling}
\end{align}
and accordingly replace $\overline{\Phi}^{-1}(q)$ with $\sqrt{d}\overline{\Phi}^{-1}(q)$. We denote by $K_{q,d}$ the map taking $\bs{M}(z,d)$ to $\bs{B}$, i.e., $\bs{B}=K_{q,d}(\bs{M}(z,d))$. In the truncated case, $\tilde{\bs{M}}$ is sampled from the GOE distribution conditioned on the event $\calQ_\calU$, and we write $\tilde{\bs{B}} = K_{q,d}(\tilde{\bs{M}}(z,d))\sim \tilde{\calG}(z,q)$.
\begin{remark}
\label{rem:goe-probability}
    We are also interested in the probability of $\mathcal{Q}_{\mathcal{U}}$ under the GOE distribution. The proof is essentially identical to the one provided in the previous section. Specifically, the high probability of $\calQ_{\calU}^{\mathrm{fro}}$ under $\mathbb{P}_{\mathcal{H}_0}$ follows from the concentration of the $\chi^2_{z-1}$ distribution around its mean. Similarly, the high probability of $\calQ_{\calU}^{\mathrm{op}}$ is a consequence of the eigenvalues of $\mathbf{M}(z)$ concentrating in the interval $[-3\sqrt{z}, 3\sqrt{z}]$ (see, e.g., \cite[Theorem 4.4.3]{vershynin2018high}). This leads us to conclude that $\P_{\calH_0}(\Gamma_\calK)=1-o(1)$.
\end{remark}

We will frequently bound moments of likelihood ratios, so it is convenient to
introduce the following notation.

\begin{definition}[$D_m$–divergence]
Let $\P,\Q$ be two probability measures such that $\P$ is absolutely continuous with respect to $\Q$, i.e., $\P\ll \Q$. For $m\in\mathbb{N}$ define
\begin{align}
    D_m(\P\|\Q)
\triangleq
    \E_{\Q}\left[\left(\frac{\mathrm{d}\P}{\mathrm{d}\Q}\right)^{m}\right] - 1.
\end{align}
\end{definition}
It is rather a straightforward task to prove that $D_m(\P\|\Q)$ is non negative and convex, implying that it is an $f$ divergence \cite[Ch. 7]{polyanskiywuIT}. We will repeatedly use the following two lemmata in our analysis; the first is standard (e.g., \cite{polyanskiywuIT}), and the second is proved in Appendix~\ref{app:proofsLemmas}.
\begin{lemma}[Data processing inequality]
\label{lem:truncated_dpi}
Let $P_X, Q_X \in \mathcal{P}(\mathcal{X})$ and let $P_{Y|X}$ be a transition kernel from $\mathcal{X}$ to $\mathcal{Y}$. Let $P_Y$ and $Q_Y$ be the resulting marginal measures on $\mathcal{Y}$ such that $P_Y = P_X P_{Y|X}$ and $Q_Y = Q_X P_{Y|X}$. Then, for any $m\in\mathbb{N}$
\begin{align}
\label{eq:truncated_dpi}
    D_{m}(P_Y \| Q_Y) \leq   D_{m}(P_X \| Q_X).
\end{align}
\end{lemma}
\begin{lemma}[Cauchy--Schwarz inequality]
\label{lem:CS-Dm}
Let $\P,\R,\Q$ be three measures over the same probability space, such that, $\P\ll \R\ll \Q$. Then, for any $m\in\mathbb{N}$
\begin{align}
\label{eq:CS-Dm}
    1+D_m(\P||\Q)\le \sqrt{1+D_{2m}(\P||\R)}\sqrt{1+D_{2m-1}(\R||\Q)}.
\end{align}
\end{lemma}

For simplicity of notations, we denote by $\tilde\P_{1}^{(z)}$ and $\P_0^{(z)}$ the probability measures induced by $\tilde{\calG}_d(z,q)\overset{d}{=}H_{q,d}(\tilde{\bs{W}}(z,d))$ and $\calG(z,q)\overset{d}{=}K_{q,d}(\bs{M}(z,d))$, respectively. Furthermore, we denote by $\tilde{\P}_{0}^{(z)}$ the probability measures induced by $H_{q,d}(\tilde{\mathbf{M}}(z,d))$. Then, using the above notations, our goal is to upper bound $\chi^{2}(\tilde\P_{1}^{(z)}\|\P_0^{(z)}) = D_{2}(\tilde\P_{1}^{(z)}\|\P_0^{(z)})$. By Lemma~\ref{lem:CS-Dm} we have
\begin{align}
\label{eq:master-bound}
1+\chi^{2}(\tilde\P_{1}^{(z)}\|\P_0^{(z)})
\leq 
\sqrt{1+D_{4}(\tilde\P_{1}^{(z)}\|\tilde{\P}_{0}^{(z)})}
\sqrt{1+D_{3}(\tilde{\P}_{0}^{(z)}\|\P_0^{(z)})}.
\end{align}
We now bound each of the two terms on the right-hand side separately. Throughout, $z\in[k]$ is fixed, and to lighten the notation, we suppress the superscript $(z)$ from the measures.

\paragraph{Bounding $D_{4}(\tilde\P_{1}^{(z)}\|\tilde{\P}_{0}^{(z)})$.}
Observe that under both $\tilde\P_{1}$ and $\tilde{\P}_{0}$ we map the random matrices $\tilde{\bs{W}}(z,d)$ and $\tilde{\bs{M}}(z,d)$, respectively, using the \emph{same} map $H_{q,d}$. Accordingly, by applying the DPI in Lemma~\ref{lem:truncated_dpi}, we obtain
\begin{align}
\label{eq:first-term-dpi}
    1+D_{4}(\tilde\P_1^{(z)}\|\tilde\P_0^{(z)})=&1+D_{4}(H_{q,d}(\tilde{\bs{W}}(z,d))\|H_{q,d}(\tilde{\bs{M}}(z,d)))\\
    &\le 1+ D_{4}(\tilde{\bs{W}}(z,d)||\tilde{\bs{M}}(z,d)).
\end{align}
Let $\calP\subset\R^{z^2}$ denote the cone of positive semidefinite matrices, and let $f_{z,d}$ be the density of $\bs{W}(z,d)$ with respect to the Lebesgue measure on $\calP$ when $d\ge z$. This condition holds since $z\le k$ and, in the impossible regime, $d/(k^2\log^2k)\to\infty$. Furthermore, let $g_{z,d}$ denote the
density of $\bs{M}(z,d)$ with respect to the Lebesgue measure on $\R^{z^2}$. Accordingly, let $\tilde{f}_{z,d}$ and $\tilde{g}_{z,d}$ be the corresponding densities of $\tilde{\bs{W}}(z,d)$ and $\tilde{\bs{M}}(z,d)$, respectively, that is, the conditional densities given $\calQ_{\calU}$, i.e., $\tilde{f}_{z,d}(\cdot)=f_{z,d}(\cdot\vert\calQ_\calU)$ and $\tilde{g}_{z,d}(\cdot)=g_{z,d}(\cdot\vert\calQ_\calU)$. We observe that
\begin{align}
\tilde{f}_{z,d}(\cdot) = \frac{1}{c_1}f_{z,d}(\cdot)\Ind\{\cdot\in\calQ_{\calU}\}
\quad\s{and}\quad
\tilde{g}_{z,d}(\cdot) = \frac{1}{c_2}g_{z,d}(\cdot)\Ind\{\cdot\in\calQ_{\calU}\},
\end{align}
where $c_1 = \int_{\calQ_{\calU}}f_{z,d}(\bs{A})\mathrm{d}\bs{A}$ and $c_2 = \int_{\calQ_{\calU}}g_{z,d}(\bs{A})\mathrm{d}\bs{A}$ are the normalization constants, and we note that because $\pr(\calQ_{\calU})=1-o(1)$, we have $c_1=1-o(1)$, and by (\ref{rem:goe-probability}), that $c_2=1-o(1)$.
Define $\alpha_{z,d}(\bs{A})\triangleq\log(f_{z,d}(\bs{A})/g_{z,d}(\bs{A}))$ and $\tilde\alpha_{z,d}(\bs{A})\triangleq\log(\tilde f_{z,d}(\bs{A})/\tilde g_{z,d}(\bs{A}))$. Using the above, we clearly have
\begin{align}
1+ D_{4}(\tilde{\bs{W}}(z,d)||\tilde{\bs{M}}(z,d))&=
\mathbb{E}_{\bs{A}\sim\tilde{\bs{M}}(z,d)}\left[\exp\left(4\tilde{\alpha}_{z,d}(\bs{A})\right)\Ind_{\calP}\p{\bs{A}}\right]\\
& \leq (1+o(1))\cdot \mathbb{E}_{\bs{A}\sim\tilde{\bs{M}}(z,d)}\left[\exp\left(4\alpha_{z,d}(\bs{A})\right)\Ind_{\calP}\p{\bs{A}}\right].\label{eqn:D4}
\end{align}
The following is a core estimate on $\alpha_{z,d}(\bs{A})$ used in the proof of \cite[Thm. 7]{bubeck2016testing}.
\begin{lemma}[{\cite[eqns. (36)--(37)]{bubeck2016testing}}]\label{lem:alphaProx} For any $\bs{A}\in\calP$, let $\ppp{\lambda_i}_{i=1}^z$ denote the eigenvalues of $\bs{A}$. With probability one,
\begin{align}
    \alpha_{z,d}(\bs{A})=\sum_{i=1}^zh(\lambda_i)+O\left(\frac{z^3}{d}\right),
\end{align}
as $d-z\to\infty$, where 
\begin{align}
    h(x)\triangleq-\frac{z+1}{2d}(x-d)+\frac{z+1}{4d^2}(x-d)^2+\frac{d-z-1}{6d^3}(x-d)^3+\frac{d-z-1}{6\xi^4}(x-d)^4,\label{eqn:hFuncTerm}
\end{align}
and $\xi$ is some real number between $x$ and $d$.
\end{lemma}
Applying Lemma~\ref{lem:alphaProx} on \eqref{eqn:D4}, we obtain
\begin{align}
1+ D_{4}(\tilde{\bs{W}}(z,d)||\tilde{\bs{M}}(z,d))
&=e^{O\left(\frac{z^3}{d}\right)}\cdot\mathbb{E}_{\bs{A}\sim\tilde{\bs{M}}(z,d)}\left[\exp\left(4\sum_{i=1}^zh(\lambda_i)\right)\mathds{1}_{\{\bs{A}\in\calP\}}\right].\label{eqn:D4int1}
\end{align}
Next, recall that under $\calQ_\calU$ the operator norm constraint in \eqref{eqn:truncationset3} implies that $\abs{\lambda_i-d}\leq L^{\s{op}}_{k,z,d}= C_2(\sqrt{dz}+\sqrt{d\log \binom{k}{z}})\leq 2C_2\sqrt{dz\log k}$, for all $i\in[z]$. Furthermore, because $\xi_i$ is between $\lambda_i$ and $d$ for all $i\in[z]$, we get
\begin{align}
    \sum_{i=1}^z\abs{\frac{z+1}{4d^2}(\lambda_i-d)^2-\frac{d-z-1}{8\xi^4}(\lambda_i-d)^4}\le O\left(\frac{z^3\log^2k}{d}\right).\label{eqn:sqaurequara}
\end{align}
It remains to bound the linear and cubic terms in \eqref{eqn:hFuncTerm}. Define $h_1(x)\triangleq-\frac{z+1}{2d}(x-d)$ and $h_3(x)\triangleq\frac{d-z-1}{6d^3}(x-d)^3$. Then, combining \eqref{eqn:D4int1} and \eqref{eqn:sqaurequara}, we have
\begin{align}
    1+ D_{4}(\tilde{\bs{W}}(z,d)||\tilde{\bs{M}}(z,d))&\leq e^{O\left(\frac{z^3\log^2k}{d}\right)}\cdot\mathbb{E}_{\bs{A}\sim\tilde{\bs{M}}(z,d)}\left[\exp\left(4\sum_{i=1}^zh_1(\lambda_i)+h_3(\lambda_i)\right)\mathds{1}_{\{\bs{A}\in\calP\}}\right]\\
    &\leq e^{O\left(\frac{z^3\log^2k}{d}\right)}\cdot\s A^{1/2}\cdot\s B^{1/2},\label{eq:D4-A-B-bound}
\end{align}
where the last passage follows from Cauchy-Schwarz inequality and we define
\begin{align}
\label{eq:D4-A}
\s A \triangleq \mathbb{E}_{\bs{A}\sim\tilde{\bs{M}}(z,d)}\left[\exp\left(8\sum_{i=1}^zh_1(\lambda_i)\right)\mathds{1}_{\{\bs{A}\in\calP\}}\right],
\end{align}
and
\begin{align}
\label{eq:D4-B}
\s B \triangleq \mathbb{E}_{\bs{A}\sim\tilde{\bs{M}}(z,d)}\left[\exp\left(8\sum_{i=1}^zh_3(\lambda_i)\right)\mathds{1}_{\{\bs{A}\in\calP\}}\right].
\end{align}
Let us upper bound $\s{A}$ and $\s{B}$, starting with the former. First, note that
\begin{align}
\sum_{i=1}^z h_1(\lambda_i)
=
-\frac{z+1}{2d}\Tr\bigl(\bs{A}-d\bs{I}_z\bigr),
\end{align}
and so,
\begin{align}
\s A
&\leq
(1+o(1))\cdot\mathbb{E}_{\bs{A}\sim\bs{M}(z)}
\left[
\exp\left(
-8\frac{z+1}{2\sqrt{d}}\Tr\bigl(\bs{A}\bigr)
\right)
\right] \\
&=
(1+o(1))\cdot\left(
\mathbb{E}_{X\sim\calN(0,2)}
\left[
\exp\left(
-8\frac{z+1}{2\sqrt{d}}X
\right)
\right]
\right)^{z} \\
\label{eq:D4-A-bound}
&=
\exp\left(
O\left(\frac{z^3}{d}\right)
\right),
\end{align}
where in the first inequality we have removed the truncation and use the fact that $c_2 = 1-o(1)$, the first equality is because $[\bs{M}(z)]_{ii}\stackrel{\s{i.i.d.}}{\sim}\calN(0,2)$, and in the second equality we used the fact that $\mathbb{E}[\exp(tX)]=\exp(\sigma^2 t^2/2)$ for $X\sim\calN(0,\sigma^2)$.
Next, we bound $\s{B}$. Denote $G(\bs{A})\triangleq\Tr(\bs{A}^3)$, and note that 
\begin{align}
    \sum_{i=1}^zh_3(\lambda_i)=\frac{d-z-1}{6d^3}G(\bs{A}-d\bs{I}_z).
\end{align}
Before continuing, we collect the statements that will be used in the proof. We first recall the following definition.
\begin{definition}[Log--Sobolev inequality]
Let $\mu$ be a probability measure on $\mathbb{R}^n$ that is absolutely continuous with respect to Lebesgue measure. We say that $\mu$ satisfies a \emph{log--Sobolev inequality} with constant $C_{\s{LSI}} > 0$ if for all smooth functions $f : \mathbb{R}^n \to \mathbb{R}$ with $f^2$ integrable,
\begin{align}
    \s{Ent}_{\mu}(f^2)
\leq
2 C_{\s{LSI}} \int_{\mathbb{R}^n} \|\nabla f(x)\|^2  \mathrm{d}\mu(x),
\end{align}
where the entropy with respect to $\mu$ is defined by
\begin{align}
\s{Ent}_{\mu}(g)\triangleq\int g \log g  \mathrm{d}\mu-\left(\int g  \mathrm{d}\mu\right)
\log\left(\int g  \mathrm{d}\mu\right).
\end{align}
The smallest constant $C_{\s{LSI}}$ for which this inequality holds is called the \emph{log--Sobolev (LSI) constant} of $\mu$.
\end{definition}
\begin{lemma}[{\cite[Lemma 2.3.3]{Anderson_Guionnet_Zeitouni_2009}}] 
\label{lem:lipschitz-concentraion}
Let $P$ be a probability measure satisfying the LSI on $\R^M$ with constant $C_{\s{LSI}}$. Let $G$ be a Lipschitz function on $\R^M$, with Lipschitz constant $\abs{G}_\calL$. Then for all $\beta\in\R$,
\begin{align}
    \E_P\left[\exp\left(\beta\left(G-\E_P[G]\right)\right)\right]\le\exp\left(C_{\s{LSI}}\beta^2\abs{G}_\calL^2/2\right).
\end{align}
\end{lemma}
\begin{lemma}[{\cite[Lemma 2.3.2]{Anderson_Guionnet_Zeitouni_2009}}]
\label{lem:lsi-constant}
Let $P$ be a Gaussian law with mean zero and variance $\sigma^2$. Then, the $C_{\s{LSI}}=\sigma^2$, additionally, if $P=\bigotimes_{i=1}^n P_i$ such that each probability measure $P_i$ has LSI constant $C_{i}$ then $C_{\s{LSI}}=\max_i{C_i}$.
\end{lemma}
\begin{lemma}[McShane-Whitney extension theorem (e.g., {\cite[Thm. 1.1]{Gutev_2025}})] 
\label{thm:lipschitz-extension}
Let $(X,d)$ be a metric space and $A\subset X$. Then every $L$-Lipschitz function $f:A\to\R$ can be extended to an $L$-Lipschitz function $f^\star:X\to\R$.
\end{lemma}
We are now in a position to bound $\s{B}$. First, note that
\begin{align}
\s B &= \mathbb{E}_{\bs{A}\sim\tilde{\bs{M}}(z)}\left[\exp\left(\frac{d-z-1}{6d^{3/2}}G(
\bs{A})\right)\mathds{1}_{\{\bs{A}\in\calP\}}\right]\\
& = \mathbb{E}_{\bs{A}\sim\tilde{\bs{M}}(z)}\left[\exp\left(\frac{C}{\sqrt{d}}G(
\bs{A})\right)\mathds{1}_{\{\bs{A}\in\calP\}}\right]\\
&\leq (1+o(1))\cdot\mathbb{E}_{\bs{A}\sim\bs{M}(z)}\left[\exp\left(\frac{C}{\sqrt{d}}G(
\bs{A})\right)\mathds{1}_{\bar\calQ_{\calU}^{\s{op}}}(\bs{A})\right],
\end{align}
where the second equality follows from the fact that $d\gg z$ and $C>0$, and the
inequality holds since $c_2=1-o(1)$. Moreover, we retain only the (rescaled, as in
\eqref{eqn:translationScaling}) operator-norm truncation set in
\eqref{eqn:truncationset3}, namely, $\bar\calQ_{\calU}^{\s{op}}\triangleq\ppp{\norm{\bs{A}}_{\s{op}}\leq \bar{L}^{\s{op}}_{k,z,d}}$ and $\bar{L}^{\s{op}}_{k,z,d}\triangleq L^{\s{op}}_{k,z,d}/\sqrt{d} =  C_2(\sqrt{z}+\sqrt{\log \binom{k}{z}}) = O(\sqrt{z\log k})$. Now, by the mean value theorem and the Cauchy--Schwarz inequality, for any $\bs{A},\bs{B}$ there exists a matrix $\bs{C}$ on the line segment connecting $\bs{A}$ and $\bs{B}$ such that
\begin{align}
    \abs{G(\bs{A})-G(\bs{B})}\le\langle\nabla G(\bs{C}),\bs{A}-\bs{B}\rangle\le\|\nabla G(\bs{C})\|_F\cdot\|\bs{A}-\bs{B}\|_F.
\end{align}
Moreover, since $\bar\calQ_{\calU}^{\s{op}}$ is convex, the matrix $\bs{C}$ also belongs to this set, and hence
\begin{align}
    \|\nabla G(\bs{C})\|_F\le \sqrt{z}\|\nabla G(\bs{C})\|_{\s{op}}=3\sqrt{z}\|\bs{C}\|_{\s{op}}^2= O(z^{3/2}\log k),
\end{align}
which implies that $\abs{G}_\calL \leq O(z^{3/2}\log k)$. We now apply Lemma~\ref{thm:lipschitz-extension} to extend $G$ from $\bar\calQ_{\calU}^{\s{op}}$ to a function $G^\star$ defined on the space of GOE matrices, such that $G^\star=G$ on $\bar\calQ_{\calU}^{\s{op}}$ and $|G^\star|_\calL=|G|_\calL$. By Lemma~\ref{lem:lsi-constant}, the LSI constant for the GOE is $C_{\s{LSI}}=2$. Using the fact that $\E_{\bs{A}\sim\bs{M}(z)}[\s{Tr}(\bs{A}^3)]=0$ and the positivity of the exponential function, Lemma~\ref{lem:lipschitz-concentraion} gives
\begin{align}
\label{eq:D4-B-bound}
    \s{B}\leq  (1+o(1))\cdot\mathbb{E}_{\bs{A}\sim\bs{M}(z)}\left[\exp\left(\frac{C}{\sqrt{d}}G^\star(
\bs{A})\right)\right]\le\exp\left(O\left(\frac{z^3\log^2k}{d}\right)\right).
\end{align}
Combining \eqref{eq:D4-A-bound} and \eqref{eq:D4-B-bound} with \eqref{eq:D4-A-B-bound} and \eqref{eq:first-term-dpi}, we conclude that
\begin{align}
\label{eq:D4-bound}
1+D_{4}(\tilde\P_1^{(z)}\|\tilde\P_0^{(z)})\leq\exp\left(O\left(\frac{z^3\log^2k}{d}\right)\right).
\end{align}
\paragraph{Bounding $D_{3}(\tilde\P_0^{(z)}\|\P_0^{(z)})$.} 
We again introduce an intermediate probability measure. Recall that $\tilde\P_0^{(z)}$ and $\P_0^{(z)}$ denote the probability measures induced by the measurable mappings $H_{q,d}(\tilde{\bs{M}}(z,d))$ and $K_{q,d}(\bs{M}(z,d))$, respectively. For notational convenience, we denote the corresponding random matrices by $\tilde{\bs{X}}$ and $\bs{Y}$, respectively. Furthermore, define
\begin{align}
    \bs{D}_{ij}
    \triangleq
    \sqrt{\left(1+[\bs{M}(z)]_{ii}/\sqrt{d}\right)
          \left(1+[\bs{M}(z)]_{jj}/\sqrt{d}\right)},
\end{align}
and
\begin{align}
    \tilde{\bs{D}}_{ij}
    \triangleq
    \sqrt{\left(1+[\tilde{\bs{M}}(z)]_{ii}/\sqrt{d}\right)
          \left(1+[\tilde{\bs{M}}(z)]_{jj}/\sqrt{d}\right)}.
\end{align}
With this notation, the entries of $\tilde{\bs{X}}$ and $\bs{Y}$ can be written as
\begin{align}
    \tilde{\bs{X}}_{ij}  =     \Ind\{\tilde{\bs{D}}_{ij}^{-1}\tilde{\bs{M}}_{ij} \ge \sqrt{d}t_{q,d}\},
    \qquad
    \bs{Y}_{ij} = \Ind\{\bs{M}_{ij} \ge \overline{\Phi}^{-1}(q)\}.
\end{align}
Additionally, let $\tilde{\bs{Z}}$ denote the $z\times z$ matrix defined by
\begin{align}
\tilde{\bs{Z}}_{ij} = \Ind\{\tilde{\bs{M}}_{ij} \ge \sqrt{d}\,t_{q,d}\},\label{eqn:tildeZ}
\end{align}
for all $1\le i<j\le z$, and let $\tilde\Q^{(z)}$ denote its law. We denote by $\bs{Z}$ and $\Q^{(z)}$ the corresponding non-truncated matrix and its law, respectively. By Lemma~\ref{lem:CS-Dm}, we have
\begin{align}
1+D_{3}(\tilde\P_0^{(z)}\|\P_0^{(z)}) \leq \sqrt{1+D_{6}(\tilde\P_0^{(z)}\|\tilde\Q^{(z)})}\sqrt{1+D_{5}(\tilde\Q^{(z)}\|\P_0^{(z)})}.\label{eqn:6to11to12}
\end{align}
We now bound each of the two terms appearing on the right-hand side separately.
\paragraph{Bounding $D_{5}(\tilde\Q^{(z)}\|\P_0^{(z)})$.}
By definition, it is clear that $\tilde{\bs{Z}}$ are stochastically dominated by ${\bs{Z}}$. Therefore
\begin{align}
    D_{5}(\tilde\Q^{(z)}\|\P_0^{(z)})\leq D_{5}(\Q^{(z)}\|\P_0^{(z)}).
\end{align}
Since both $\Q^{(z)}$ and $\P_0^{(z)}$ are product measures, we have
\begin{align}
    1+D_{5}(\Q^{(z)}\|\bar{\P}^{(z)}) = \prod_{1\le i<j\le z} \left(1+D_{5}(\Q^{(z)}_{ij}\|\bar{\P}^{(z)}_{ij})\right),
\end{align}
where $\Q^{(z)}_{ij} = \s{Bern}(q+\delta_{q,d})$ and $\bar{\P}^{(z)}_{ij} = \s{Bern}(q)$, where $\delta_{q,d} \triangleq \bar{\Phi}(\sqrt{d}t_{q,d}) - \Psi_d(\sqrt{d}t_{q,d})$. We utilize the following result.
\begin{lemma}[{\cite[Lemma~7]{bubeck2016testing}}]
\label{lem:delta-bound}
For any fixed $q \in (0,1)$, there exists a constant $C_q$ such that:
\begin{align}
    |\delta_{q,d}| = |\bar{\Phi}(\sqrt{d}t_{q,d}) - \Psi_d(\sqrt{d}t_{q,d})| \le C_q d^{-1}.
\end{align}
\end{lemma}

Expanding $D_{5}(\Q^{(z)}_{ij}\|\bar{\P}^{(z)}_{ij})$, we obtain
\begin{align}
    1+D_{5}(\Q^{(z)}_{ij}\|\bar{\P}^{(z)}_{ij}) 
    &= q\left(1+\frac{\delta_{q,d}}{q}\right)^{5} + (1-q)\left(1-\frac{\delta_{q,d}}{1-q}\right)^{5}  \\
    &= 1 + \frac{55 \delta_{q,d}^2}{q(1-q)} + O_q(\delta_{q,d}^3).
\end{align}
Applying Lemma~\ref{lem:delta-bound}, we have $\delta_{q,d}^2 \le C_q^2 d^{-2}$. Substituting this bound into the product over edges and using the inequality $1+x \le e^x$, we obtain
\begin{align}
    1+D_{5}(\Q^{(z)}\|\bar{\P}^{(z)})
    &\le \left(1 + \frac{55 C_q^2}{q(1-q)} d^{-2} + O(d^{-3})\right)^{\binom{z}{2}} \\
    &\le \exp\left( \frac{z(z-1)}{2} \left[ \frac{55 C_q^2}{q(1-q) d^2} \right] \right) \\
    &\le \exp\left( \bar C_{q} \frac{z^2}{d^2} \right),\label{eq:D11-bound}
\end{align}
for a constant $\bar C_{q}$ depending only on $q$.
\paragraph{Bounding $D_{6}(\tilde\P_0^{(z)}\|\tilde\Q^{(z)})$.} We begin by introducing some notation. Recall that $\tilde{\bs{X}}=H_{q,d}(\tilde{\bs{M}}(z,d))\sim\tilde\P_0^{(z)}$ and $\tilde{\bs{Z}}\sim\tilde\Q^{(z)}$, where the entries of $\tilde{\bs{Z}}$ are defined in \eqref{eqn:tildeZ}. Let $\bs{X}$ and $\bs{Z}$ denote the untruncated versions of $\tilde{\bs{X}}$ and $\tilde{\bs{Z}}$, respectively, that is, $\bs{X}=H_{q,d}(\bs{M}(z,d))$, and the entries of $\bs{Z}$ are given by \eqref{eqn:tildeZ} with $\tilde{\bs{M}}$ replaced by $\bs{M}$. 

For any matrix $\bs{A}$, let $D'(\bs{A})$ denote the diagonal matrix formed by retaining only the diagonal entries of $\bs{A}$, and define the off-diagonal component by $D(\bs{A}) = \bs{A}-D'(\bs{A})$. Introduce the transformation
\begin{align}
    f(\bs{A})\triangleq\left(\bs{I}_z+\frac{1}{\sqrt{d}}D'(\bs{A})\right)^{-1/2}\bs{A}\left(\bs{I}_z+\frac{1}{\sqrt{d}}D'(\bs{A})\right)^{-1/2},
\end{align}
Observe that the random matrix $\bs{Z}$ is obtained by applying an entrywise thresholding procedure to $D(\bs{M})$, whereas $\bs{X}$ is obtained analogously from $D(f(\bs{M}))$. Similarly, $\tilde{\bs{Z}}$ and $\tilde{\bs{X}}$ are obtained from $D(\tilde{\bs{M}})$ and $D(f(\tilde{\bs{M}}))$, respectively, via the same entrywise thresholding procedure. Hence, Lemma~\ref{lem:truncated_dpi} implies that
\begin{align}
    1+D_{6}(\bs{X} \| \bs{Z}) \leq 1+D_{6}(D(f({\bs{M}}))\|D({\bs{M}})),
\end{align}
where, with a slight abuse of notation, we use the matrix-valued random variables to denote the corresponding probability measures. Similarly,
\begin{align}
  1+D_{6}(\tilde\P_0^{(z)}\|\tilde\Q^{(z)})=  1+D_{6}(\tilde{\bs{X}} \| \tilde{\bs{Z}}) \leq 1+D_{6}(D(f(\tilde{\bs{M}}))\|D(\tilde{\bs{M}})).%= \int_{\Gamma_U} \left( \frac{dD(f(\bs{M}))}{dD(\bs{M})} \right)^{6}dD(\bs{M})
\end{align}
Thus, to control the right-hand side of the above inequalities, it suffices to study the densities of
$D(\bs{M})$ and $D(f(\bs{M}))$, denoted by $q(x)$ and $w(x)$, respectively, together with their
truncated counterparts, whose densities are denoted by $\tilde{q}(x)$ and $\tilde{w}(x)$.

Let $\Omega = \mathbb{R}^{(n^2-n)/2}$, identified with the space of symmetric $n \times n$ matrices
having zero diagonal, and let $\Omega' = \mathbb{R}^n$, corresponding to the space of diagonal
entries. Slightly abusing notation, we treat $D$ and $D'$ as maps from symmetric matrices into
$\Omega$ and $\Omega'$, respectively. Under this identification, the mapping $D \circ f$ acts from
$\Omega \oplus \Omega'$ into $\Omega$.

Accordingly, $w(x)$ arises as the push-forward of the Gaussian measure with density $\gamma(x,y)$
under the transformation $D \circ f$, where $\gamma(x,y)$ is the GOE density on
$\Omega \oplus \Omega'$, i.e.,
\begin{align}
\gamma(x,y)\triangleq
\frac{1}{2^{n/2}(2\pi)^{(n^2+n)/2}}
\exp\left(
-\frac{1}{2}\sum_{1 \le i < j \le n} x_{ij}^2
-\frac{1}{4}\sum_{i=1}^n y_i^2
\right),
\end{align}
for $(x,y)\in\Omega \oplus \Omega'$. Similarly, since $D(\bs{M})$ is the image of the same Gaussian measure under the map $D$, the density $q(x)$ is given by
\begin{align}
q(x)=\int_{\Omega'} \gamma(x,y)\mathrm{d}y
= \frac{1}{(2\pi)^{(n^2-n)/2}}
\exp\left(-\frac{1}{2}\sum_{1 \le i < j \le n} x_{ij}^2\right),
\end{align}
which is the standard Gaussian density on $\Omega$. A closed-form expression for $w(x)$ was derived in \cite[Lemma~6]{bubeck2016testing}.
\begin{lemma}[{\cite[Lemma~6]{bubeck2016testing}}]\label{lem:Radonwq}
The law of $D(f(\bs{M}))$ is absolutely continuous with respect to the law of $D(\bs{M})$. Moreover, the Radon--Nikodym derivative is given by
\begin{align}
\frac{\mathrm{d}\calL(D(f(\bs{M})))}{\mathrm{d}\calL(D(\bs{M}))}(x)
&=\frac{w(x)}{q(x)}, \qquad x \in \Omega,
\end{align}
where $w$ and $q$ denote the densities of $D(f(\bs{M}))$ and $D(\bs{M})$, respectively. Moreover,
\begin{align}
\frac{w(x)}{q(x)} =
\mathbb{E}_\Lambda\Bigg[
\prod_{i=1}^n 
\Bigl|1+\frac{\Lambda_i}{\sqrt{d/2}}\Bigr|^{\frac{z-1}{2}}
\exp\left(
-\frac12\sum_{1\le i<j\le z}
x_{ij}^2\left(
\frac{\Lambda_i+\Lambda_j}{\sqrt{d/2}}
+
\frac{2\Lambda_i\Lambda_j}{d}
\right)
\right)
\Bigg],
\end{align}
where $\Lambda=(\Lambda_1,\dots,\Lambda_z)\sim N(0,\bs{I}_z)$.
\end{lemma}
Let us now characterize the truncated densities $\tilde{q}(x)$ and $\tilde{w}(x)$, starting with
the former. Define the conditional (truncated) law
\begin{align}
\tilde{\gamma}(x,y)
\triangleq
\frac{\gamma(x,y)\mathds{1}_{\calQ_{\calU}}(x,y)}{\mathbb P_\gamma(\calQ_{\calU})},
\qquad
\mathbb P_\gamma(\calQ_{\calU})
\triangleq
\int_{\Omega \times \Omega'} \gamma(x,y)\mathds{1}_{\calQ_{\calU}}(x,y)
\,\mathrm{d}x\,\mathrm{d}y.
\end{align}
Define the random variables $X \triangleq D(\bs{M}) \in \Omega$ and $Z \triangleq D(f(\bs{M})) \in \Omega$. Then
\begin{align}
\tilde{q}(x) &= \int_{\Omega'}\tilde{\gamma}(x,y)\mathrm{d}y \\
&= \frac{1}{\mathbb P_\gamma(\calQ_{\calU})}\int_{\Omega'} \gamma(x,y)\mathds{1}_{\calQ_{\calU}}(x,y)\mathrm{d}y.
\end{align}
Define the conditional density of $D'(\bs{M})$ given $X=x$ by $\gamma(y\vert x)\triangleq\frac{\gamma(x,y)}{q(x)}$. Then the conditional probability of $\calQ_{\calU}$ given $X=x$ is
\begin{align}
\mathbb P_\gamma(\calQ_{\calU}\vert X=x) &= \int_{\Omega'} \mathds{1}_{\calQ_{\calU}}(x,y)\gamma(y\vert x)\mathrm{d}y \\
&= \frac{1}{q(x)}\int_{\Omega'} \gamma(x,y)\mathds{1}_{\calQ_{\calU}}(x,y)\mathrm{d}y.
\end{align}
Rearranging,
\begin{align}
\int_{\Omega'} \gamma(x,y)\mathds{1}_{\calQ_{\calU}}(x,y)\mathrm{d}y
= q(x)\mathbb P_\gamma(\calQ_{\calU} \vert X=x).
\end{align}
Substituting into the expression for $\tilde{q}(x)$ yields
\begin{align}
\tilde{q}(x) = q(x)\frac{\mathbb P_\gamma(\calQ_{\calU} \vert X=x)}{\mathbb P_\gamma(\calQ_{\calU})}.
\end{align}
In a similar fashion, we obtain
\begin{align}
\tilde{w}(x) = w(x)\frac{\mathbb P_\gamma(\calQ_{\calU} \vert Z=x)}{\mathbb P_\gamma(\calQ_{\calU})}.
\end{align}
Therefore, for all $x \in \Omega$,
\begin{align}
\frac{\tilde{w}(x)}{\tilde{q}(x)} = \frac{w(x)}{q(x)} \cdot \frac{\mathbb P_\gamma(\calQ_{\calU} \vert Z=x)}{\mathbb P_\gamma(\calQ_{\calU} \vert X=x)}\leq\frac{w(x)}{q(x)} \cdot \frac{1}{\mathbb P_\gamma(\calQ_{\calU} \vert X=x)}.
\end{align}
We are now in a position to bound $D_{6}(D(f(\tilde{\bs{M}}))\|D(\tilde{\bs{M}}))$. By definition, we have
\begin{align}
    1+D_{6}(D(f(\tilde{\bs{M}}))\|D(\tilde{\bs{M}})) &= \bE_{X\sim\tilde{q}}\pp{\frac{\tilde{w}(X)}{\tilde{q}(X)}}^{6}\\
    &\leq \bE\pp{\pp{\frac{w(X)}{q(X)}}^{6}\frac{1}{\mathbb P_\gamma^{6}(\calQ_{\calU} \vert X)}}\\
    & = \int_{\calQ_{\calU}}\pp{\frac{w(x)}{q(x)}}^{6}\frac{1}{\mathbb P_\gamma^{6}(\calQ_{\calU} \vert X=x)}\tilde{q}(x)\mathrm{d}x\\
    & = \int_{\calQ_{\calU}}\pp{\frac{w(x)}{q(x)}}^{6}\tilde{q}(x)\mathrm{d}x.\label{eqn:D12ExpliFor}
\end{align}
Next, we derive a uniform upper bound on $\frac{w(x)}{q(x)}$, for all $x\in\calQ_{\calU}$. Specifically, using Lemma~\ref{lem:Radonwq} and the fact that $|1+u|\le e^{u+u^2}$ (applied with $u=\Gamma_i/\sqrt{d/2}$), we obtain
\begin{align}
\frac{w(x)}{q(x)}\leq
\mathbb{E}_\Lambda\Bigg[
\exp\Bigg(
&\frac{z-1}{\sqrt{2d}}\sum_{i=1}^z \Lambda_i
+\frac{z-1}{d}\sum_{i=1}^z \Lambda_i^2 \notag\\
&\quad
-\frac{1}{\sqrt{2d}}\sum_{i<j} x_{ij}^2(\Lambda_i+\Lambda_j)
-\frac{2}{d}\sum_{i<j} x_{ij}^2\Lambda_i\Lambda_j
\Bigg)
\Bigg] \\
= &\mathbb{E}_\Lambda\left[\exp\left(\bs{L}^T\Lambda + \Lambda^T\bs{Q}\Lambda\right)\right],\label{eq:w-over-q}
\end{align}
where the entries of the vector $\bs{L}$ are defined by
\begin{align}
\bs{L}_i \triangleq \frac{1}{\sqrt{2d}} \left( (z-1) - \sum_{j=1, j\neq i}^z x_{ij}^2 \right),
\end{align}
for $i\in[z]$, and the entries of the matrix $\bs{Q}$ are given by
\begin{align}
\bs{Q}_{ij} = \begin{cases}
\frac{z-1}{d}, & \s{if}\; i = j \\
-\frac{1}{d} x_{ij}^2, & \s{if}\; i \neq j, 
\end{cases}
\end{align}
for $1\le i,j\le z$. Crucially, the expectation in \eqref{eq:w-over-q} is finite only if the matrix $\bs{I}_z-2\bs{Q}$ is positive definite. By the Gershgorin circle theorem, the eigenvalues of $\bs{Q}$ lie in the union of the discs
\begin{align}
    \calD_i=\ppp{\lambda\in\mathbb{R}:\abs{\lambda-\frac{z-1}{d}}\leq\frac{1}{d}\sum_{j=1,j\neq i}^zx_{ij}^2},
\end{align}
for $i\in[z]$. Consequently, the operator norm of $\bs{Q}$ satisfies
\begin{align}
\label{Q-operator-norm}
\|\bs{Q}\|_{\s{op}}\leq \frac{z-1}{d}+\frac{1}{d}\max_{1\le i\le z}\sum_{j=1,j\neq i}^z x_{ij}^2.
\end{align}
%\begin{align}
%\label{Q-operator-norm}
%    \|\bs{Q}\|_{\s{op}}\in\ppp{\lambda\in\mathbb{R}:\abs{\lambda-\frac{z-1}{d}}\leq\frac{1}{d}\max\sum_{j=1,\,j\neq i}^zx_{ij}^2}.
%\end{align}
Furthermore,
\begin{align}
    \|\bs{L}\|^2_2=\frac{1}{2d}\sum_{i=1}^z\left(z-1-\sum_{j\neq i}x_{ij}^2\right)^2\leq\frac{z(z-1)^2}{2d}+\frac{z}{2d}\left(\max_{1\leq i\leq z}\sum_{j\neq i}x_{ij}^2\right)^2.
\end{align}
Now, for any $x\in\calQ_{\calU}$, by the definition of $\calQ_{\calU}$, we have
\begin{align}
    \sum_{j=1,j\neq i}x_{ij}^2\leq(1+C_1)d(z-1)\log k,
\end{align}
for all $i\in[z]$. Consequently,
\begin{align}
\label{eq:L-l2-bound}
    \|\bs{L}\|_2^2\le \frac{z^3}{2d}+\frac{z^3\log^2k}{2d}\le C'\frac{z^3\log^2k}{d},
\end{align}
and
\begin{align}
\label{eq:Q-op-bound}
    \|\bs{Q}\|_{\s{op}}\le \frac{z-1}{d}+(1+C_1)\frac{(z-1)\log k}{d}\le C''\frac{z\log k}{d},
\end{align}
for some positive constants $C'$ and $C''$. In the impossible regime, where $d\gg k^2\log^2 k$, it follows that for $d$ sufficiently large we have $\|\bs{Q}\|_{\s{op}}<\tfrac12$, and we may therefore apply the Gaussian integral formula to obtain
\begin{align}
\label{gaussian-formula}
\mathbb{E}_{\Lambda}\left[\exp\left(\bs{L}^T\Lambda + \Lambda^T \bs{Q} \Lambda\right)\right] = \det(\bs{I}_z - 2\bs{Q})^{-1/2} \exp\left( \frac{1}{2} \bs{L}^T (\bs{I}_z - 2\bs{Q})^{-1} \bs{L} \right).
\end{align}
Moreover, since all eigenvalues of $\bs{Q}$ are of order $\Theta(z\log k/d)$, the eigenvalues of $\bs{I}_z-2\bs{Q}$ lie in the interval $[1-O(z\log k/d),\,1+O(z\log k/d)]$. In the regime where $d\gg k^2\log^2 k$, this interval is contained in $[1/2,3/2]$. Hence, letting $\{\lambda_i\}_{i=1}^z$ denote the eigenvalues of $\bs{I}_z-2\bs{Q}$, we obtain
\begin{align}
\label{det-lower-bound}
\det(\bs{I}_z-2\bs{Q})=\prod_{i=1}^z \lambda_i\geq (1-2\|\bs{Q}\|_{\s{op}})^z .
\end{align}
This implies that
\begin{align}
-\frac{1}{2}\log\det\left(\bs{I}_z-2\bs{Q}\right)
&\leq -\frac{z}{2}\log\left(1-2\|\bs{Q}\|_{\s{op}}\right) \\
&\leq \frac{z}{2}\cdot\frac{2\|\bs{Q}\|_{\s{op}}}{1-2\|\bs{Q}\|_{\s{op}}} \\
&\leq \bar{C}\frac{z^2\log^2 k}{d},\label{eqn:boundDetGauss}
\end{align}
for some constant $\bar{C}>0$. Here, the first inequality follows from \eqref{det-lower-bound}, while the second uses the bound $\log(1-x)\le x/(1-x)$ for $0<x<1$, together with \eqref{eq:Q-op-bound}. Turning to the second factor in \eqref{gaussian-formula}, let $\{v_i\}_{i=1}^z$ denote the eigenvectors of $\bs{I}_z-2\bs{Q}$. Then
\begin{align}
\frac{1}{2}\bs{L}^T(\bs{I}_z-2\bs{Q})^{-1}\bs{L}&= \frac{1}{2}\sum_{i=1}^z \lambda_i^{-1}\abs{\langle v_i,\bs{L}\rangle}^2 \\
&\leq \frac{1}{2}\frac{\|\bs{L}\|_2^2}{1-2\|\bs{Q}\|_{\s{op}}}\\
&\leq \bar{C}'\|\bs{L}\|_2^2 \\
&\leq \bar{C}''\frac{z^3\log^2 k}{d},\label{eqn:boundQadGauss}
\end{align}
for some positive constants $\bar{C}',\bar{C}''>0$. Finally, combining \eqref{eqn:D12ExpliFor}, \eqref{eq:w-over-q}, \eqref{gaussian-formula}, \eqref{eqn:boundDetGauss}, and \eqref{eqn:boundQadGauss}, we get
\begin{align}
\label{eq:D12-bound}
   1+D_{6}(\tilde\P_0^{(z)}\|\tilde\Q^{(z)})\leq \exp\left(C\left(\frac{z^2}{d}+\frac{z^3}{d}\right)\log^2k\right),
\end{align}
for some constant $C>0$.
%\paragraph{Bounding $D_{5}(\bar{\P}_0^{(z)}\|\P_0^{(z)})$.}
%Both $\bar{\P}_0^{(z)}$ and $\P_0^{(z)}$ are product measures of i.i.d. Bernoulli random variables with parameters $p$ and $q$, respectively, and thus
%\begin{align}
%1 + D_{5}(\bar{\P}_0^{(z)}\|\P_0^{(z)})
%= \left(1 + D_5\bigl(\s{Bern}(p)\|\s{Bern}(q)\bigr)\right)^{\binom{z}{2}}
%\leq \exp\left(D_5(p\|q)z^2\right).\label{eq:D5-bound}
%\end{align}
%If $p,q$ are constants then we can trivially write $D_5(p\|q)\le C_{p,q}\chi^2(p\|q)$ for some positive constant $C_{p,q}>0$. Inferring that,
%\begin{align}
%\label{eq:D5-bound}
%    1 + D_5(\P|\Q)\le\exp\left(C_{p,q}\chi^2(p\|q)z^2\right).
%\end{align}
\paragraph{Completing the second-moment calculation.} After bounding each divergence separately, we are ready to assemble the final bound on $\mathbb{E}_{\mathcal{H}_0}\bigl[\tilde{\calL}^2\bigr]$. To this end, we substitute \eqref{eq:D4-bound}, \eqref{eqn:6to11to12}, \eqref{eq:D11-bound}, and \eqref{eq:D12-bound} into \eqref{eq:master-bound}, and then into \eqref{eq:L-2-bound}. We obtain
\begin{align}
    \mathbb{E}_{\mathcal{H}_0}\bigl[\tilde{\calL}^2\bigr]& \leq 1+o(1)+\mathbb{E}_{\calU}\pp{\chi^2(\tilde{\calG}_d(Z,q)\|\calG(Z,q))\Ind\ppp{Z\geq2}}\\
    &\leq 1+o(1)+\E_Z\left[\p{\exp\left(C\frac{Z^3\log^2k}{d}\right)-1}\Ind\ppp{Z\geq2}\right]\\
    & \leq\E_Z\left[\exp\left(C\frac{Z^3\log^2k}{d}\right)\right]+o(1),\label{eq:L-tilde-2-bound}
    %&\leq \E_Z^{1/2}\left[\exp\left(C\frac{Z^3\log^2k}{d}\right)\right]\cdot\E_Z^{1/2}\left[\exp\left(2D_5(p\|q)Z^2\right)\right]+o(1)
\end{align}
for some constant $C>0$, and the third inequality follows from the fact that $\pr[Z\geq0]=1$. Let us find the conditions under which the right-hand side of \eqref{eq:L-tilde-2-bound} is converging to unity or bounded. 
%Consider first the second expectation on the right-hand side of \eqref{eq:L-tilde-2-bound}. This term is, in fact, an upper bound on the second moment of the likelihood ratio corresponding to the planted densest subgraph detection problem (see, e.g.,~\cite{verzelen2015community}). Accordingly, it is well known that, in the dense regime where $p$ and $q$ are fixed constants independent of $n$, this expectation equals $1+o(1)$ provided that $k \le C_{p,q}\log n$, for some constant $C_{p,q}>0$. In our analysis, the constant $C_{p,q}$ is not optimized; nevertheless, this condition captures one of the requirements appearing in Theorem~\ref{thm:ITlower}. Since the analysis of this expectation is classical, we omit the details. The argument relies on the fact that $Z$ is stochastically dominated by $B\sim\s{Binomial}\left(k,\frac{k}{n-k}\right)$, which allows one to replace $Z^2$ by $k\cdot B$ and to control the resulting expression via the moment generating function of the binomial random variable.

Next, we analyze the expectation on the right-hand side of \eqref{eq:L-tilde-2-bound}. We decompose the expectation into two regimes: the first corresponding to $Z \le z_1$, and its complement, where we define $z_1\triangleq C'\lceil(d/\log^2k)^{1/3}\rceil$, for some $C'>0$. We have
\begin{align}
\E_Z\left[\exp\left(C\frac{Z^3\log^2k}{d}\right)\right]&= \E_Z\left[\exp\left(C\frac{Z^3\log^2k}{d}\right)\mathds{1}_{\{Z\le z_1\}}\right]\nonumber\\
&\qquad+\E_Z\left[\exp\left(C\frac{Z^3\log^2k}{d}\right)\mathds{1}_{\{Z> z_1\}}\right].\label{eqn:Z3dFirst}
\end{align}
As for the first expectation on the right-hand side of \eqref{eqn:Z3dFirst}, we consider the impossibility of strong and weak detection separately. For the former, we simply have
\begin{align}
\label{eq:z-small-crude}
    \E_Z\left[\exp\left(C\frac{Z^3\log^2k}{d}\right)\mathds{1}_{\{z\le z_1\}}\right]\le \E_Z\left[\exp\left(C\frac{z_1^3\log^2k}{d}\right)\right]=O(1).
\end{align}
Next, we turn to the impossibility of weak detection. We assume that $\frac{d n^3}{k^6 \log^2 k\, f(k)} \to \infty$ for any function $f(k)=\omega(1)$. We again split the analysis into two regimes: $Z \le z_0 \triangleq (k^2 f(k))/n$ and $z_0 < Z \le z_1$. For $Z \le z_0$, we bound the expectation as follows:
\begin{align}
    \E_Z\left[\exp\left(C\frac{Z^3\log^2k}{d}\right)\mathds{1}_{\{Z\le z_0\}}\right]\le \exp\left(C\frac{k^6\log^2kf(k)}{n^3d}\right)=1+o(1).
\end{align}
For the regime $z_0 < Z \le z_1$, we use the fact that $Z\sim\s{Hypergeometric}(n,k,k)$ is stochastically dominated by $B\sim\s{Binomial}(k,\rho)$ with $\rho = k/(n-k)$. Together with the Chernoff bound $
\P(B \ge z) \le \exp\bigl(-k d_{\s{KL}}(z/k \|\rho)\bigr)$, this yields
\begin{align}
\E_Z\left[\exp\left(C\frac{Z^3\log^2k}{d}\right)\mathds{1}_{\{z_0< z\le z_1\}}\right]
&\le \sum_{z=z_0+1}^{z_1}\exp\left[C\frac{z^3\log^2k}{d}-kd_{\s{KL}}\left(\frac{z}{k}\|\rho\right)\right]\\
&\le \sum_{z=z_0+1}^{z_1}\exp\left[z\left(C\frac{z^2\log^2k}{d}-\log\left(\frac{z}{k\rho}\right)+1\right)\right],
\end{align}
where the last inequality follows from the identity $(1-x)\log(1-x)\ge -x$. For any fixed $a>0$, the function $f(x)=a x^2-\log x$ is decreasing on $(0,1/\sqrt{2a})$ and increasing on $(1/\sqrt{2a},\infty)$. Consequently, for any $z\in\{z_0+1,\ldots,k\}$, we have
\begin{align}
    C\frac{z^2\log^2k}{d}-\log\left(\frac{z}{k\rho}\right)\le -\omega,
\end{align}
where
\begin{align}
    \omega\triangleq\min\ppp{C\frac{z_0^2\log^2k}{d}-\log\left(\frac{z_0}{k\rho}\right),C\frac{z_1^2\log^2k}{d}-\log\left(\frac{z_1}{k\rho}\right)}.
\end{align}
Focusing first on the first quantity in the minimum, we note that in our regime $(k^2\log^2 k)/d \to 0$, and hence $(z_0^2\log^2 k)/d \to 0$ as well. Moreover, $\log(z_0/(k\rho)) = \log(f(k)+o(1)) \to \infty$, and therefore the first term diverges. Turning to the second quantity, its first term converges to zero by the same reasoning. In addition, since $z_1/(k\rho) = (d n^3/(k^6\log^2 k))^{1/3}$ and $d n^3/(k^6\log^2 k)\to\infty$, it follows that $\log(z_1/(k\rho))\to\infty$. This implies that
\begin{align}
    \E_Z\left[\exp\left(C\frac{Z^3\log^2k}{d}\right)\mathds{1}_{\{z_0<Z\le z_1\}}\right]\le \sum_{z=z_0+1}^{z_1}e^{-\omega z}\le\frac{e^{-\omega/2}}{1-e^{-\omega/2}}=o(1),
\end{align}
which consequently leads to
\begin{align}
\label{eq:z-small}
    \E_Z\left[\exp\left(C\frac{Z^3\log^2k}{d}\right)\mathds{1}_{\{Z\le z_1\}}\right]\le 1+o(1).
\end{align}

We now turn to the second expectation on the right-hand side of \eqref{eqn:Z3dFirst}. We again use the stochastic dominance of the binomial distribution over the hypergeometric distribution, together with a Chernoff bound, to obtain
\begin{align}
\E_Z\left[\exp\left(C\frac{Z^3\log^2k}{d}\right)\mathds{1}_{\{Z\ge z_1\}}\right]
&\le \sum_{z=z_1+1}^ke^{-\omega z},
\end{align}
where
\begin{align}
    \omega\triangleq\min\ppp{C\frac{z_1^2\log^2k}{d}-\log\left(\frac{z_1}{k\rho}\right),C\frac{k^2\log^2k}{d}-\log\left(\frac{1}{\rho}\right)}.
\end{align}
We have already shown that the first term in the minimum diverges under the condition $d n^3/(k^6\log^2 k)\to\infty$. For the second term, we distinguish between the regimes $k=o(n)$ and $k=\Theta(n)$.

For $k=o(n)$, the second term clearly diverges due to the growth of $\log(1/\rho)=\log(n/k+o(1))$ together with the assumption $d\gg k^2\log^2 k$. Hence $\omega\to\infty$, which implies
\begin{align}
\label{eq:z-big}
\E_Z\left[\exp\left(C\frac{Z^3\log^2k}{d}\right)\mathds{1}_{\{Z > z_1\}}\right] = o(1).
\end{align}
Combining \eqref{eq:z-small} and \eqref{eq:z-big}, we conclude that under the conditions of Theorem~\ref{thm:ITlower} but with $\frac{dn^3}{k^6\log^2 k}\to\infty$ replaced by $\frac{dn^3}{k^6 f(k)}\to\infty$ for some function $f(k)=\omega(\log^2 k)$, $\E_{\calH_0}[\tilde{\calL}^2] \le 1 + o(1)$, and hence weak detection is impossible. Similarly, under the conditions of Theorem~\ref{thm:ITlower}, equations \eqref{eq:z-small-crude} and \eqref{eq:z-big} yield $\E_{\calH_0}[\tilde{\calL}^2] = O(1)$, and therefore strong detection is impossible.

Finally, when $k=\Theta(n)$, we use the fact that $Z\leq k$ to obtain
\begin{align}
    \E_Z\left[\exp\left(C\frac{Z^3\log^2k}{d}\right)\right] \le \exp\left(C\frac{k^3\log^2k}{d}\right).
\end{align}
Then, under the conditions of Theorem~\ref{thm:ITlower}, namely $(d n^3)/(k^6\log^2 k)\to\infty$, when $k=\Theta(n)$ we have $k^3\log^2 k/d\to 0$. It follows that $\E_{\calH_0}[\tilde{\calL}^2]\le 1+o(1)$, and hence both weak and strong detection are impossible. This concludes the proof of Theorem~\ref{thm:ITlower}.

\subsection{Computational bounds}

In this subsection we prove Theorem~\ref{thm:gap}, starting with the lower bound component, showing when $\|\calL_{n,\leq D}\|_{\calH_0}^2$ remains bounded.

\paragraph{Low-degree computational lower bound.} Recall that $L^2(\calH_0)$ is the Hilbert space of random variables over the probability space on which $\P_{\calH_0}$ is defined, with a finite second moment, and equipped with the inner product
\begin{align}
    \innerP{\varphi(\s{G}),\psi(\s{G})}_{\calH_0}\triangleq  \E_{\calH_0}\pp{\varphi(\s{G})\cdot \psi(\s{G})}.
\end{align}
For an edge $e=\{i,j\}$ define the centered and normalized edge variable
\begin{align}
\psi_{ij}(\s{G})\triangleq \frac{\s{G}_{ij}-q}{\sqrt{q(1-q)}}.
\end{align}
For any non-empty subset of edges in the complete graph $\alpha\subseteq \binom{[n]}{2}$, define the Fourier character $\chi_\alpha$ as
\begin{align}
\chi_\alpha(\s{G})\triangleq \prod_{e\in\alpha}\psi_e(\s{G})=\prod_{\{i,j\}\in\alpha}\frac{\s{G}_{ij}-q}{\sqrt{q(1-q)}}.\label{eqn:Fouriercharacter} 
\end{align}
and $\chi_{\emptyset}(\s{G})\equiv1$, for each $\s{G}\in\{0,1\}^{\binom{n}{2}}$. Note that any subset of edges $\alpha\subseteq \binom{[n]}{2}$ induces a subgraph $H_\alpha = (V(\alpha),\alpha)$ containing no isolated vertices. In fact, there is a one-to-one correspondence between subgraphs without isolated vertices and subsets of edges. We therefore identify each character $\chi_\alpha$ with a subgraph $H_\alpha$. Observe that $\chi_\alpha$ is polynomial in the entries of $\s{G}$, with degree $|e(\alpha)|$, which, with slight abuse of notation, we also denote by $|\alpha|$. It is well known (and easy to verify) that the set $\ppp{\chi_{\alpha}}_{\alpha\subseteq \binom{[n]}{2}}$ forms an orthonormal basis for $L^2(\calH_0)$.

Define the likelihood ratio $\calL_n\triangleq \frac{\mathrm{d}\mathbb{P}_{\calH_1}}{d\mathbb{P}_{\calH_0}}$. For $D\geq 0$, let $\calV_{n,\leq D}$ be the subspace of polynomials of total degree at most $D$ with respect to the edge coordinates, and let $\calP_{\leq D}$ be the orthogonal projection onto $\calV_{n,\leq D}$ in $L^2(\mathbb{P}_{\calH_0})$. The degree-$D$ truncated likelihood ratio is
\begin{align}
\calL_{n,\leq D}\triangleq \calP_{\leq D}\calL_n.
\end{align}
By Parseval's identity we have
\begin{align}
\|\calL_{n,\leq D}\|_{\calH_0}^2 = \sum_{|\alpha|\leq D}\langle \calL_n,\chi_\alpha\rangle_{\calH_0}^2
= \sum_{|\alpha|\leq D}\p{\bE_{\calH_1}[\chi_\alpha(\s{G})]}^2
= 1+\sum_{1\le|\alpha|\leq D}\p{\bE_{\calH_1}[\chi_\alpha(\s{G})]}^2.\label{lem:orthonormal_basis_quiet}
\end{align}
Thus it suffices to control
\begin{align}
\calE_D\triangleq \sum_{1\leq|\alpha|\leq D}\p{\bE_{\calH_1}[\chi_\alpha(\s{G})]}^2.
\end{align}

Fix $\alpha\subseteq\binom{[n]}{2}$, and let $V(\alpha)\subseteq[n]$ denote the set of vertices incident to edges in $\alpha$, with $v(\alpha)\triangleq |V(\alpha)|$. For $i,j\in\calK$ define the planted, centered and normalized variable
\begin{align}
Z_{ij}\triangleq \frac{\sigma_{ij}-q}{\sqrt{q(1-q)}}.
\end{align}
We have the following result.
\begin{lemma}[Coefficient formula]
\label{lem:coeff_formula_quiet}
For every $\alpha\subseteq\binom{[n]}{2}$,
\begin{align}
\bE_{\calH_1}[\chi_\alpha(\s{G})] = \frac{(k)_{v(\alpha)}}{(n)_{v(\alpha)}}\bE\pp{\prod_{\{i,j\}\in\alpha} Z_{ij}},
\end{align}
where the expectation on the right is taken over i.i.d. latent vectors $\{\mathbf{x}_1,\dots,\mathbf{x}_{v(\alpha)}\}\sim\s{Unif}(\mathbb{S}^{d-1})$ and $\sigma_{ij}=\mathds{1}\{\langle \mathbf{x}_i,\mathbf{x}_j\rangle\geq t_{q,d}\}$.
\end{lemma}

\begin{proof}
Condition on $(\calK,\{\mathbf{x}_i\}_{i\in\calK})$. Under this conditioning, if an edge $e=\{i,j\}$ is not fully contained in $\calK$, then $\s{G}_{ij}\sim \s{Bern}(q)$ and hence $\bE[\psi_{ij}(\s{G})\vert\calK,\mathbf{x}]=0$. If $i,j\in\calK$ then $\s{G}_{ij}=\sigma_{ij}$ deterministically and
$\psi_{ij}(\s{G})=Z_{ij}$. Therefore,
\begin{align}
\bE\pp{\chi_\alpha(\s{G})\vert \calK,\mathbf{x}} =
\begin{cases}
\prod_{e\in\alpha} Z_e, & \s{if}\;V(\alpha)\subseteq \calK,\\
0, & \s{otherwise}.
\end{cases}
\end{align}
Taking expectation gives
\begin{align}
\bE_{\calH_1}[\chi_\alpha(\s{G})] =\bE\pp{\mathds{1}\{V(\alpha)\subseteq\calK\}\prod_{e\in\alpha}Z_e}
=\mathbb{P}(V(\alpha)\subseteq\calK)\cdot \bE\prod_{e\in\alpha}Z_e,
\end{align}
using independence of $\calK$ and the latent vectors. Finally, for any
fixed $S\subseteq[n]$ with $|S|=v(\alpha)$,
\begin{align}
\mathbb{P}(S\subseteq\calK)=\frac{\binom{n-v(\alpha)}{k-v(\alpha)}}{\binom{n}{k}}
=\frac{(k)_{v(\alpha)}}{(n)_{v(\alpha)}},
\end{align}
and by exchangeability the latent expectation depends only on the unlabeled
structure of $\alpha$. This yields the stated formula.
\end{proof}
Combining \eqref{lem:orthonormal_basis_quiet} and Lemma~\ref{lem:coeff_formula_quiet} yields
\begin{align}
\calE_D =\sum_{1\leq|\alpha|\leq D}\p{\frac{(k)_{v(\alpha)}}{(n)_{v(\alpha)}}}^2\p{\bE\prod_{e\in\alpha} Z_e}^2.
\end{align}

We proceed by considering two regimes: $k\leq n^{1/2-\epsilon}$, for arbitrary $\epsilon>0$, and the complementary regime. We start with the former and show that $\calE_D=O(1)$ uniformly over all $d$. To this end, since $\sigma_e\leq1$, we have
\begin{align}
    |Z_e|\leq \frac{1-q}{\sqrt{q(1-q)}} = \sqrt{\frac{1-q}{q}}.
\end{align}
Consequently, 
\begin{align}
    \calE_D\leq \sum_{1\leq|\alpha|\leq D}\p{\frac{(k)_{v(\alpha)}}{(n)_{v(\alpha)}}}^2\p{\frac{1-q}{q}}^{|\alpha|}.\label{eqn:cliqueLDP}
\end{align}
The right-hand side of \eqref{eqn:cliqueLDP} is exactly the same expression that arises in the low-degree analysis of the planted clique model. Indeed, this is not surprising: by bounding $|Z_e|$ by $\sqrt{\frac{1-q}{q}}$, we effectively dominate the geometric model by planted clique model. It is well-known (see, e.g., \cite{Hopkins18}) that for the planted clique model, the low-degree second moment is bounded, i.e., $\calE_D=O(1)$, if and only if $k\leq n^{1/2-\epsilon}$, uniformly over all $d$. This establishes the first part of Theorem~\ref{thm:gap}.

Next, we consider the more challenging regime where $k=\Omega(\sqrt{n})$. a key property of the centered variables $Z_{ij}$ is that products over graphs with \emph{leaves} have zero expectation, as proved next.
\begin{lemma}
\label{lem:leafvanish}
Let $H=(V,E)$ be a finite simple graph on vertex set $V$ and edge set $E$. Let $\{\mathbf{x}_v\}_{v\in V}$ be i.i.d. $\s{Unif}(\mathbb{S}^{d-1})$ and define $\sigma_{uv}=\mathds{1}\{\langle \mathbf{x}_u,\mathbf{x}_v\rangle\geq t_{q,d}\}$ and $Z_{uv}=(\sigma_{uv}-q)/\sqrt{q(1-q)}$. Then if $H$ has a leaf vertex, i.e., there exists $u\in V$ with $\deg_H(u)=1$, we have
\begin{align}
\bE\pp{\prod_{\{i,j\}\in E} Z_{ij}}=0.
\end{align}
\end{lemma}

\begin{proof}[Proof of Lemma~\ref{lem:leafvanish}]
Let $u$ be a leaf and let $v$ be its unique neighbor, so that $e=\{u,v\}$ is the unique edge incident to $u$. Condition on all latent vectors except $\mathbf{x}_u$. Then the product factors as
\begin{align}
\prod_{f\in E} Z_f= Z_{uv}\cdot \prod_{f\in E\setminus\{e\}} Z_f,
\end{align}
where $\prod_{f\in E\setminus\{e\}} Z_f$ is measurable with respect to $\{\mathbf{x}_w\}_{w\in V\setminus\{u\}}$. Therefore,
\begin{align}
\bE\pp{\prod_{f\in E} Z_f} = \bE\pp{\p{\prod_{f\in E\setminus\{e\}} Z_f}\cdot\bE\pp{Z_{uv}\vert \{\mathbf{x}_w\}_{w\in V\setminus\{u\}}}}.
\end{align}
By rotational symmetry, conditional on $\mathbf{x}_v$, the random vector $\mathbf{x}_u$ is uniform
on $\mathbb{S}^{d-1}$, and the event $\{\langle \mathbf{x}_u,\mathbf{x}_v\rangle\geq t_{q,d}\}$ has
probability $q$. Hence
\begin{align}
\bE[\sigma_{uv}\vert \mathbf{x}_v]=q \quad\Longrightarrow\quad \bE[Z_{uv}\vert \mathbf{x}_v]=\frac{q-q}{\sqrt{q(1-q)}}=0.
\end{align}
Thus $\bE[Z_{uv}\vert\{\mathbf{x}_w\}_{w\in V\setminus\{u\}}]=0$, and the whole expectation is zero.
\end{proof}
Accordingly,  only edge sets $\alpha$ whose associated graph has minimum degree at least two contribute to $\calE_D$. The challenge is to bound the remaining contributions. Fortunately enough, in a recent paper \cite{BangachevBresler2024FourierRGG}, the following result was proved.
\begin{lemma}[{\cite[Thm. 1.1 \& Prop. 1.2]{BangachevBresler2024FourierRGG}}]
\label{lem:BB_bound}
There exist constants $A,B>0$ depending only on $q$ such that for every connected graph $H=(V,E)$ with $v\triangleq |V|\geq 2$ and $m\triangleq |E|\geq 1$, we have
\begin{align}
\abs{\bE\pp{\prod_{e\in E(H)} Z_e}} \leq A^{m}\p{B\frac{v\cdot m\cdot (\log d)^{3/2}}{\sqrt d}}^{\left\lceil\frac{v-1}{2}\right\rceil}.
\end{align}
\end{lemma}
Now, denote the connected component of $\alpha$ by $\alpha_1,\dots,\alpha_c$. Since these components are supported on disjoint vertex sets, then the latent vectors on different components are independent, and so
\begin{align}
\bE\prod_{e\in \alpha}Z_e = \prod_{r=1}^c \bE\prod_{e\in \alpha_r}Z_e.
\end{align}
Applying Lemma~\ref{lem:BB_bound} to each component gives
\begin{align}
\label{eq:mu_disconnected_bound}
\abs{\bE\prod_{e\in E(H)}Z_e}\leq A^{|\alpha|}\p{B\frac{|\alpha|v(\alpha)(\log d)^{3/2}}{\sqrt d}}^{\sum_{r=1}^c\lceil(v(\alpha_r)-1)/2\rceil}.
\end{align}
Note that $\sum_{r=1}^c\lceil(v(\alpha_r)-1)/2\rceil\geq \frac{v(\alpha)-c}{2}$. Because we assume that $\alpha$ has minimum degree at least two, then each connected component has at least three vertices, so $c\leq v(\alpha)/3$ and therefore
\begin{align}
\label{eq:sH_lowerbound}
\abs{\bE\prod_{e\in E(H)}Z_e}\leq A^{|\alpha|}\p{B\frac{|\alpha|v(\alpha)(\log d)^{3/2}}{\sqrt d}}^{v(\alpha)/3}.
\end{align}

We are now in a position to bound $\calE_D$. Define
\begin{align}
\beta_n\triangleq \frac{k^6}{n^3 d}(\log d)^3.
\end{align}
Suppose that there exists $\varepsilon>0$ such that for all sufficiently large $n$,
\begin{align}
\label{eq:beta_poly_small}
\beta_n=\frac{k^6}{n^3 d}(\log d)^3 \leq n^{-\varepsilon}.
\end{align}
Let $D=D(n)$ satisfy $D\leq c\log n$ for a sufficiently small constant $c>0$ (depending only on $p$ and $\varepsilon$). Fix integers $v\geq 3$ and $m\geq 1$. Consider any $\alpha$ with $v(\alpha)=v$ and
$|\alpha|=m$ and such that $H_\alpha = (V(\alpha),\alpha)$ has minimum degree at least $2$. First, Stirling's approximation gives
\begin{align}
\frac{(k)_{v}}{(n)_{v}}\leq \p{\frac{k}{n}}^v.
\end{align}
Then, \eqref{eq:sH_lowerbound} gives
\begin{align}
\label{eq:term_bound}
\p{\frac{(k)_v}{(n)_v}}^2\p{\bE\prod_{e\in\alpha}Z_e}^2\leq\p{\frac{k}{n}}^{2v}A^{2m}\p{B^2\frac{v^2m^2(\log d)^3}{d}}^{v/3}.
\end{align}
Next, we bound the number of edge-sets $\alpha$ with $v(\alpha)=v$ and $|\alpha|=m$. Choose the $v$ vertices in $\binom{n}{v}$ ways, and then choose the $m$ edges among the $\binom{v}{2}$ possible edges on those vertices in $\binom{\binom{v}{2}}{m}$ ways. Hence
\begin{align}
\abs{\{\alpha\subseteq\binom{[n]}{2}:v(\alpha)=v,\; |\alpha|=m\}}\leq \binom{n}{v}\binom{\binom{v}{2}}{m}.
\end{align}
Note that
\begin{align}
\binom{n}{v}\leq \p{\frac{en}{v}}^v,\quad\binom{\binom{v}{2}}{m}\leq \p{\frac{e\binom{v}{2}}{m}}^m\leq \p{\frac{e v^2}{2m}}^m.
\end{align}
Because we assume that $H_\alpha$ has minimum degree at least two, then $2m\geq 2v$ and hence $m\geq v$. Therefore $v^2/m\leq v$ and thus there exists an absolute constant $C_0>0$ such that
\begin{align}
\binom{\binom{v}{2}}{m}\leq (C_0 v)^m.
\end{align}
Combining these gives in our case
\begin{align}
\label{eq:count_bound}
\abs{\{\alpha\subseteq\binom{[n]}{2}:v(\alpha)=v,\; |\alpha|=m\}}\leq \p{\frac{en}{v}}^v (C_0 v)^m.
\end{align}

Let us sum \eqref{eq:term_bound} over all $\alpha$ with fixed $(v,m)$:
\begin{align}
\sum_{\substack{\alpha:\ v(\alpha)=v\\ |\alpha|=m}}\p{\frac{(k)_v}{(n)_v}}^2 \p{\bE\prod_{e\in\alpha}Z_e}^2&\leq \p{\frac{en}{v}}^v (C_0 v)^m \p{\frac{k}{n}}^{2v} A^{2m}\p{B^2\frac{v^2m^2(\log d)^3}{d}}^{v/3}\\
&\leq \p{\frac{k^2}{n}}^v\p{B^2\frac{v^2m^2(\log d)^3}{d}}^{v/3}(C_1 v)^m,
\end{align}
where $C_1=C_0A^2$ and we used the fact that $v\geq3$. Note that
\begin{align}
\p{\frac{k^2}{n}}^v\p{\frac{1}{d}}^{v/3}(\log d)^v = \p{\frac{k^6}{n^3 d}(\log d)^3}^{v/3} = \beta_n^{v/3}.
\end{align}
Therefore,
\begin{align}
\sum_{\substack{\alpha:\ v(\alpha)=v\\ |\alpha|=m}}\p{\frac{(k)_v}{(n)_v}}^2 \p{\bE\prod_{e\in\alpha}Z_e}^2&\leq \p{\frac{k^6}{n^3 d}(\log d)^3}^{v/3}
\p{B^2v^2m^2}^{v/3}(C_1 v)^m \\
&= \beta_n^{v/3}(B^2v^2m^2)^{v/3}(C_1 v)^m.
\end{align}

Finally, we sum over $m$ and $v$. Since the minimum degree is two, we must have $m\geq v$, and we
also have $m\leq D$ because we only consider $|\alpha|\leq D$. Thus
\begin{align}
\calE_D \leq \sum_{v=3}^{D}\sum_{m=v}^{D}\beta_n^{v/3}(B^2v^2m^2)^{v/3}(C_1 v)^m \leq
\sum_{v=3}^{D}\beta_n^{v/3}(B^2v^2D^2)^{v/3}\sum_{m=v}^{D}(C_1 v)^m.
\end{align}
Since $v\leq D$, we have $(C_1 v)^m\leq (C_1 D)^m$, and hence
\begin{align}
\sum_{m=v}^{D}(C_1 v)^m\leq \sum_{m=0}^{D}(C_1 D)^m \leq (D+1)(C_1 D)^D\leq (C_2 D)^D,
\end{align}
for a constant $C_2>0$ depending only on $C_1$. Therefore,
\begin{align}
\calE_D \leq (C_2 D)^D\sum_{v=3}^{D}\beta_n^{v/3}(B^2v^2D^2)^{v/3}
\leq (C_3 D^{2.5})^D\cdot \frac{\beta_n}{1-\beta_n^{1/3}}
\leq 2(C_3 D^{2.5})^D\beta_n,
\end{align}
for a constant $C_3>0$, and for all sufficiently large $n$ since $\beta_n\to 0$. Assume now \eqref{eq:beta_poly_small}, i.e., $\beta_n\leq n^{-\varepsilon}$ for all large $n$. Let $D\leq c\log n/\log\log n$ where $c>0$ is chosen sufficiently small (as specified below). Then
\begin{align}
(C_3 D^{2.5})^D\beta_n \leq \exp\p{3D\log D}\cdot n^{-\varepsilon}
\leq \exp\p{3c\frac{\log n}{\log\log n}\cdot \log\log n}\cdot n^{-\varepsilon} = n^{3c-\varepsilon}.
\end{align}
Thus for a sufficiently small absolute constant $c$ such that $c<\varepsilon/3$, we have $(C_3 D^{2.5})^D\beta_n=o(1)$, and hence $\calE_D=o(1)$. Therefore
\begin{align}
\|\calL_{n,\leq D}\|_{\calH_0}^2=1+\calE_D\leq 1+o(1),
\end{align}
which implies $\|\calL_{n,\leq D}\|_{\calH_0}=O(1)$.

\paragraph{Low-degree computational upper bound.} Finally we prove that when the dimension is below the signed-triangle threshold, the low-degree norm diverges. The proof uses only the first nonzero coefficients (degree-$1$ vanishes, and graphs with leaves vanish), and in particular it suffices to lower bound the triangle contribution. Specifically, let $\s{Tri}$ denote the set of edge sets $\alpha\subseteq\binom{[n]}{2}$ that form a triangle on three vertices, i.e. $\alpha=\{\{i,j\},\{i,\ell\},\{j,\ell\}\}$, for some distinct $i,j,\ell$.

For any $D\ge 3$, by \eqref{lem:orthonormal_basis_quiet},
\begin{align}
\label{eq:triangleLowerboundStart}
\norm{\calL_{n,\le D}}_{\calH_0}^2
\;\ge\;
1+\sum_{\alpha\in\s{Tri}}\p{\bE_{\calH_1}[\chi_\alpha(\s{G})]}^2.
\end{align}
Thus it suffices to show that the triangle sum diverges. Fix $\alpha=\{\{1,2\},\{1,3\},\{2,3\}\}$. Then $v(\alpha)=3$, so by Lemma~\ref{lem:coeff_formula_quiet},
\begin{align}
\label{eq:triangleCoeff}
\E_{\calH_1}[\chi_\alpha(\s{G})]=\frac{(k)_3}{(n)_3}\E[Z_{12}Z_{13}Z_{23}].
\end{align}
By exchangeability, the same value holds for every labeled triangle $\alpha\in\s{Tri}$. Since $|\s{Tri}|=\binom{n}{3}$, we obtain
\begin{align}
\label{eq:triangleSumExact}
\sum_{\alpha\in\s{Tri}}\p{\bE_{\calH_1}[\chi_\alpha(\s{G})]}^2 = \binom{n}{3}\p{\frac{(k)_3}{(n)_3}}^2
\p{\E[Z_{12}Z_{13}Z_{23}]}^2.
\end{align}
Therefore, the proof reduces to a quantitative lower bound on the \emph{signed triangle moment}
$\E[Z_{12}Z_{13}Z_{23}]$ for the spherical random geometric graph. It is rather classical that for every $0<q<1$, there exists a constant $c_q>0$ such that for all $n$ and $d$ we have (see, e.g., \cite[Lemma 3]{bubeck2016testing})
\begin{align}
\label{eq:DeltaLowerbound}
\E[Z_{12}Z_{13}Z_{23}]\geq  \frac{c_q}{\sqrt d}.
\end{align}
Fix any $D\ge 3$. By \eqref{eq:triangleLowerboundStart}, \eqref{eq:triangleSumExact}, and \eqref{eq:DeltaLowerbound},
\begin{align}
\norm{\calL_{n,\le D}}_{\calH_0}^2\geq 1+\binom{n}{3}\p{\frac{(k)_3}{(n)_3}}^2\frac{c_q^2}{d}.
\end{align}
Note that $\binom{n}{3}\geq n^3/6$ for all $n\ge 3$, and
\begin{align}
\frac{(k)_3}{(n)_3}=\frac{k(k-1)(k-2)}{n(n-1)(n-2)}=(1+o(1))\p{\frac{k}{n}}^3,
\end{align}
whenever $k,n\to\infty$. 
Hence
\begin{align}
\binom{n}{3}\p{\frac{(k)_3}{(n)_3}}^2 = \frac{(1+o(1))}{6}\frac{k^6}{n^3}.
\end{align}
Combining,
\begin{align}
\norm{\calL_{n,\le D}}_{\calH_0}^2 \geq 1 + \frac{(1+o(1))c_q^2}{6}\frac{k^6}{n^3 d}.
\end{align}
Thus, if $d\ll k^6/n^3$, then the right-hand side diverges to $\infty$, implying $\norm{\calL_{n,\leq D}}_{\calH_0}=\omega(1)$.

\section{Conclusion and Outlook}
This work provides a sharp characterization of the statistical and computational limits of detecting a planted high-dimensional geometric subgraph in a random graph. While our results focus on a specific geometric model, they open several directions for future research.

\begin{itemize}
    \item It would be natural to extend our analysis beyond the Euclidean sphere to other probabilistic metric spaces. For example, one could consider latent positions drawn from more general manifolds, anisotropic distributions, or non-Euclidean geometries, and investigate how curvature, intrinsic dimension, or metric structure affect detectability and computational hardness.
    \item Our model, as is standard in the literature, assumes that the latent feature vectors are completely unobserved. In many practical settings, however, one has partial access to side information about the latent geometry. For instance, one might observe a quantized or noisy version of the feature vectors, or only a subset of their coordinates. Understanding how such partial observations alter the statistical thresholds and computational barriers is an interesting open problem. In particular, it is unclear whether even coarse or noisy geometric side information can significantly reduce the sample complexity or break the computational barriers identified here.
    \item While we focus on detection, the recovery problem remain largely unexplored in the geometric setting. This includes identifying the planted vertex set, estimating latent positions, or designing adaptive querying procedures that exploit sequential access to the graph. These problems are likely to exhibit their own statistical–computational tradeoffs, distinct from those governing detection.
    \item Although we primarily study the geometry-only regime in which the null and alternative hypotheses have matching edge marginals, it is also of interest to understand settings where the edge densities differ under the null and alternative. In the dense regime with fixed probabilities, this reduces to classical planted dense subgraph detection, where geometry plays a limited role. However, when edge probabilities depend on the graph size, such as in sparse or vanishing density regimes, the interaction between density contrast and geometry becomes substantially more intricate. Analyzing these regimes, where both sparsity and latent geometry are present, poses significant challenges and may reveal new phase transitions beyond those identified in this work.
\end{itemize}

\bibliographystyle{alpha}
\bibliography{bibfile}

@article{bubeck2016testing,
  title={Testing for high-dimensional geometry in random graphs},
  author={Bubeck, S{\'e}bastien and Ding, Jian and Eldan, Ronen and R{\'a}cz, Mikl{\'o}s Z},
  journal={Random Structures \& Algorithms},
  volume={49},
  number={3},
  pages={503--532},
  year={2016},
  publisher={Wiley Online Library}
}

@article{kim2004divide,
  title={Divide and conquer martingales and the number of triangles in a random graph},
  author={Kim, Jeong Han and Vu, Van H},
  journal={Random Structures \& Algorithms},
  volume={24},
  number={2},
  pages={166--174},
  year={2004},
  publisher={Wiley Online Library}
}

@article{brennan2020phase,
  title={Phase transitions for detecting latent geometry in random graphs},
  author={Brennan, Matthew and Bresler, Guy and Nagaraj, Dheeraj},
  journal={Probability Theory and Related Fields},
  volume={178},
  number={3},
  pages={1215--1289},
  year={2020},
  publisher={Springer}
}

@article{liu2023phase,
  title        = {Phase transition in noisy high-dimensional random geometric graphs},
  author       = {Liu, Suqi and R\'acz, Mikl\'os Z.},
  journal      = {Electronic Journal of Statistics},
  volume       = {17},
  number       = {2},
  pages        = {3512--3574},
  year         = {2023},
  doi          = {10.1214/23-EJS2162},
  url          = {https://doi.org/10.1214/23-EJS2162},
}

@article{liu2023probabilistic,
  title        = {A probabilistic view of latent space graphs and phase transitions},
  author       = {Liu, Suqi and R\'acz, Mikl\'os Z.},
  journal      = {Bernoulli},
  volume       = {29},
  number       = {3},
  pages        = {2417--2441},
  year         = {2023},
  doi          = {10.3150/22-BEJ1547},
  url          = {https://doi.org/10.3150/22-BEJ1547},
}

@article{devroye2011high,
  title        = {High-dimensional random geometric graphs and their clique number},
  author       = {Devroye, Luc and Gy{\"o}rgy, Andr{\'a}s and Lugosi, G{\'a}bor and Udina, Frederic},
  journal      = {Electronic Journal of Probability},
  volume       = {16},
  pages        = {2481--2508},
  year         = {2011},
  doi          = {10.1214/EJP.v16-967},
  url          = {https://doi.org/10.1214/EJP.v16-967},
}

@article{butucea2013detection,
	Author = {Butucea, Cristina and Ingster, Yuri I},
	Journal = {Bernoulli},
	Number = {5B},
	Pages = {2652--2688},
	Publisher = {Bernoulli Society for Mathematical Statistics and Probability},
	Title = {Detection of a sparse submatrix of a high-dimensional noisy matrix},
	Volume = {19},
	Year = {2013}}

@article{Ingster1997HypothesisTesting,
  author  = {Ingster, Yuri I.},
  title   = {Some Problems of Hypothesis Testing Leading to Infinitely Divisible Distributions},
  journal = {Mathematical Methods of Statistics},
  volume  = {6},
  number  = {1},
  pages   = {47--69},
  year    = {1997}
}

@article{WuXuYu2023,
  author    = {Yihong Wu and Jiaming Xu and Sophie H. Yu},
  title     = {Testing correlation of unlabeled random graphs},
  journal   = {Annals of Applied Probability},
  volume    = {33},
  number    = {4},
  pages     = {2519--2558},
  year      = {2023},
}

@article{hopkins2017power,
  title={The power of sum-of-squares for detecting hidden structures},
  author={Hopkins, Samuel B and Kothari, Pravesh K and Potechin, Aaron and Raghavendra, Prasad and Schramm, Tselil and Steurer, David},
  journal={Proceedings of the fifty-eighth IEEE Foundations of Computer Science (FOCS)},
  year={2017},
  booktitle = {Foundations of Computer Science (FOCS), 2017 IEEE 58th Annual Symposium on},
  organization = {IEEE},
  pages = {720--731},
}

@INPROCEEDINGS{barak2016nearly,  author={B. {Barak} and S. B. {Hopkins} and J. {Kelner} and P. {Kothari} and A. {Moitra} and A. {Potechin}},  booktitle={2016 IEEE 57th Annual Symposium on Foundations of Computer Science (FOCS)},   title={A Nearly Tight Sum-of-Squares Lower Bound for the Planted Clique Problem},   year={2016},  volume={},  number={},  pages={428-437},}

@INPROCEEDINGS{hopkins2017bayesian,
  author={S. B. {Hopkins} and D. {Steurer}},
  booktitle={2017 IEEE 58th Annual Symposium on Foundations of Computer Science (FOCS)}, 
  title={Efficient Bayesian Estimation from Few Samples: Community Detection and Related Problems}, 
  year={2017},
  volume={},
  number={},
  pages={379-390},}

@article{verzelen2015community,
  title={Community detection in sparse random networks},
  author={Verzelen, Nicolas and Arias-Castro, Ery},
  journal={The Annals of Applied Probability},
  volume={25},
  number={6},
  pages={3465--3510},
  year={2015},
  publisher={Institute of Mathematical Statistics}
}

@article{arias2014community,
  title={Community detection in dense random networks},
  author={Arias-Castro, Ery and Verzelen, Nicolas},
  journal={The Annals of Statistics},
  volume={42},
  number={3},
  pages={940--969},
  year={2014},
  publisher={Institute of Mathematical Statistics}
}

@InProceedings{brennan18a, title = {Reducibility and Computational Lower Bounds for Problems with Planted Sparse Structure}, author = {Brennan, Matthew and Bresler, Guy and Huleihel, Wasim}, booktitle = {Proceedings of the 31st Conference On Learning Theory}, pages = {48--166}, year = {2018}, volume = {75}, month = {06--09 Jul}, }

@InProceedings{brennan20a, title = {Reducibility and Statistical-Computational Gaps from Secret Leakage}, author = {Brennan, Matthew and Bresler, Guy}, booktitle = {Proceedings of Thirty Third Conference on Learning Theory}, pages = {648--847}, year = {2020}, volume = {125}, month = {09--12 Jul},}

@inproceedings{elimelech2025detecting,
  title        = {Detecting Arbitrary Planted Subgraphs in Random Graphs},
  author       = {Elimelech, Dor and Huleihel, Wasim},
  booktitle    = {Proceedings of the Thirty Eighth Conference on Learning Theory},
  series       = {Proceedings of Machine Learning Research},
  volume       = {291},
  pages        = {1691--1798},
  year         = {2025},
  publisher    = {PMLR},
}

@article{dettmann2016random,
  title = {Random geometric graphs with general connection functions},
  author = {Dettmann, Carl P. and Georgiou, Orestis},
  journal = {Phys. Rev. E},
  volume = {93},
  issue = {3},
  pages = {032313},
  numpages = {14},
  year = {2016},
  month = {Mar},
  publisher = {American Physical Society},
}

@article{BokLiYu2026,
  title   = {Detection of Local Geometry in Random Graphs: Information-Theoretic and Computational Limits},
  author  = {Bok, Jinho and Li, Shuangping and Yu, Sophie H.},
  journal = {arXiv preprint arXiv:2603.24545},
  year    = {2026}
}

@article{Bet2020Detecting,
  author  = {Bet, Gianmarco and Bogerd, Kay and Castro, Rui M. and van der Hofstad, Remco},
  title   = {Detecting a botnet in a network},
  journal = {Mathematics and Statistics Learning},
  volume  = {3},
  number  = {3--4},
  pages   = {315--343},
  year    = {2020}
}

@inproceedings{parthasarathy2017quest,
  title     = {A Quest to Unravel the Metric Structure Behind Perturbed Networks},
  author    = {Parthasarathy, Srinivasan and Sivakoff, David and Tian, Minghao and Wang, Yusu},
  booktitle = {33rd International Symposium on Computational Geometry (SoCG 2017)},
  series    = {Leibniz International Proceedings in Informatics (LIPIcs)},
  volume    = {77},
  pages     = {53:1--53:15},
  year      = {2017},
  publisher = {Schloss Dagstuhl--Leibniz-Zentrum f{\"u}r Informatik},
  doi       = {10.4230/LIPIcs.SoCG.2017.53},
}

@article{wilsher2020connectivity,
  title   = {Connectivity in one-dimensional soft random geometric graphs},
  author  = {Wilsher, Michael and Dettmann, Carl P. and Ganesh, Ayalvadi},
  journal = {Physical Review E},
  volume  = {102},
  number  = {6},
  pages   = {062312},
  year    = {2020},
  doi     = {10.1103/PhysRevE.102.062312},
}

@book{Mathew07,
 author = {Mathew Penrose},
 title = {Random Geometric Graphs},
 year = {2007},
 publisher = {Oxford University Press},
 address = {Oxford},
}

@phdthesis{Hopkins18,
	Author = {Hopkins, Samuel B.},
	School = {Cornell University},
	Title = {Statistical Inference and the Sum of Squares Method},
	Year = {2018}}

@incollection{Dmitriy19,
  title        = {Notes on Computational Hardness of Hypothesis Testing: Predictions using the Low-Degree Likelihood Ratio},
  author       = {Kunisky, Dmitriy and Wein, Alexander S. and Bandeira, Afonso S.},
  booktitle    = {Mathematical Analysis, its Applications and Computation},
  series       = {Springer Proceedings in Mathematics \& Statistics},
  volume       = {385},
  pages        = {1--50},
  year         = {2022},
  publisher    = {Springer},
  note         = {{ISAAC} Congress 2019},
  url          = {https://arxiv.org/abs/1907.11636},
}

@article{janson_oleszkiewicz_rucinski_2004,
  author  = {Janson, Svante and Oleszkiewicz, Krzysztof and Ruci\'nski, Andrzej},
  title   = {Upper Tails for Subgraph Counts in Random Graphs},
  journal = {Israel Journal of Mathematics},
  volume  = {142},
  pages   = {61--92},
  year    = {2004},
  doi     = {10.1007/BF02771528}
}

@article{ganguly_hiesmayr_nam_2024,
  author  = {Ganguly, Shirshendu and Hiesmayr, Elisabeth and Nam, Kyeongsik},
  title   = {Upper Tail Behavior of the Number of Triangles in Random Graphs with Constant Average Degree},
  journal = {Combinatorica},
  volume  = {44},
  pages   = {699--740},
  year    = {2024},
  doi     = {10.1007/s00493-024-00086-3},
  publisher = {Springer}
}

@book{BoucheronLugosiMassart2013,
  title     = {Concentration Inequalities: A Nonasymptotic Theory of Independence},
  author    = {Boucheron, St{\'e}phane and Lugosi, G{\'a}bor and Massart, Pascal},
  year      = {2013},
  publisher = {Oxford University Press},
  address   = {Oxford},
  isbn      = {9780199535255}
}

@article{eldan_mikulincer_2020,
  author  = {Eldan, Ronen and Mikulincer, Danny},
  title   = {Information and Dimensionality of Anisotropic Random Geometric Graphs},
  journal = {Annals of Probability},
  volume  = {48},
  number  = {4},
  pages   = {1980--2015},
  year    = {2020},
  doi     = {10.1214/19-AOP1406}
}

@article{brennan2024threshold,
  title        = {Threshold for Detecting High Dimensional Geometry in Anisotropic Random Geometric Graphs},
  author       = {Brennan, Matthew S. and Bresler, Guy and Huang, Brice},
  journal      = {Random Structures \& Algorithms},
  volume       = {64},
  number       = {1},
  pages        = {125--137},
  year         = {2024},
}

@article{brennan2021de,
  title        = {De Finetti–Style Results for Wishart Matrices: Combinatorial Structure and Phase Transitions},
  author       = {Brennan, Matthew and Bresler, Guy and Huang, Brice},
  journal      = {arXiv preprint},
  eprint       = {2103.14011},
  year         = {2021},
  archivePrefix= {arXiv},
  primaryClass = {math.PR},
  url          = {https://arxiv.org/abs/2103.14011},
}

@ARTICLE{huleihel2021inferring,
  author={Huleihel, Wasim},
  journal={IEEE Transactions on Signal and Information Processing over Networks}, 
  title={Inferring Hidden Structures in Random Graphs}, 
  year={2022},
  volume={8},
  number={},
  pages={855-867},
  doi={10.1109/TSIPN.2022.3211208}}

@article{wein2025computational,
  title        = {Computational Complexity of Statistics: New Insights from Low‑Degree Polynomials},
  author       = {Wein, Alexander S.},
  journal      = {arXiv preprint arXiv:2506.10748},
  year         = {2025},
  month        = {Jun},
  doi          = {10.48550/arXiv.2506.10748},
  url          = {https://arxiv.org/abs/2506.10748},
}

@book{tsybakov2004introduction,
  author    = {Tsybakov, Alexandre B.},
  title     = {Introduction to Nonparametric Estimation},
  publisher = {Springer},
  series    = {Springer Series in Statistics},
  year      = {2009},
  doi       = {10.1007/b13794}
}

@article{sason2014bounds,
  title={Bounds on f-divergences and related distances},
  author={Sason, Igal},
  journal={CCIT Report},
  number={859},
  year={2014},
  volume={859},
  publisher={Dept. of Electrical Engineering, Technion—Israel Inst. of Technology Haifa~…}
}

@book{delapena1999decoupling,
  title     = {Decoupling: From Dependence to Independence},
  author    = {Victor H. de la Pena and Gin{\'e}, E.},
  series    = {Probability and its Applications},
  publisher = {Springer},
  address   = {New York},
  year      = {1999},
}

@article{de_la_pena_montgomery_smith_1995,
author = {Victor H. de la Pena and S. J. Montgomery-Smith},
title = {{Decoupling Inequalities for the Tail Probabilities of Multivariate $U$-Statistics}},
volume = {23},
journal = {The Annals of Probability},
number = {2},
publisher = {Institute of Mathematical Statistics},
pages = {806 -- 816},
year = {1995},
}

@book{polyanskiywuIT,
  author    = {Polyanskiy, Yury and Wu, Yihong},
  title     = {Information Theory: From Coding to Learning},
  publisher = {Cambridge University Press},
  year      = {2024},
}

@misc{vershynin2011,
      title={Introduction to the non-asymptotic analysis of random matrices}, 
      author={Roman Vershynin},
      year={2011},
      eprint={1011.3027},
      archivePrefix={arXiv},
      primaryClass={math.PR},
      url={https://arxiv.org/abs/1011.3027}, 
}

@book{Anderson_Guionnet_Zeitouni_2009, place={Cambridge}, series={Cambridge Studies in Advanced Mathematics}, title={An Introduction to Random Matrices}, publisher={Cambridge University Press}, author={Anderson, Greg W. and Guionnet, Alice and Zeitouni, Ofer}, year={2009}, collection={Cambridge Studies in Advanced Mathematics}}

@article{Gutev_2025,
title = {On real-valued functions of Lipschitz type},
journal = {Expositiones Mathematicae},
volume = {43},
number = {5},
pages = {125701},
year = {2025},
author = {Valentin Gutev},
}

@inproceedings{BangachevBresler2024FourierRGG,
  title        = {On The Fourier Coefficients of High-Dimensional Random Geometric Graphs},
  author       = {Kiril Bangachev and Guy Bresler},
  booktitle    = {Proceedings of the 56th Annual ACM Symposium on Theory of Computing (STOC 2024)},
  year         = {2024},
  pages        = {549--560},
  address      = {Vancouver, BC, Canada},
  publisher    = {ACM},
  doi          = {10.1145/3596830.3602707}, 
  note         = {Also available as arXiv:2402.12589},
}

@book{vershynin2018high,
  title     = {High-Dimensional Probability: An Introduction with Applications in Data Science},
  author    = {Vershynin, Roman},
  publisher = {Cambridge University Press},
  year      = {2018},
  address   = {Cambridge},
  isbn      = {9781108415194},
}

\appendix

\section{Proof of lemmata}\label{app:proofsLemmas}

\begin{proof}[Proof of Lemma~\ref{lem:CS-Dm}]
Let $\mu$ be a measure such that $\P,\R,\Q\ll\mu$, and write $p \triangleq \frac{\mathrm{d}\P}{\mathrm{d}\mu}$, $r \triangleq \frac{\mathrm{d}\R}{\mathrm{d}\mu}$, and $q \triangleq \frac{\mathrm{d}\Q}{\mathrm{d}\mu}$. Then, $\frac{\mathrm{d}\P}{\mathrm{d}\Q} = \frac{p}{q}$, $\frac{\mathrm{d}\P}{\mathrm{d}\R} = \frac{p}{r}$, and $\frac{\mathrm{d}\R}{\mathrm{d}\Q} = \frac{r}{q}$. Hence, by definition
\begin{align}
1 + D_m(\P\|\Q)
&= \int \left(\frac{\mathrm{d}\P}{\mathrm{d}\Q}\right)^m \mathrm{d}\Q\\
 &  = \int \left(\frac{p}{q}\right)^m q \mathrm{d}\mu\\
 &  = \int p^m q^{1-m} \mathrm{d}\mu \\
&= \int p^m r^{1-m} \left(\frac{r}{q}\right)^{m-1} \mathrm{d}\mu\\
  & = \int \left(\frac{\mathrm{d}\P}{\mathrm{d}\R}\right)^m
            \left(\frac{\mathrm{d}\R}{\mathrm{d}\Q}\right)^{m-1} \mathrm{d}\R.
\end{align}
Applying Cauchy--Schwarz we get
\begin{align}
1 + D_m(P\|Q)
&\le
\sqrt{\int \left(\frac{\mathrm{d}\P}{\mathrm{d}\R}\right)^{2m} \mathrm{d}\R }
\sqrt{ \int \left(\frac{\mathrm{d}\R}{\mathrm{d}\Q}\right)^{2m-2} \mathrm{d}\R }.
\end{align}
Now, we note that
\begin{align}
\int \left(\frac{\mathrm{d}\P}{\mathrm{d}\R}\right)^{2m} \mathrm{d}\R
= 1 + D_{2m}(\P\|\R)
\le 1 + D_{2m}(\P\|\Q),
\end{align}
and
\begin{align}
\int \left(\frac{\mathrm{d}\R}{\mathrm{d}\Q}\right)^{2m-2} \mathrm{d}\R
= \int \left(\frac{\mathrm{d}\R}{\mathrm{d}\Q}\right)^{2m-1} \mathrm{d}\Q
= 1 + D_{2m-1}(\R\|\Q).
\end{align}
Thus
\begin{align}
1 + D_m(\P\|\Q)
\leq
\sqrt{1 + D_{2m}(\P\|\Q)}
\sqrt{1 + D_{2m-1}(\R\|\Q)},
\end{align}
as claimed.
\end{proof}

\section{Densest subgraph test}\label{app:apppnotq}

In this appendix, we consider the case where the edge probability inside the planted set is taken to be $p\neq q$, and propose an optimal detection algorithm. Specifically, consider the following statistic:
\begin{align}
    \s{T}_{\s{dense}}(\s{G}_n)&\triangleq\max_{\calS\subset[n]:|\calS|=k}\sum_{i,j\in\calS:i<j}\mathbf{A}_{i,j}.\label{eqn:scandense}
\end{align}
This statistic scans over all $k$-vertex subgraphs and selects the densest one in terms of edge count. 
Define the following detection algorithm:
\begin{align}
    \calA_{\s{dense}}(\s{G}_n) &\triangleq\Ind\ppp{\s{T}_{\s{dense}}(\s{G}_n)\geq \tau_{\s{dense}}},\label{eqn:denseTest}
\end{align}
where $\tau_{\s{dense}}\in\mathbb{R}_+$ is specified below. We have the following result.
\begin{theorem}[Upper bounds]\label{thm:upperBound2}
    Consider the detection problem in \eqref{eqn:model}, and the detection algorithm in  \eqref{eqn:denseTest}, with $\tau_{\s{dense}} = \binom{k}{2}\frac{p+q}{2}$. If $\frac{(p-q)^2}{q(1-q)} = \omega\p{\frac{\log n}{k}}$, then $\s{R}(\calA_{\s{dense}})\to0$, as $k,n\to\infty$.
\end{theorem}
When $p\neq q$ is fixed, Theorem~\ref{thm:upperBound2} shows that strong detection is achievable whenever $k=\omega(\log n)$.

\begin{proof}[Proof of Theorem~\ref{thm:upperBound2}]
    Recall the densest subgraph test defined in \eqref{eqn:scandense} and~\eqref{eqn:denseTest}, and let us analyze its $\s{Type}$-$\s{I}$ and $\s{Type}$-$\s{II}$ error probabilities. Set $\tau_{\s{dense}} = \binom{k}{2}\frac{p+q}{2}$. We begin with the analysis of the $\s{Type}$-$\s{I}$ error probability. For any
subset $\calS\subset[n]$ with $|\calS|=k$, define $e(\calS)\triangleq\sum_{i,j\in\calS:i<j}\bs{A}_{i,j}$. 
Under the null hypothesis $\calH_0$, we have $e(\calS)\sim \s{Binomial}\left(\binom{k}{2},q\right)$. Consequently, by the union bound together with Bernstein's inequality, we obtain
\begin{align}
    \pr_{\calH_0}\pp{\calA_{\s{dense}}(\s{G}_n) =1} &\leq \binom{n}{k}\cdot\pr\pp{\s{Binomial}\p{\binom{k}{2},q}\geq \tau_{\s{dense}}}\\
    &\leq \binom{n}{k}\exp\p{-\frac{\binom{k}{2}^2(p-q)^2/4}{2\binom{k}{2}q+\binom{k}{2}(p-q)/3}}\\
    &\leq \exp\pp{k\log n-\Omega\p{k^2\frac{(p-q)^2}{q(1-q)}}},\label{eqn:denseTypeI}
\end{align}
where the last inequality follows from Stirling's approximation. We now turn to the analysis of the $\s{Type}$-$\s{II}$ error probability. Denoting by $\calK$ the underlying planted vertex set, an application of Chebyshev’s inequality yields
\begin{align}
  \pr_{\calH_1}\pp{\calA_{\s{dense}}(\s{G}_n) =0}&\leq \pr_{\calH_1}\pp{e(\calK)<\tau_{\s{dense}}}\\
  &\leq \frac{\s{Var}_{\calH_1}(e(\calK))}{\binom{k}{2}^2\frac{(p-q)^2}{4}}.\label{eqn:ChebyshevMiddle}
\end{align}
Now,
\begin{align}
   \s{Var}_{\calH_1}(e(\calK)) &= \s{Var}_{\calH_1}\p{\sum_{i,j\in\calK:i<j}\mathbf{A}_{i,j}} \\
   & = \sum_{i,j\in\calK:i<j}\s{Var}_{\calH_1}(\mathbf{A}_{i,j})+\sum_{(i,j)\neq (i',j')}\s{cov}_{\calH_1}(\mathbf{A}_{i,j},\mathbf{A}_{i',j'}).
\end{align}
By definition, for any $i,j\in\calK$ we have $\s{Var}_{\calH_1}(\bs{A}_{i,j})=p(1-p)$. We claim that for any
$i,i',j,j'\in\calK$ with $i<j$ and $i'<j'$, such that $(i,j)\neq(i',j')$, it holds that
\begin{align}
\s{Cov}_{\calH_1}(\bs{A}_{i,j},\bs{A}_{i',j'})=0.
\end{align}
There are two cases to consider. First, if the indices $i,i',j,j'$ are all distinct, then $\bs{A}_{i,j}$ and $\bs{A}_{i',j'}$ are independent, and hence their covariance is zero. The second, less immediate case is when the two edges share a vertex, for instance when $i=i'$ and $j\neq j'$. In this case, we have
\begin{align}
    \s{cov}_{\calH_1}(\mathbf{A}_{i,j},\mathbf{A}_{i,j'}) &= \bE_{\calH_1}\pp{\mathbf{A}_{i,j}\cdot\mathbf{A}_{i,j'}}-p^2\\
     &= \P_{\calH_1}\pp{\left\langle \bs{x}_i,\bs{x}_j \right\rangle\geq t_{q,d},\left\langle \bs{x}_i,\bs{x}_{j'} \right\rangle\geq t_{q,d}}-p^2.\label{eqn:covmiddle}
\end{align}
By rotation invariance on the sphere, we can fix $\bs{x}_i = \bs{e}_1$, where $\bs{e}_1$ is the $d$-dimensional unit vector $\bs{e}_1 = (1,\ldots,0)$. Then
\begin{align}
    \P_{\calH_1}\pp{\left\langle \bs{e}_1,\bs{x}_j \right\rangle\geq t_{q,d},\left\langle \bs{e}_1,\bs{x}_{j'} \right\rangle\geq t_{q,d}} &= \P_{\calH_1}\pp{\left\langle \bs{e}_1,\bs{x}_j \right\rangle\geq t_{q,d}}\cdot\P_{\calH_1}\pp{\left\langle \bs{e}_1,\bs{x}_{j'} \right\rangle\geq t_{q,d}}\\
    & = p^2.
\end{align}
Combined with \eqref{eqn:covmiddle} we obtain that $ \s{cov}_{\calH_1}(\mathbf{A}_{i,j},\mathbf{A}_{i,j'})=0$. Therefore
\begin{align}
    \s{Var}_{\calH_1}(e(\calK)) = \binom{k}{2}p(1-p),
\end{align}
and by \eqref{eqn:ChebyshevMiddle}, we have
\begin{align}
  \pr_{\calH_1}\pp{\calA_{\s{dense}}(\s{G}_n) =0}&\leq \frac{\binom{k}{2}p(1-p)}{\binom{k}{2}^2\frac{(p-q)^2}{4}}\\
  & = \s{C}\frac{p(1-p)}{k^2(p-q)^2},\label{eqn:denseTypeII}
\end{align}
for some constant $\s{C}>0$. Combining \eqref{eqn:denseTypeI} and \eqref{eqn:denseTypeII}, we conclude that
$\s{R}(\calA_{\s{dense}})\to 0$ provided that
\begin{align}
\frac{k^2(p-q)^2}{p(1-p)}=\omega(1)
\quad\s{and}\quad
\frac{(p-q)^2}{q(1-q)}=\omega\left(\frac{\log n}{k}\right).
\end{align}
Since the latter condition implies the former, it follows that $\frac{(p-q)^2}{q(1-q)}=\omega\left(\frac{\log n}{k}\right)$ is the governing requirement.
\end{proof}

\begin{comment}
    \subsection{Proof of Lemma~\ref{lem:rev_TV_Chi}}

Recall that $P,Q$ are probability measures on a finite set $\calX$ with $P\ll Q$ and $Q_{\s{min}}\triangleq\min_{x\in\calX} Q(x)>0$. Then
\begin{align}
\chi^2(P\|Q)= \sum_{x\in\calX} \frac{(P(x)-Q(x))^2}{Q(x)}\leq \frac{1}{Q_{\min}}\sum_{x\in\calX} (P(x)-Q(x))^2,
\end{align}
so
\begin{align}
\sum_{x\in\calX} (P(x)-Q(x))^2 \geq Q_{\s{min}}\,\chi^2(P\|Q).
\end{align}
Using $\norm{v}_1\geq \norm{v}_2$ and $d_{\s{TV}}(P,Q)=\frac{1}{2}\sum_{x\in\calX} |P(x)-Q(x)|$,
\begin{align}
d_{\s{TV}}(P,Q)\geq \frac12 \sqrt{\sum_{x\in\calX} (P(x)-Q(x))^2}
\geq \frac{1}{2} \sqrt{Q_{\s{min}}\chi^2(P\|Q)}.
\end{align}
Hence
\begin{align}
d_{\s{TV}}(P,Q)\geq\frac{1}{2}\sqrt{Q_{\s{min}}\chi^2(P\|Q)}.
\end{align}
\end{comment}

\end{document}